\documentclass[11pt]{article}
\RequirePackage{algorithm,algorithmic,amssymb,epsf,epic,graphicx}

\def\denseformat{
\setlength{\textheight}{9.5in}
\setlength{\textwidth}{6.9in}
\setlength{\evensidemargin}{-0.3in}
\setlength{\oddsidemargin}{-0.3in}
\setlength{\headsep}{10pt}
\setlength{\topmargin}{-0.44in}
\setlength{\columnsep}{0.375in}
\setlength{\itemsep}{0pt}
}
\newtheorem{theorem}{Theorem}[section]

\newtheorem{definition}[theorem]{Definition}
\newtheorem{claim}[theorem]{Claim}
\newtheorem{lemma}[theorem]{Lemma}

\newtheorem{corollary}[theorem]{Corollary}

\newtheorem{comment}[theorem]{Comment}

\newtheorem{observation}[theorem]{Observation}

\def\boldhead#1:{\par\vskip 7pt\noindent{\bf #1:}\hskip 10pt}
\def\ithead#1:{\par\vskip 7pt\noindent{\it #1:}\hskip 10pt}

\def\inline#1:{\par\vskip 7pt\noindent{\bf #1:}\hskip 10pt}
\def\midinline#1:{\par\noindent{\bf #1:}\hskip 10pt}
\def\dnsinline#1:{\par\vskip -7pt\noindent{\bf #1:}\hskip 10pt}
\def\ddnsinline#1:{\newline{\bf #1:}\hskip 10pt}
\def\largeinline#1:{\par\vskip 7pt\noindent{\large\bf #1:}\hskip 10pt}
\long\def\comment #1\commentend{}
\long\def\commhide #1\commhideend{}
\long\def\commfull #1\commend{#1}
\long\def\commabs #1\commenda{}
\long\def\commtim #1\commendt{#1}
\long\def\commb #1\commbend{}
\long\def\commedit #1\commeditend{} % Editing comments, marked also by $>>>$ 

\long\def\commB #1\commBend{}       % Omit in 1996 (both TR and Siena)
\long\def\commex #1\commexend{}     % LN home exercise (hide solutions)

\long\def\commsiena #1\commsienaend{}  % omit in Siena, show in TR
                                         
\long\def\commBI #1\commBIend{}  % omit in Bar-Ilan
                                         
\long\def\CProof #1\CQED{}

\def\blackslug{\hbox{\hskip 1pt \vrule width 4pt height 8pt
    depth 1.5pt \hskip 1pt}}
\def\QED{\quad\blackslug\lower 8.5pt\null\par}
\def\inQED{\quad\quad\blackslug}

\long\def\PPP#1{\noindent{\bf Proof:}{ #1}{\quad\blackslug\lower 8.5pt\null}}

\long\def\denspar #1\densend
{#1}
\newif\ifnotesw\noteswtrue% T to show box & marginal notes; F supresses.
   {\ifnotesw\marginpar[\hfill\(\top\)]{\(\top\)}\fi}%
   {\ifnotesw\marginpar[\hfill\(\bot\)]{\(\bot\)}\fi}

\newcommand{\mnote}[1]%
    {\ifnotesw\marginpar%
        [{\scriptsize\it\begin{minipage}[t]{\marginparwidth}
        \raggedleft#1%
                        \end{minipage}}]%
        {\scriptsize\it\begin{minipage}[t]{\marginparwidth}
        \raggedright#1%
                        \end{minipage}}%
    \fi}

\def\cA{{\cal A}}
\def\cB{{\cal B}}
\def\cC{{\cal C}}

\def\cE{{\cal E}}
\def\cF{{\cal F}}
\def\cG{{\cal G}}
\def\cH{{\cal H}}
\def\cI{{\cal I}}
\def\cJ{{\cal J}}

\def\cL{{\cal L}}

\def\cS{{\cal S}}

\def\MathF{\hbox{\rm I\kern-2pt F}}
\def\MathP{\hbox{\rm I\kern-2pt P}}
\def\MathR{\hbox{\rm I\kern-2pt R}}
\def\MathZ{\hbox{\sf Z\kern-4pt Z}}
\def\MathN{\hbox{\rm I\kern-2pt I\kern-3.1pt N}}
\def\MathC{\hbox{\rm \kern0.7pt\raise0.8pt\hbox{\footnotesize I}
\kern-4.2pt C}}
\def\MathQ{\hbox{\rm I\kern-6pt Q}}

\newsavebox{\ttop}\newsavebox{\bbot}

\def\eps{\epsilon}

\def\nin{{~\not \in~}}

\denseformat
\newcommand {\ignore} [1] {}

\begin{document}

\title{\vspace{-0.02in} Optimal Euclidean spanners: really short, thin and lanky\thanks{An earlier version of this paper can be found in \cite{ES12}.}}
\author{
Michael Elkin \thanks{Department of Computer Science,
        Ben-Gurion University of the Negev, POB 653, Beer-Sheva 84105, Israel.
          E-mail: {\tt elkinm@cs.bgu.ac.il}.
          This author is supported
         by the BSF grant No.\ 2008430, by the ISF grant No.\ 87209011,
        and by the Lynn and William Frankel Center for Computer Sciences.}
\and
Shay Solomon \thanks{Department of Computer Science and Applied Mathematics, The Weizmann Institute of Science, Rehovot 76100, Israel.
E-mail: {\tt shay.solomon@weizmann.ac.il}.
Part of this work was done while this author was a graduate student in the Department of Computer Science, Ben-Gurion University of the Negev,
under the support of the Clore Fellowship grant No.\ 81265410, the BSF grant No.\ 2008430, and the ISF grant No.\ 87209011.}}

\date{\empty}

\begin{titlepage}
\def\thepage{}
\maketitle

\begin{abstract}
 The   degree, the (hop-)diameter, and the weight are the most basic and well-studied parameters of geometric spanners.
 In a seminal STOC'95 paper, titled ``Euclidean spanners: short, thin and lanky'', Arya et al.\ \cite{ADMSS95} devised a construction
 of Euclidean $(1+\eps)$-spanners that achieves constant degree,
 diameter $O(\log n)$, weight $O(\log^2 n) \cdot \omega(MST)$,
 and has running time $O(n \cdot \log n)$. This construction applies to $n$-point constant-dimensional Euclidean spaces.
 Moreover, Arya et al.\ conjectured that the weight bound can be improved by a logarithmic factor, without increasing the degree
 and the diameter of the spanner, and within the same running time.

 This conjecture of Arya et al.\ became one of the most central open problems in the area of Euclidean spanners.
 Nevertheless, the only progress since 1995 towards   its resolution was achieved in the lower bounds front:
 Any spanner with diameter $O(\log n)$ must incur weight $\Omega(\log n) \cdot \omega(MST)$, and this lower bound holds
 regardless of the stretch or the degree of the spanner \cite{DES08,AWY05}.
%  (see the SODA'05 paper of Agarwal et al.\ \cite{AWY05}, and the FOCS'08 paper of Dinitz et al.\ \cite{DES08}).

 In this paper we resolve the long-standing conjecture of Arya et al.\ in the affirmative.
 We present a spanner construction with the same stretch, degree, diameter, and running time,
 as in Arya et al.'s result, but with \emph{optimal weight}
 $O(\log n) \cdot \omega(MST)$.
 So our spanners are as thin and lanky as those of Arya et al., but they are \emph{really} short!

 Moreover, our result is  more general in three ways.
 First, we demonstrate that the conjecture holds true not only in constant-dimensional Euclidean spaces, but also in  \emph{doubling metrics}.
 Second, we provide a general tradeoff between the three involved parameters, which is \emph{tight in the entire range}.
 Third, we devise a transformation that decreases the lightness of spanners in \emph{general metrics}, while keeping all their other parameters
 in check. Our main result is obtained as a corollary of this transformation.
 %Specifically, we prove that for any $n$-point doubling metric, any $\eps > 0$, and any parameter $\rho \ge 2$,
 %there exists a $(1+\eps)$-spanner with degree $O(\rho)$, diameter $O(\log_\rho n + \alpha(\rho))$, and weight $O(\rho \cdot \log_\rho n) \cdot  \omega(MST)$.
% To prove this result, we develop a general procedure for transforming any given (possibly heavy-weight) spanner $H$
% into a \emph{light-weight spanner} $H'$ that has (essentially) the same number of edges, stretch, degree and diameter as those of $H$.
% Our procedure works for \emph{general} metrics, i.e., its applicability is not restricted to Euclidean or doubling metrics.
\end{abstract}
\end{titlepage}

\pagenumbering {arabic}

\section{Introduction}

{\bf 1.1 ~Euclidean Metrics.}
~~Consider a set $P$ of $n$ points in $\mathbb R^d$, $d \ge 2$, and a real number $t \ge 1$.
A graph $G = (P,E,\omega)$ in which the weight $\omega(p,q)$ of each edge $e = (p,q) \in E$
is equal to the Euclidean distance $\|p-q\|$ between $p$ and $q$ is called a \emph{Euclidean graph}.
We say that the Euclidean graph $G$ is a \emph{$t$-spanner} for $P$ if for
every pair $p,q \in P$ of distinct points, there exists a path $\Pi(p,q)$ in $G$
between $p$ and $q$ whose weight (i.e., the sum of all edge weights in it)
is at most $t \cdot \|p-q\|$. The parameter $t$ is called the \emph{stretch}
of the spanner. The path $\Pi(p,q)$ is said to be a \emph{$t$-spanner path}
between $p$ and $q$. In this paper we focus on the regime $t = 1+\eps$,
for $\eps > 0$ being an arbitrarily small constant. We will also concentrate
on spanners with $|E| = O(n)$ edges.
Euclidean spanners were introduced by Chew \cite{Chew86} in 1986.
The first constructions of $(1+\eps)$-spanners with $O(n)$ edges were devised
soon afterwards \cite{Clark87,Keil88}, and the running time
of such constructions was improved to $O(n \cdot \log n)$
a few years later \cite{Vai91,Sal91}.
% the first such constructions that require time $O(n \cdot \log n)$
%were given in.

Euclidean spanners turned out to be a fundamental geometric construct,
with numerous applications. In particular, they were found useful
in geometric approximation algorithms
\cite{RS98,GLNS02,GLNS08}, geometric distance oracles
\cite{GLNS02,GNS05,GLNS08} and network design \cite{HP00,MP00}.
Various properties of Euclidean spanners are
a subject of intensive ongoing research effort \cite{KG92,CDNS92,ADDJS93,DHN93,DN94,
ADMSS95,DNS95,RS98,GLN02,AWY05,CG06,DES08}.
See also the book by Narasimhan and Smid \cite{NS07}, and the references therein.
This book is titled ``Geometric Spanner Networks'', and it is devoted
almost exclusively to Euclidean spanners and their numerous applications.

In addition to stretch ($t = 1+\eps$) and sparsity ($|E| = O(n)$),
other fundamental properties of Euclidean spanners include their
\emph{(maximum) degree}, their \emph{(hop-)diameter}, and
their \emph{lightness}.
The \emph{degree} $\Delta(G)$ of a spanner $G$ is the
maximum degree of a vertex in $G$.
The \emph{diameter} $\Lambda(G)$ of a
$(1+\eps)$-spanner $G$
is the smallest number $\Lambda$ such that for every pair of points $p,q
\in P$ there exists a $(1+\eps)$-spanner path between $p$ and $q$ in $G$
that consists of at most $\Lambda$ edges (or \emph{hops}).
%in a $(1+\eps)$-spanner path in $G$ that connects a pair $p,q \in P$ of points.
The \emph{lightness} $\Psi(G)$ of a spanner $G$ is defined as the ratio
between the \emph{weight} $\omega(G) = \sum_{e \in E} \omega(e)$ of $G$ and
the weight $\omega(MST(P))$ of the minimum spanning tree $MST(P)$ for
the point set $P$.

%Arya and Smid \cite{AS94}
%devised a construction of $(1+\eps)$-spanners with constant
%degree [[S: but their spanners have constant lightness too; note that citing Keil and Gutwin \cite{KG92} is wrong here,
%as their ``$\Theta$-graph spanner'' has unbounded degree. BTW, Arya et al.\ \cite{ADMSS95} showed how to transform
%the $\Theta$-graph spanner to obtain constant degree, but I don't know if this fact is interesting enough to be mentioned here]].
%Das and Narasimhan \cite{DN94} devised a construction of
%$(1+\eps)$-spanners with $O(n)$ edges and constant lightness.
In this section we may write ``spanner'' as
a shortcut for a ``$(1+\eps)$-spanner with $O(n)$ edges''.
Feder and Nisan devised a construction of spanners with bounded degree (see \cite{AS94,Salowe92,Vai91}).
In FOCS'94, Arya et al.\ \cite{AMS94} devised a construction of spanners with
logarithmic diameter. The diameter was improved to $O(\alpha(n))$,
where $\alpha(n)$ is the inverse-Ackermann function, by Arya et al.\
\cite{ADMSS95} in STOC'95. 
%(See also \cite{CG06,NS07,Sol11}.)
(Further work on the tradeoff between the diameter and number of
edges in spanners can be found in \cite{CG06,NS07,Sol11}.)

Also, in the beginning of the nineties researchers started to systematically
investigate
spanners
that \emph{combine} several  parameters (among degree, diameter
and lightness).
Arya and Smid \cite{AS94} devised a construction of spanners with constant degree
and lightness. The running time of their construction is $O(n \cdot \log^d n)$, where $d = O(1)$ stands for the Euclidean dimension.
Other spanner constructions with constant degree and lightness, 
but with running time of $O(n \cdot \log n)$,
were subsequently devised in \cite{DN94,GLN02}. 
%was given by a construction of spanners with constant degree
%and lightness. The running time of their construction is $O(n \cdot \log^2 n)$.
%[[S: Let's mention here that the running time of their construction was $O(n \cdot \log^2 n)$,
%and 
%Another construction with constant degree and lightness, but with an improved running time of $O(n \cdot \log n)$
%was devised by 
%Finally, 
%Gudmundsson et al.\ \cite{GLN02}
%devised such a spanner construction within time $O(n \cdot \log n)$.
%this demonstrates that the issue of running is important by itself.]]
Arya et al.\ \cite{ADMSS95}  devised a construction of spanners with logarithmic
diameter and logarithmic lightness. (This combination was shown to be
optimal by Dinitz et al.\ \cite{DES08} in FOCS'08; see also
\cite{LSW94,AWY05} for previous lower bounds on this problem.)
This construction of \cite{ADMSS95}
may have, however, an arbitrarily large degree.
On the other hand, Arya et al.\ \cite{ADMSS95} also devised a construction
of spanners with constant degree, logarithmic diameter and lightness $O(\log^2 n)$.
In the end of their seminal work
Arya et al.\ \cite{ADMSS95} conjectured that one can obtain
a spanner with constant degree, logarithmic diameter and logarithmic lightness within time $O(n \cdot \log n)$.
Specifically, they wrote: {\vspace{0.05in}\\{\bf Conjecture 1 (\cite{ADMSS95})} \emph{~For any $t > 1$, and any
dimension $d$, there is a $t$-spanner, constructible in $O(n \cdot \log n)$ time,
with bounded degree, $O(\log n)$ diameter, and weight $O(\omega(MST) \cdot \log n)$.}

\vspace{0.05in}
In this paper\footnote{An earlier version of this paper can be found in \cite{ES12}.} we prove the conjecture of Arya et al.\ \cite{ADMSS95}, and devise
a construction of $(1+\eps)$-spanners with bounded degree, and with
logarithmic diameter and lightness. The running time of our construction is
$O(n \cdot \log n)$, matching the time bound conjectured in \cite{ADMSS95}. Moreover, this running time
is optimal in the algebraic computation-tree model \cite{CDS01}.
(We remark that regardless of the running time, prior to our work it was unknown whether $(1+\eps)$-spanners with constant degree, and logarithmic
diameter and lightness exist, even for 2-dimensional point sets.)

In fact, our result is far more general than this.
Specifically, we provide a tradeoff parameterized by a degree parameter $\rho \ge 2$, summarized below.
\begin{theorem} \label{tm1}
For any set of $n$ points in Euclidean space of any constant dimension $d$,
any $\eps > 0$ and any parameter $\rho \ge 2$, there exists a $(1+\eps)$-spanner with  $O(n)$ edges, degree $O(\rho)$,
diameter $O(\log_\rho n + \alpha(\rho))$ and lightness $O(\rho \cdot \log_\rho n)$.
The running time of our   construction is   $O(n \cdot \log n)$.
\end{theorem}
Due to lower bounds by \cite{CG06,DES08},
this tradeoff is \emph{optimal in the entire range} of the parameter $\rho$.
See Table \ref{tab2} for a concise summary of previous and our results for low-dimensional Euclidean metrics.
\begin{table*}
\begin{center}
%\resizebox {\textwidth }{!}{
\footnotesize
\begin{tabular}{|c|c|c|c|c|c|}
\hline  Reference & Degree &  Diameter & Lightness   \\
%\hline Arya and Smid \cite{AS94}  & $O(1)$ & unspecified  & $O(1)$ & $O(n \cdot \log^d n)$ \\
%\hline Das and Narasimhan \cite{DN94}  & $O(1)$ & unspecified & $O(1)$ & $O(n \cdot \log^2 n)$  \\
\hline Gudmundsson et al.\ \cite{GLN02} & $O(1)$ & unspecified & $O(1)$     \\
\hline Arya et al.\ \cite{ADMSS95} & unspecified & $O(\log n)$ & $O(\log n)$    \\
\hline Arya et al.\ \cite{ADMSS95}  & $O(1)$ & $O(\log n)$ & $O(\log^2 n)$    \\
\hline Arya et al.\ \cite{ADMSS95}  & unspecified & $O(\alpha(n))$ & unspecified   \\
%\hline Solomon and Elkin \cite{SE10} & $O(n^{1/\alpha(n)})$ &  $O(\alpha(n))$ & $O(n^{1/\alpha(n)} \cdot \alpha(n) \cdot \log n)$  & $O(n \cdot \log n)$ \\
\hline Solomon and Elkin \cite{SE10} & $O(\rho)$ &  $O(\log_\rho n + \alpha(\rho))$ & $O(\rho \cdot \log_\rho n \cdot \log n)$    \\
%\hline Chan et al.\ \cite{CGMZ05}  & $O(1)$ & unspecified & unspecified & unspecified & doubling \\
%\hline Gottlieb and Roditty \cite{GR082}  & $O(1)$ & unspecified & unspecified & $O(n \cdot \log n)$ & doubling \\
%\hline Chan and Gupta \cite{CG06} & unspecified & $O(\alpha(n))$ & unspecified & $O(n \cdot \log n)$ & doubling \\
%\hline Smid \cite{Smid09} & unspecified & unspecified & $O(\log n)$ & $O(n^2 \cdot \log n)$ & doubling \\
%\hline Gottlieb et al.\ \cite{GKK12} &$O(1)$ & $O(\log n)$ & unspecified & $O(n\cdot \log n)$ & doubling \\
\hline \hline {\bf New} &   {\boldmath $O(1)$} & {\boldmath $O(\log n)$} &  {\boldmath $O(\log n)$}   \\
\hline {\bf New} &   {\boldmath $O(\rho)$} & {\boldmath $O(\log_\rho n + \alpha(\rho))$} &  {\boldmath $O(\rho \cdot \log_\rho n)$}  \\
%\hline $\#$ edges  &  $O(n \log \log n)$ & $O(n \log^* n)$ &  $O(n \log^* n)$ & $O(n \log^* n)$  & $O(n \log^{***} n)$ & {\boldmath $O(n \log^{*} n)$} & %{\boldmath $O(n \log^{***} n)$} \\
%\hline Time & $O(n \log n)$ & $O(n (\log n) \log^* n)$ & $O(n \log n)$ & $O(n^2)$  & $O(n^2)$ & {\boldmath $O(n \log n)$} & {\boldmath $O(n \log n)$} \\
\hline
\end{tabular}
\end{center}
%\vspace{-0.1in}
\caption[]{ \label{tab2} \footnotesize A  comparison of previous and new
constructions of $(1+\epsilon)$-spanners with $O(n)$ edges for  low-dimensional Euclidean  metrics.
All these constructions have the same running time $O(n \cdot \log n)$.}
%[[S: Remarks: (1) Some rows in the table are subsumed in others (and can thus be omitted). For example, compare the second row with the fifth row.
%I marked with (**) rows that can be omitted.
%(2) Some unspecified values in the table are prone to be large (for example
%the construction of Smid \cite{Smid09} must have degree $\Omega(n)$), whereas other unspecified values
%are small (even though they were not analyzed in the paper; in particular, as far as I understand
%the constructions of
%Gottlieb and Roditty \cite{GR082} and
%Gottlieb et al.\ \cite{GKK12} have lightness $O(\log n)$ and $O(\log^2 n)$, respectively). Maybe we should
%fill in those values (even if it means that we give these results more credit than they deserve).
%On second thought, I'm not completely sure that these bounds are correct (especially the one of \cite{GKK12}), so perhaps we should just leave them unspecified.
%(3) We did not mention the result of Theorem 1.3 in the table (because it is was not stated in any paper before),
%but I believe we should.]]}
\end{table*}
\ignore{
\\In particular, by setting $\rho = 2$, we get the positive answer for Arya et al.'s
conjecture.
\\For $\rho = n^{1/\alpha(n)}$ we get $O(n)$ edges, diameter $O(\alpha(n))$,
degree $O(n^{1/\alpha(n)})$ and lightness $O(n^{1/\alpha(n)} \cdot \alpha(n))$.
In the corresponding result from \cite{ADMSS95},
which was achieved by a separate construction,
the spanner has the same number of edges and diameter, but its
degree and lightness may be arbitrarily large.}
%(Other variants of our construction
%give rise to $(1+\eps)$-spanners with  $O(n \cdot \alpha_k(n))$ edges and
%diameter $O(k)$, where $\alpha_k(n))$ is the $k$th inverse Ackermann function.
%[[S: what about the lightness and degree?]]
%This is optimal too; see \cite{ADMSS95,CG06,Sol11}.)
\vspace{0.12in}
$\\$
{\bf 1.2 ~Doubling Metrics.}
~~Our result extends in another direction as well.
Specifically, it applies to any \emph{doubling metric}.\footnote{
The \emph{doubling dimension} of a metric
is the smallest value $d$
such that every ball $B$ in the metric can be covered by at most
$2^{d}$ balls of half the radius of $B$.
This generalizes the Euclidean dimension, because the doubling dimension
of Euclidean space $\mathbb R^d$ is proportional to $d$.
 A metric
is called \emph{doubling}
if its doubling dimension is constant.}
Doubling metrics, implicit in the works of Assoud \cite{Ass83} and Clarkson \cite{Clark99}, were explicitly defined
by Gupta et al.\ \cite{GKL03}. They were subject of intensive research since then
\cite{KL04,Tal04,HPM06,CG06,ABN11,BGK12}.

Spanners for doubling metrics were also intensively studied \cite{GGN04,CGMZ05,HPM06,Rod07,GR081,GR082,Smid09}.
They were also found useful for approximation algorithms \cite{BGK12},
and for machine learning \cite{GKK12}. In SODA'05 Chan et al.\ \cite{CGMZ05} showed that
for any doubling metric there exists a spanner with constant
degree. In SODA'06, Chan and Gupta \cite{CG06} devised  a construction of
%$(1+\eps)$-spanners with $O(n)$ edges 
spanners with diameter $O(\alpha(n))$.
%(More generally, their construction exhibits an optimal tradeoff between the
%diameter and the number of edges.)
Smid \cite{Smid09} showed that in doubling metrics a greedy
construction produces spanners with logarithmic lightness. (The greedy spanner can be constructed within time $O(n^2 \cdot \log  n)$
in doubling metrics \cite{BCFMS10}.)
Gottlieb et al.\ \cite{GKK12} devised a construction of spanners
with constant degree and logarithmic diameter, within $O(n \cdot \log n)$ time.
To the best of our knowledge, prior to our work, there were no known constructions of
spanners for doubling metrics that
provide logarithmic diameter and lightness simultaneously (even allowing arbitrarily large degree).\footnote{On the other hand,
as was mentioned in Section 1.1, for Euclidean metrics such a construction was devised by
Arya et al.\ \cite{ADMSS95}. However, the degree in the latter construction is unbounded.}

We show that our construction extends to doubling metrics without incurring any overhead (beyond constants)
in the degree, diameter, lightness, and running time. In other words, Theorem \ref{tm1} applies to doubling metrics.
See Table \ref{tab1} for a   summary of previous and our results for doubling metrics.
\begin{table*}
\begin{center}
%\resizebox {\textwidth }{!}{
\footnotesize
\begin{tabular}{|c|c|c|c|c|c|}
\hline  Reference & Degree &  Diameter & Lightness & Running Time  \\
\hline Chan et al.\ \cite{CGMZ05}  & $O(1)$ & unspecified & unspecified & unspecified  \\
\hline Gottlieb and Roditty \cite{GR082}  & $O(1)$ & unspecified & unspecified & $O(n \cdot \log n)$ \\
\hline Chan and Gupta \cite{CG06} & unspecified & $O(\alpha(n))$ & unspecified & $O(n \cdot \log n)$  \\
\hline Smid \cite{Smid09} & unspecified & unspecified & $O(\log n)$ & $O(n^2 \cdot \log n)$  \\
\hline Gottlieb et al.\ \cite{GKK12} &$O(1)$ & $O(\log n)$ & unspecified & $O(n\cdot \log n)$   \\
\hline \hline {\bf New} &   {\boldmath $O(1)$} & {\boldmath $O(\log n)$} &  {\boldmath $O(\log n)$} &  {\boldmath $O(n \cdot \log n)$}  \\
\hline {\bf New} &   {\boldmath $O(\rho)$} & {\boldmath $O(\log_\rho n + \alpha(\rho))$} &  {\boldmath $O(\rho \cdot \log_\rho n)$} &  {\boldmath $O(n \cdot \log n)$}  \\
%\hline $\#$ edges  &  $O(n \log \log n)$ & $O(n \log^* n)$ &  $O(n \log^* n)$ & $O(n \log^* n)$  & $O(n \log^{***} n)$ & {\boldmath $O(n \log^{*} n)$} & %{\boldmath $O(n \log^{***} n)$} \\
%\hline Time & $O(n \log n)$ & $O(n (\log n) \log^* n)$ & $O(n \log n)$ & $O(n^2)$  & $O(n^2)$ & {\boldmath $O(n \log n)$} & {\boldmath $O(n \log n)$} \\
\hline
\end{tabular}
\end{center}
%\vspace{-0.1in}
\caption[]{ \label{tab1} \footnotesize A  comparison of previous and new
constructions of $(1+\epsilon)$-spanners with $O(n)$ edges for   doubling metrics.}
%[[S: Remarks: (1) Some rows in the table are subsumed in others (and can thus be omitted). For example, compare the second row with the fifth row.
%I marked with (**) rows that can be omitted.
%(2) Some unspecified values in the table are prone to be large (for example
%the construction of Smid \cite{Smid09} must have degree $\Omega(n)$), whereas other unspecified values
%are small (even though they were not analyzed in the paper; in particular, as far as I understand
%the constructions of
%Gottlieb and Roditty \cite{GR082} and
%Gottlieb et al.\ \cite{GKK12} have lightness $O(\log n)$ and $O(\log^2 n)$, respectively). Maybe we should
%fill in those values (even if it means that we give these results more credit than they deserve).
%On second thought, I'm not completely sure that these bounds are correct (especially the one of \cite{GKK12}), so perhaps we should just leave them unspecified.
%(3) We did not mention the result of Theorem 1.3 in the table (because it is was not stated in any paper before),
%but I believe we should.]]}
\end{table*}
%In this paper we devise a construction of $(1+\eps)$-spanners for doubling
%metrics with constant degree, logarithmic diameter and logarithmic lightness.
%Furthermore, for any parameter $\rho \ge 2$, we devise a construction
%of $(1+\eps)$-spanners with degree $O(\rho)$, diameter $O(\log_\rho n)$
%and lightness $O(\rho \cdot \log_\rho n)$.
%The running time of our construction is $O(n \cdot \log n)$.
%This matches exactly our results
%for Euclidean metrics.
%Also, as was discussed in Section 1.1, this tradeoff is
%optimal in the entire range of the parameter $\rho$.
\vspace{0.12in}
$\\$
{\bf 1.3 ~Our and Previous Techniques.}
~~Our starting point is the paper of Chandra et al.\ \cite{CDNS92} from SoCG'92 (see also \cite{CDNS95}).
In this paper  the authors devised a general
transformation: given a  construction of spanners with certain stretch
and number of edges their transformation returns a construction with
roughly the same stretch and number of edges, but with a logarithmic
lightness. The drawback of their transformation is that it blows up the
degree and the diameter of the original spanner.

In this paper we devise a much more refined transformation.
Our transformation    enjoys  all the useful properties of the transformation
of \cite{CDNS92}, but, in addition, it preserves (up to constant factors) the degree
and the diameter of the original construction. We then compose our refined
transformation on top of known constructions of spanners with constant
degree and logarithmic diameter (due to Arya et al.\ \cite{ADMSS95} in
the Euclidean case, and due to Gottlieb et al.\ \cite{GKK12} in the case
of doubling metrics).
As a result we obtain a construction of spanners with constant degree,
logarithmic diameter and logarithmic lightness.
The latter proves the conjecture of Arya et al.\ \cite{ADMSS95}.

We remark that our transformation can be applied not only for  Euclidean
or doubling metrics, but rather in much more general scenarios.
In fact, in \cite{ES13} we have already  obtained some improved results for spanners in general
graphs that are based on a variant of  this transformation.
%[[S: Let's write here something like:]]
%Furthermore, we have discovered recently (as a part of ongoing research) a number of additional applications and features of
%[[S: the following sentences are not going to the arXiv, I need to keep it for now]] 
%We have also obtained several new results as additional corollaries of our transformation,
%Several additional applications of our transformation include 
%including optimal constructions of spanners for metrics induced by graphs of bounded tree-width or bounded tree-length,
%and general constructions (i.e., constructions that provide a general tradeoff between the degree, diameter and lightness) of fault-tolerant spanners for %doubling metrics. 
%[[S: issue with distance labeling and routing]], and [[S: issue with FRT]]. 
%These new results are part of ongoing research,
%and are outside the scope of the current paper.

Next, we provide a schematic overview of
the two transformations (the one due to  \cite{CDNS92},
and our refined one). The transformation of \cite{CDNS92} starts with constructing
an MST $T$ of the input metric. Then it constructs the preorder traversal
path $\cL$ of $T$. The path $\cL$ is then partitioned into  $c \cdot n$
intervals of length $\frac{|\cL|}{c \cdot n}$ each, for a constant $c > 1$.
This is the bottom-most level $\cF_1$ of the hierarchy $\cF$ of intervals that the
transformation constructs.
Pairs of consecutive intervals are grouped together; this gives rise to
$c \cdot n/2$   intervals of length $2 \cdot \frac{|\cL|}{c \cdot n}$  each.
The hierarchy $\cF$ consists of $\ell = \log n$ levels, with $c$
 intervals of length $\frac{|\cL|}{c}$ each in the last level $\cF_\ell$.

In each level $j \in [\ell]$ of the hierarchy each non-empty interval
is represented by a point of the original metric (henceforth, its \emph{representative}). Let $Q_j$ denote
the set of $j$-level representatives. The transformation then invokes
its input black-box construction of spanners on each point set $Q_j$ separately.
Each of those $\ell$ auxiliary spanners is then pruned, i.e., ``long'' edges are removed
from it. The remaining edges in all the auxiliary spanners, together with the MST $T$, form the output spanner.

Intuitively, the pruning step ensures that the resulting spanner is reasonably light.
The stretch remains roughly intact, because each distance is taken care ``on its
own scale''. The number of edges does not grow by much, because the sequence
$|Q_1|,|Q_2|,\ldots,|Q_\ell|$ decays geometrically. However, the diameter is blown
up, because within each interval the MST-paths (which may contain   many edges) are used to reach points that
do not serve as representatives. Also, the degree is blown up because the same
point may serve as a representative in many different levels.

%Our first idea is to
To fix the problem with the diameter we
 use a construction of 1-dimensional spanners to shortcut
the traversal path $\cL$. We remark that $(1+\eps)$-spanners with $O(n)$
edges, constant degree, logarithmic diameter and logarithmic lightness
for sets of $n$ points on a line (1-dimensional case) were devised already in 1995
by Arya et al.\ \cite{ADMSS95}.
Plugging\footnote{In fact, we use our own more recent construction
\cite{SE10} of 1-spanners for 1-dimensional spaces with the above properties.
Having stretch 1 instead of $(1+\eps)$ simplifies the analysis.}
this 1-dimensional spanner construction  into the transformation of Chandra et al.\ \cite{CDNS92} gives
rise to an improved transformation that keeps the diameter in check, but still
blows up the degree.

%Our second idea is to
To fix the problem with the degree, it is natural to try
distributing the degree load evenly between ``nearby'' points along $\cL$.
Alas, if one sticks with the original
hierarchy  $\cF$ of partitions of $\cL$ into intervals,
this turns out to be impossible. 
The problem is that the same point
may well be the only eligible representative for many   levels of the hierarchy.
Overcoming this hurdle is the heart of our paper.
Instead of intervals we divide the point set into a different hierarchy $\hat \cF$ of sets,
which we call \emph{bags}. On the lowest level of the hierarchy the bags and the intervals
coincide. %This, however, soon changes.
As our algorithm proceeds it carefully
moves points between bags so as to guarantee that no point will ever be overloaded.
At the same time we never put points that are far away from one another in the original metric into the same bag.
%try not to move points too far from their initial position on $\cL$.
%At the same time we move a new point into some bag, only if it is quite close to some
%``original'' point of that bag.
Indeed, if remote points
%(with respect to their initial position on $\cL$)
end up in the same bag,
then  the auxiliary spanners for the sets of representatives, as well as the 1-dimensional spanner for $\cL$, cease  providing
short
%(in terms of both weighted distance and hop-distance)
$(1+\eps)$-spanner paths for the
original point set.
On the other hand, degree constraints may force our algorithm to relocate points arbitrarily far
away from their initial position on $\cL$.
Our construction balances carefully 
between these two contradictory requirements.  % is the main technical contribution of this paper.
\vspace{-0.02in}
\\
{\bf 1.4 ~Related Work.}
~~Most of the related work was already discussed above.
One more relevant result is the ESA'10 paper \cite{SE10} by the authors of the current paper.
There we devised a construction of spanners that trades gracefully between
the degree, diameter and lightness. That construction, however, could only
match the previous  suboptimal bounds of Arya et al.\ \cite{ADMSS95}, but not improve them.
In particular, the lightness of the construction of \cite{SE10} is $\Omega(\log^2 n)$,
regardless of the other parameters.
\vspace{0.12in}
\\
{\bf 1.5 ~Consequent Work.}
~~A preliminary version of this paper started to circulate in April 2012 \cite{ES12}.
It sparked a number of follow-up papers. First, in \cite{ES13} we used the technique developed in this paper
to devise an efficient construction of light spanners for general graphs.
Second, in \cite{CLN12} Chan et al.\ came up with an alternative construction of  spanners for doubling
metrics with constant degree, and logarithmic diameter and lightness. Their construction
and analysis are arguably simpler than ours.
In addition, they   extended this result to the fault-tolerant setting. 
A yet alternative construction of fault-tolerant spanners with the same properties and with running time $O(n \cdot \log n)$
was devised recently by Solomon \cite{Sol12}.
%On the other hand, the running time in \cite{CLN12}
%is not analyzed.  (To our understanding it requires at least $O(n^2)$ time.)
However, while our construction provides an optimal tradeoff
between the diameter and lightness ($O(\log_\rho n + \alpha(\rho))$ versus $O(\rho \cdot \log_\rho n)$ for the entire range of the parameter
$\rho \ge 2$), the constructions of \cite{CLN12,Sol12} do the job only for $\rho = O(1)$.
As far as we know they cannot be extended to provide the general tradeoff.
Finally, the constructions of \cite{CLN12,Sol12} do not provide a transformation for converting
spanners into light spanners in general metrics. 
%In another development, 
% devised  for the case $\rho = O(1)$, which requires $O(n \cdot \log n)$ time.
%However, the construction of \cite{Sol12} too does not provide a general tradeoff between the involved parameters.
%Also, 
%it does not provide a general transformation for converting
%spanners into light spanners.

Finally, we stress that both constructions \cite{CLN12,Sol12} are consequent to our work.
These constructions build upon ideas
and techniques that we present in the current paper.
\vspace{0.12in}
\\
{\bf 1.6 ~Structure of the Paper.}
%[[S: needs rewriting, we can leave this to the end...]]
%~~In Section 2 we describe our algorithm. We analyze it in Section 3.
~~In Section  \ref{section2} we describe our construction (Algorithm $LightSp$). The description of the algorithm is provided in Sections  \ref{sec21}-\ref{s:jlevel}. 
A detailed outline of Section \ref{section2} appears in the paragraph preceding Section \ref{sec21}.
We analyze the properties of the spanners produced by our algorithm in Section \ref{section3}. 
The most elaborate and technically involved parts of the analysis concern 
the stretch and diameter (Section \ref{stchdiam}) and the degree (Section \ref{deg:app}) of the produced spanners.
%Appendix C contains the description of two subroutines of our algorithm.
%Appendix D is devoted to a few parts of the analysis that we could not provide in the main part of
%the paper due to lack of space.
%The bibliography appears at the very end of the submission.
%ructure of the Paper.}~~[[S: fill in]]
\vspace{0.12in}
\\
{\bf 1.7 ~Preliminaries.}
%~~We will use the following two results as a black-box. [[S: remove this sentence?]]
~~The following theorem provides optimal spanners for \emph{1-dimensional} Euclidean metrics with respect to all three parameters
(degree, diameter and lightness).
%[[S: rewrite; I suggest to move this stuff to Section \ref{secbasic}]]
%To prove Theorem \ref{ourresult} we employ a construction of 1-spanners for 1-dimensional spaces from \cite{SE10}.
%[[S: I've added reference also to the technical report, because there we prove running time bounds;
%perhaps this is redundant?]]
%
%, Theorem 2.1 in \cite{SE11}
\begin{theorem} [\cite{ADMSS95,SE10}] \label{1span}
%\begin{enumerate}
%\item
For any $n$-point 1-dimensional space $M$
and any   $\rho \ge 2$,
there
exists a 1-spanner $H$ with $|H| = O(n)$, $\Delta(H) = O(\rho)$,
$\Lambda(H) = O(\log_\rho n + \alpha(\rho))$ and $\Psi(H) = O(\rho \cdot \log_\rho n)$.
The running time of this construction is $O(n)$.
\end{theorem}
%\item
The following theorem provides spanners for doubling metrics with an optimal tradeoff between
the degree and diameter. Note, however, that this tradeoff does not involve lightness.
%[[S: OK, this is vague enough. (But note that the actual tradeoff does involve lightness,
%we just do not include it in the statement.)]]
\begin{theorem} [\cite{ADMSS95,GR082,SE10}] \label{twoprop}
For any $n$-point doubling metric $M = (P,\delta)$,
any $\eps > 0$ and any   $\rho \ge 2$,
there exists a $(1+\eps)$-spanner $H$ with
$|H| = O(n)$, $\Delta(H) = O(\rho)$ and $\Lambda(H) = O(\log_\rho n + \alpha(\rho))$.
%(Note that there is no guarantee on the lightness.) 
%[[S: that's inaccurate, the guarantee
%on the lightness is optimal up to a factor of $\log n$; I suggest to be more vague (to remove the sentence in the parenthesis).]]
The running time of this construction is $O(n \cdot \log n)$.
%\end{theorem}
%\end{enumerate}
\end{theorem}
For the sake of completeness we provide a proof of Theorem \ref{twoprop}
in Appendix \ref{appB}.

%Let $\rho \ge 2$ be an integer parameter, and $\eps > 0$.
%[\cite{ADMSS95,GR082,SE10}]

Our transformation theorem is formulated below.

\begin{theorem} \label{ourresult}
Let $M = (P,\delta)$ be an arbitrary metric. Let $t \ge 1,\rho \ge 2$ be arbitrary
 parameters. Suppose that
for any subset $Q \subseteq P$, $|Q| = n$,
there exists an algorithm (henceforth, Algorithm $BasicSp$) which builds a $t$-spanner $H$
for the sub-metric $M[Q]$ of $M$ induced by the point set $Q$, so that $|H| \le SpSz(n)$, $\Delta(H) \le \Delta(n)$,
$\Lambda(H) \le \Lambda(n)$. Moreover, Algorithm $BasicSp$ requires at most
$SpTm(n)$ time. Suppose also that all the functions $SpSz(n),\Delta(n),\Lambda(n)$ and $SpTm(n)$
are monotone non-decreasing, while the functions $SpSz(n)$ and $SpTm(n)$
are also convex and vanish at zero.

Then there is an algorithm (henceforth, Algorithm $LightSp$) which builds, for every subset $Q \subseteq P$,
$|Q| = n$, and any $\eps > 0$, a $(t+\eps)$-spanner $H'$ for $M[Q]$
with $|H'| = O(SpSz(n) \cdot \log_{\rho}(t/\eps))$,
$\Delta(H') = O(\Delta(n) \cdot \log_{\rho}(t/\eps) + \rho)$,
$\Lambda(H') = O(\Lambda(n) + \log_{\rho}n + \alpha(\rho))$,
$\Psi(H') = O(\frac{SpSz(n)}{n} \cdot \rho \cdot \log_{\rho}n \cdot (t^3/\eps))$.
The running time of Algorithm $LightSp$ is
$O(SpTm(n) \cdot \log_{\rho}(t/\eps) + n \cdot \log n)$.
\end{theorem}
\ignore{
{\bf Remark:} If one is also given an algorithm (henceforth, Algorithm $LightTree$)
which builds a spanning tree $T$ for $M[Q]$ with lightness $O(1)$
in $TrTm(n)$ time, then Algorithm $LightSp$ performs a little bit better.
Specifically, in this case it constructs a spanner $H'$ with the same guarantees
on the stretch, maximum degree and diameter as the ones listed above.
However, its lightness becomes smaller by a factor of $t$, i.e.,
$\Psi(H') = O(\frac{SpSz(n)}{n} \cdot \rho \cdot \log_{\rho}n \cdot (t^2/\eps))$.
In addition, the running time of the algorithm becomes
$O(SpTm(n) \cdot \log_{\rho}(t/\eps)  + TrTm(n))$.
In low-dimensional Euclidean and doubling metrics,
it is well known that one can construct a spanning tree with lightness $O(1)$ within time $TrTm(n) = O(n \cdot \log n)$.
(This can be done, for example, by computing an MST using Prim's Algorithm over any $O(1)$-spanner of the metric; see
Section \ref{time:app} for more detail.)
Hence the running time in these cases is the same, up to constant factors,
as the one cited in Theorem \ref{ourresult}. [[S: remove remark?]]
}

Given this theorem we derive our main result by instantiating the algorithm from
Theorem \ref{twoprop} as Algorithm $BasicSp$ in Theorem \ref{ourresult}. As a result we obtain a
construction of $(1+\eps)$-spanners $H$ for doubling metrics with $|H| = O(n)$,
$\Delta(H) = O(\rho)$, $\Lambda(H) = O(\log_\rho n + \alpha(\rho))$,
$\Psi(H) = O(\rho \cdot \log_{\rho} n)$, in time $O(n \cdot \log n)$.
(We substituted $t = 1+\eps$, and $\eps > 0$ is a   constant.)
%See
In Appendix \ref{appA} we explicate the dependencies on $\eps$
and the doubling dimension $d$ on various parameters of the spanner constructed by Theorem \ref{ourresult}.
%for the more general statement of our result, which applies
%to general (not necessarily constant) $\eps$

For a pair of non-negative integers $i,j,i \le j$, we denote $[i,j] = \{i,i+1,\ldots,j\}, [i] = \{1,2,\ldots,i\}$.

For paths $\Pi,\Pi'$ connecting vertices $v$ and $u$ and $u$ and $w$, respectively,
denote by $\Pi \circ \Pi'$ the concatenation of these paths.

%{k, k + 1, . . . , n} and {1, 2, . . . , n}
%by [k, n] and [n], respectively

%\vspace{-0.05in}
\section{Algorithm $LightSp$} \label{section2}
%\vspace{-0.05in}
%\subsection{Getting Started}
Let $M = (P,\delta)$ be an arbitrary metric, and let
$Q \subseteq P$ be an arbitrary subset of $n$ points from $P$.

Algorithm $LightSp$ starts with computing an MST, or an approximate
MST, $T$, for the metric $M[Q]$.
In low-dimensional Euclidean and doubling metrics an $O(1)$-approximate MST can be computed within $O(n \cdot \log n)$ time.
In general, a 
%In low-dimensional Euclidean and doubling metrics
$t$-approximate MST can be computed within time $O(SpTm(n) + n \cdot \log n)$ 
by running Prim's MST
Algorithm over the $t$-spanner produced
%any $O(1)$-spanner of the metric with $O(n)$ edges, such as the spanner computed 
by Algorithm $BasicSp$.  % Theorem \ref{twoprop}.
%The case of more general metrics is addressed in Section \ref{time:app}. [[S: keep this sentence?]]

Let $\cL$ be the Hamiltonian path of $M[Q]$ obtained by taking the preorder traversal of $T$.
Define $L = \omega(\cL)$; it is well known (\cite{CLRS90}, ch.\ 36) that $L \le 2 \cdot \omega(T)$, and so $L = O(t \cdot \omega(MST(M[Q])))$.
Write $\cL = (q_1,q_2,\ldots,q_n)$, and let $M_\cL = (Q,\delta_\cL)$ be the 1-dimensional space induced by the path $\cL$,
where $\delta_\cL$ is the distance in $\cL$ (henceforth, \emph{path distance}), i.e.,
$\delta_\cL(v_k,v_{k'}) = \sum_{i=k}^{k'-1} \delta(v_i,v_{i+1})$,
for every pair $k,k'$ of indices, $1 \le k < k' \le n$.
%Notice  that $MST(M_\cL) = \cL$, and so $\omega(MST(M_\cL)) = L$.
We employ Theorem \ref{1span} to build in $O(n)$ time
a 1-spanner $H_\cL$ for $M_\cL$ with $|H_\cL| = O(n)$, $\Delta(H_\cL) = O(\rho)$, $\Lambda(H_\cL) = O(\log_\rho n + \alpha(\rho))$
and $\Psi(H_\cL) = O(\rho \cdot \log_\rho n)$.
%For any pair $p,q \in Q$ of points, there is a 1-spanner path $\Pi_{H_\cL}(p,q)$ in $H_\cL$
%that
%consists of at most $O(\log_\rho n + \alpha(\rho))$ edges; by definition, the weight $\omega(\Pi_{H_\cL}(p,q))$ of this path is equal to
%the path distance $\delta_\cL(p,q)$ between $p$ and $q$.   % = \sum_{\ell=i}^{i'} \delta(v_\ell,v_{\ell+1})$.
%Notice that  edge weights in $H_\cL$ may be greater but cannot be smaller than the corresponding distances
%in the original metric $M$.
Let $H = (Q,E_{H})$  be the graph obtained from $H_\cL$ by assigning weight $\delta(p,q)$
to each edge $(p,q) \in H_\cL$. Since edge weights in $H$ are no greater than the corresponding edge weights in $H_\cL$,
we have (i) $\omega(H) \le \omega(H_\cL) = O(\rho \cdot \log_\rho n) \cdot L$, and (ii) for any pair $p,q \in Q$ of points,
there is a path $\Pi_{H}(p,q)$ in $H$ that has weight at most
$\delta_\cL(p,q)$ and $O(\log_\rho n + \alpha(\rho))$ edges. We henceforth call $H$ the \emph{path-spanner}.
We also define an order relation $\prec_\cL$ on the point set $Q$.
Specifically, we write $q_i \prec_\cL q_j$ (respectively, $q_i \preceq_\cL q_j$)
iff $i < j$ (resp., $i \le j$).

Let $\ell = \lceil \log_\rho n \rceil$.
Define $Q_0 = Q$, let $n_0 = |Q_0| = n$,
and define the \emph{0-level threshold} $\tau_0 = 2 \cdot \frac{L}{n} \cdot t \cdot (1+\frac{1}{c})$,
where $c = \lceil \frac{4\cdot(t+1)}{\eps} \rceil = \Theta(t/\eps)$ is a constant ($t$ and $\eps$
will be set as constants).
%where $0 < hi \le 1$ is an arbitrarily small constant [[S:?]].
For $j \in [\ell]$,
%we define $Q_j \subset Q$ in the following way.
we define $\xi_j = \rho^{j-1} \cdot \frac{L}{n}$. Divide the path $\cL$ into $n_j = \lceil \frac{c \cdot L}{\xi_j} \rceil = \lceil \frac{c\cdot n}{\rho^{j-1}} \rceil$
intervals of length $\mu_j = \frac{\xi_j}{c}$ each (except for maybe one interval of possibly shorter length). 
%(We can round $n_j$ to an integer. This change will have negligible effect on the analysis.)
From now on we assume that each $n_j$ is equal to $\frac{c \cdot L}{\xi_j} = \frac{c\cdot n}{\rho^{j-1}}$, because non-integrality of this expression has no effect whatsoever on the analysis. 
%From now on we assume that each $n_j$ is rounded to $\frac{c \cdot L}{\xi_j} = \frac{c\cdot n}{\rho^{j-1}}$, because non-integrality 
%of this expression has no effect whatsoever on the analysis.
Define also the \emph{$j$-level threshold} $\tau_j = 2\mu_j \cdot \rho \cdot t \cdot (c+1) = 2 \cdot \frac{L}{n} \cdot t \cdot (1+\frac{1}{c}) \cdot \rho^j$.
These intervals induce a partition of the
point set $Q$ in the obvious way; denote these intervals
and the corresponding
point sets by
$I^{(1)}_j,I^{(2)}_j,\ldots,I^{(n_j)}_{j}$
and $Q^{(1)}_j,Q^{(2)}_j,\ldots,Q^{(n_j)}_{j}$, respectively.
% Define $Q_j$ to be a set containing exactly one point $r^{(i)}_j$
%arbitrarily chosen from each non-empty point set $Q^{(i)}_j$ of the partition, $i \in [n_j]$;
%we call the point $r^{(i)}_j$
%the \emph{representative} of the point set $Q^{(i)}_j$.
%The points $r^{(i)}_j$, $i \in [n_j]$, will also be referred to as the \emph{$j$-level representatives}.
%Finally, we will refer to the intervals $I^{(1)}_1,\ldots,I^{(n_1)}_{1},I^{(1)}_2,\ldots,I^{(n_2)}_{2},\ldots,I^{(1)}_\ell,\ldots,I^{(n_\ell)}_{\ell}$
%as the \emph{hierarchy $\cI$ of intervals}.
%\\Note that
%\begin{equation} \label{qjnj}
%|Q_j| ~\le~ \min\{n,n_j\} ~=~  \min\left\{n,\frac{c\cdot n}{\rho^{j-1}}\right\}.
%\end{equation}     % = O(\frac{n}{k^{j-1}})$.
%Also, observe that the path distance $\delta_\cL(p,q)$ between every pair $p,q$ of points
%in the same point set of the partition is bounded above by the interval length $\frac{\xi_j}{c}$;
%consequently, the path $\Pi_{H}(p,q)$ that is guaranteed by $H$ has weight at most $\delta_\cL(p,q) \le \frac{\xi_j}{c}$ and $O(\log_\rho n + \alpha(\rho))$ %edges.

We define ${\cal I}_j = \{I_j^{(1)},\ldots,I_j^{(n_j)}\}$, and ${\cal I} = \bigcup_{j=1}^\ell {\cal I}_j$.
Note that, for each $j \in [2,\ell]$, every $j$-level interval $I$ is a union of $\rho$ consecutive $(j-1)$-level
intervals. (Similarly to above, we may assume that $\rho$ is an integer.) 
The interval $I$ is called the \emph{parent}
of these $(j-1)$-level intervals, and they are called its \emph{children}.
This nested hierarchy of intervals defines in a natural way a forest $\cal F$
of $\rho$-ary trees, whose vertices (henceforth, \emph{bags}) correspond to intervals from $\cal I$.
With a slight abuse of notation we denote by $\cI$ also the set of bags in $\cF$, and by $\cF_j = \cI_j$ the set of $j$-level bags in $\cF$,
for each $j \in [\ell]$.
Each of the trees in $\cF$ is rooted at an $\ell$-level interval.
Thus, the number of trees in $\cal F$ is equal to the number $|\cF_\ell| = |{\cal I}_\ell| = n_\ell$
of $\ell$-level intervals. Specifically, $n_\ell = \frac{c \cdot n}{\rho^{\ell -1}}$,
and so  %=\frac{c \cdot n}{\rho^{\log_\rho n -1}} = c \cdot \rho$.
$c < n_\ell \le c \cdot \rho$.
Denote the interval that corresponds to a bag $v$ of $\cal F$ by $I(v)$,
and denote the point set of $I(v)$ by $N(v)$. We call
%$I(v)$ is the \emph{native}
%interval of the bag $v$. Also,
the point set $N(v)$ the \emph{native point set} of $v$.
For an inner bag $v$ in $\cal F$
with $\rho$ children $c_1(v),\ldots,c_{\rho}(v)$, we have $I(v) = \bigcup_{i=1}^\rho
I(c_i(v))$, and $N(v) = \bigcup_{i=1}^\rho N(c_i(v))$.
%Denote by ${\cal F}_j = {\cI}_j$ [[S: this equality is somewhat inaccurate]] the set of $j$-level bags in $\cal F$, for each   $j \in [\ell]$;
%For each   $j \in [\ell]$, we have
Note that
$\bigcup_{v \in {\cal F}_j} I(v) = [q_1,q_n]$, and $\bigcup_{v \in {\cal F}_j} N(v) = Q$.
Also, for any pair of distinct  bags  $u,v \in \cF_j$, $I(u) \cap I(v) = N(u) \cap N(v) = \emptyset$.
%[[S: yes, we're using this definition in this section, see below]]

In Algorithm $LightSp$ we (implicitly) maintain
another forest $\hat \cF$ over the same bag set $\cal I$.
Specifically, a $j$-level
bag $v$, for some index $j \in [\ell-1]$, may become a child of
some $(j+1)$-level bag $u$, other than the parent $\pi(v)$ of $v$ in $\cal F$.
If this happens we say that $u$ becomes a \emph{step-parent} of $v$ in $\cal F$ (and $u$ is a parent of $v$ in $\hat \cF$),
and $v$ becomes a \emph{step-child} of $u$ in $\cal F$ (and $v$ is a child of $u$ in $\hat \cF$).
As a result the points associated with the bag $v$ become associated with $u$. We will soon provide more details on this.

Observe that in $\cF$ each bag $v$ corresponds to a specific interval $I(v) \in \cI$, and contains only points that lie
within this interval (i.e., the points of $N(v)$). On the other hand, each bag $v \in \hat \cF$ may contain points from many different intervals of $\cI$.
We denote by $\hat \cF_j$ the set of $j$-level bags of $\hat \cF$. If a bag $v \in \cF_j$ becomes a child
in $\hat \cF$ of a bag $u$, then it will hold that $u \in \cF_{j+1}$.
%it does not change its level [[S: ?]]. 
This guarantees that $\cF_j = \hat \cF_j$, for every $j \in [\ell]$. 
%There is, however, a natural 1-1 correspondence between bags in $\cF$ and in $\hat \cF$. Moreover, for every index $j, j \in [\ell]$,
%this correspondence maps $\cF_j$ to $\hat \cF_j$. Hence a bag $v$ is in level $j$ in $\cF$ iff it is in level $j$ in $\hat \cF$.
%[[S: not very clear]]

The rest of this section is organized as follows. In Section \ref{sec21} we describe point sets which are associated with bags of $\hat \cF$.
In Section \ref{disregard} we describe an important subset of edges of the ultimate spanner that the algorithm constructs. This subset
is called the \emph{base edge set}. During the execution of the algorithm some bags are labeled as \emph{zombies} or \emph{incubators}.
These notions are discussed in Section \ref{thealg}. In Section \ref{sec24} we describe how our algorithm selects representatives of different bags.
Section \ref{app:attachalg} is devoted to Procedure $Attach$, which is a subroutine of our algorithm. The algorithm itself is described in Section
\ref{s:jlevel}

%\vspace{-0.07in}
\subsection{Point Sets} \label{sec21}
%\vspace{-0.05in}
%In this section we describe different sets that Algorithm $LightSp$ constructs,
%and explain how it selects the representatives for bags $v \in \cF$.
In addition to the native point set $N(v)$,
% and native interval $I(v)$ of a bag $v$,
the algorithm will also maintain for each bag $v$
%Recall that with every bag $v \in \cF$ we associate the native interval $I(v)$,
%and the native point set $N(v)$. In fact, we will also associate with $v$
three more point sets: the \emph{base point set} $B(v)$, the \emph{kernel set}
$K(v)$, and the \emph{point set} $Q(v)$. These sets will satisfy
$B(v) \subseteq K(v) \subseteq Q(v)$. It will also hold that $B(v) \subseteq N(v)$.
A bag $v$ is called \emph{empty} if $Q(v) = \emptyset$.
%[[S: a side remark: note that if $Q(v)$ is non-empty, then $B(v)$ is non-empty too
%(and thus $\hat Q(v)$ is non-empty as well); this can be proved easily by induction.
%This means that we can define an empty bag as one with either empty native
%point set or one with point set, it doesn't matter]]

Algorithm $LightSp$ processes the forest $\cF$ bottom-up. In other words, it starts
with processing bags of $\cF_1$, then it proceeds to processing bags
of $\cF_2$, and so on. At the last iteration the algorithm processes bags of $\cF_\ell$.
We refer to the processing of bags of $\cF_j$ as the \emph{$j$-level processing},
for each index $j \in [\ell]$. (It will be described in Section \ref{s:jlevel}.) The algorithm maintains the point sets $B(v),K(v)$ and $Q(v)$
of all bags $v \in \cF_j$ during the $j$-level processing
 in the following way.
For a bag $v \in {\cal F}_1$, we set
$B(v) = K(v) = Q(v) =  N(v)$.

%Next, consider a bag $v \in \cF_j$, $j \in [2,\ell]$.
%Denote by $c_1(v),c_2(v),\ldots,c_\rho(v)$ the children of $v$ in $\cal F$.
%Some of them might become step-children of other $j$-level bags $u, u \ne v$.
%We say that these bags are the \emph{disappearing} children of $v$.
%Denote by $c^{(1)}(v),c^{(2)}(v),\ldots,c^{(h)}(v)$ the non-empty children of $v$
%which are not disappearing, i.e., those which did not become step-children
%of some other $j$-level bags $u, u\ne v$. We say that these bags are
%the \emph{surviving children of $v$}. The \emph{base point set} $B(v)$
%of $v$ is defined as the union of the base point sets of its surviving children,
%i.e., $B(v) = \bigcup_{i=1}^h B(c^{(i)}(v))$.
A non-empty $(j-1)$-level bag $z$, $j \in [2,\ell]$, may become a step-child
%by the algorithm 
of some $j$-level bag $v$, other than the parent $\pi(z)$ of $z$
in $\cF$. If this happens, we say that \emph{$z$ is disintegrated from $\pi(z)$},
and also that \emph{$z$ joins $v$}.
%We next consider a bag  $v \in {\cal F}_j$, for some index $j \in [2,\ell]$.
%For a bag $v$ in $\cF$, denote by $\cD(v)$ the set of descendants of $v$
%in $\cF$, including $v$. Denote by $\cZ(v)$ the subset of $\cD(v)$ that
%contains those descendants $w$ of $v$ that were disintegrated from their
%$\cF$-parents. (In other words, they were integrated into bags $u$ which are located
%elsewhere in $\cF$.) [[S: Do we need this paragraph?]]
%For a bag $v \in \cF_j$, $j \in [2,\ell]$,
Denote by $\cJ(v)$ the set of bags $z$ that join the bag $v$. They will be   referred to as the \emph{joining step-children} (or shortly, \emph{step-children}) of $v$.
%[[S: We will refer to these bags as \emph{step-children} of $v$,
%and $\cJ(v)$ will be referred to as the set of step-children of $v$ [[S: meyutar]].
Denote also by $\cS(v)$ the set of
\emph{surviving children} of $v$, i.e., the non-empty bags $z$ with $v = \pi(z)$
that did not join some other $j$-level bag $v'$,
$v' \ne v$ ($v' \in \cF_j$). Let $\chi(v) = \cS(v) \cup \cJ(v)$ be
%. We will refer to $\chi(v)$ as
the set
of \emph{extended children} of $v$.
Observe that $\chi(v) \subseteq \cF_{j-1}$,
and that all bags in $\chi(v)$ are non-empty.

Each bag $z$ will be a step-child of at most one bag $v$. Also, for any bag $v$, each non-empty child $u$ of $v$
%[[S: I think that this scenario is impossible (that children of
%$v$ will not survive, even though it adopts now a step-child)]]
which is not surviving will necessarily be
a step-child of some other bag $v' \ne v$. (The bags $v$ and $v'$ are of the same level.) Hence, for each level $j \in [\ell]$,
the collection $\{Q(v) ~\vert~ v \in \cF_j\}$ is a partition of $Q$. In particular, for distinct $u,v \in \cF_j$,
$Q(u) \cap Q(v) = \emptyset$.
%[[S: We should move this paragraph below,
%since $Q(v)$ is undefined here]]

The \emph{base point set} $B(v)$ (respectively, \emph{point set} $Q(v)$) of a bag $v \in \cF_j, j \in [2,\ell]$, is
%defined
%recursively. For $v \in \cF_1, B(v) = N(v)$.
%For $v \in \cF_j, j \in [2,\ell]$, the base point set $B(v)$ is
defined as the union of the base point sets (resp., point sets) of its surviving (resp., extended) children,
i.e.,
%\begin{equation}
%\label{twovertast}
$
B(v) ~=~ \bigcup_{z \in \cS(v)} B(z),~Q(v) ~=~  \bigcup_{z \in \chi(v)} Q(z).
$
The \emph{kernel set} $K(v)$ of $v$ is an intermediate set, in the sense that
$B(v) \subseteq K(v) \subseteq Q(v)$. We will soon specify which of the points of
$Q(v)$ are included into $K(v)$. Intuitively, all points of $K(v)$
will always be pretty close to the base point set $B(v)$, both in terms of the metric
distance in $M$, and in terms of the hop-distance. This will guarantee that points of $K(v)$ 
provide good substitutes for points of $B(v)$. Consequently, the points of $K(v)$ will be used to alleviate 
the degree load from the points of $B(v)$. 

%Recall that in Algorithm $SimpleLightSp$
The algorithm will assign to every bag $v$ a representative point $r(v)$.
As discussed in the introduction, if one selects representatives only from the native point set $N(v)$, then large
maximum degree of the resulting spanner may be \emph{inevitable}, regardless
of the specific way in which representatives are selected. This may happen, for example, if there
is a point $p$ which is far away in the path metric $M_\cL$ from any other point
of $M$, but close to many points of $M$ in the original metric.
This point may be the
only point in the point set of some bag $v = v^{(0)}$, as well as in the point sets of many
of its ancestors
%$v^{(2)},v^{(3)},\ldots$.
$v^{(1)} = \pi(v),
v^{(2)} = \pi(\pi(v)),\ldots$ in $\cF$.
In this case $p$ will necessarily serve as a
representative of all these bags, and will accumulate a large degree.
%However, this representative cannot be chosen arbitrarily from
%the native point set $N(v)$ of $v$.    %, this is no longer the case in Algorithm $LightSp$.
%Indeed selecting the representatives arbitrarily may result in a large maximum degree
%of the constructed spanner. This is because the same point $p$ may serve as a
%representative of a bag $v = v^{(1)}$, and of many of its ancestors
Instead, we will pick $r(v)$ from the kernel set $K(v)$.
%[[S: I think that we used a similar notation for
%the surviving children, in particular, $c^{(i)}(v)$]]
%many different bags $v_{i_1},v_{i_2},\ldots,v_{i_h}$,
%$1 \le i_1 < i_2 < \ldots < i_h \le \ell$, with $v_{i_1}$ being a descendant
%of $v_{i_2}$, $\ldots$, $v_{i_{h-1}}$ being a descendant of $v_{i_h}$. [[S: nisuach]]

The kernel set $K(v)$ of a bag $v \in \cF_j, j \in [2,\ell]$, is defined as follows.
%recursively too.
%For a bag $v \in \cF_1$, $K(v) = B(v) (=Q(v))$.
%Consider a bag $v \in \cF_j, j \in [2,\ell]$.
%The \emph{surviving point set} $Q'(v)$ and 
The \emph{surviving kernel set} $K'(v)$
is given by
%$Q'(v) ~=~ \bigcup_{z \in \cS(v)} Q(z)$
%\begin{equation} \label{fourvertast}
%and 
$K'(v) ~=~ \bigcup_{z \in \cS(v)} K(z)$.
%, respectively.
%\end{equation}
If $|K'(v)| \ge \ell$ then the kernel set of $v$ is set to be equal to
its surviving kernel, i.e., $K(v) = K'(v)$. Otherwise (if $|K'(v)| < \ell$), we
set
%\begin{equation} \label{excmark}
$K(v) ~=~ K'(v) \cup \bigcup_{z \in \cJ(v)} K(z) ~=~ \bigcup_{z \in \chi(v)} K(z).$
%\end{equation}

%  Having
%$p$ the only point in the point set of many bags $v^{(1)},v^{(2)},\ldots$
%is the most extreme scenario, but
%we encounter the same problem whenever a super-constant number of consecutive bags
%contain a constant number of points, or more generally, whenever there are
%only $k$ points in $\omega(k)$ consecutive bags, for each $k = 1,2,\ldots,o(\ell)$; needs rewriting of course]]
%there are constant number of points
%in a super-constant number of iterations.
% The same
%problem occurs when, instead of a single point $p$, these bags will contain
%a small number of points

The intuition behind increasing the kernel set $K(v)$ beyond its surviving kernel set
$K'(v)$ (i.e., setting $K(v) = K'(v) \cup \bigcup_{z \in \cJ(v)} K(z)$)
in the case that $|K'(v)| < \ell$ is that in this case the surviving kernel set is
too small. Hence one needs to add to it more points to alleviate
the degree load.

In the complementary case $(|K'(v)| \ge \ell)$, 
%the surviving kernel set $K'(v)$ also
%satisfies $|K'(v)| \ge \ell$. (See the proof of Lemma \ref{kv} below.)
%(see the  second assertion of Lemma \ref{kv} below)
%(this follows from the fact that $K'(v) = K(Lemma \ref{kv} below).
%  Thus even without increasing
one can distribute the load of the $O(\ell)$ auxiliary spanners that Algorithm
$LightSp$ constructs (see Section \ref{s:jlevel}) among the points of $K'(v) (=K(v))$ in such a way that no kernel point
is overloaded.
%the kernel set $K(v)$ beyond its surviving kernel set $K'(v)$, we will have $|K(v)| = |K'(v)| \ge \ell$]]
%Thus we can distribute the degree load of the $O(\ell)$ spanners that
%Algorithm $LightSp$ will construct between the points of $K(v)$ in such a
%way that each point of $K(v)$ will get loaded by only $O(1)$ auxiliary spanners.
%  needs rewriting, I think; see my (not so great) suggestions]]

%Our way to bypass this problem is to integrate such bags $v$ (that have few
%potential representatives, or whose points are already overloaded) with other bags.
%The kernel set $K(v)$ of a bag $v$ will contain \emph{potential representatives},
%that is, points that are ``eligible'' to become representatives either for $v$ or for
%$v$'s ancestors in $\cF$.
\begin{definition} \label{smalllarge}
A bag $v$ is called \emph{small} if $|Q(v)| < \ell$, and \emph{large} otherwise.
\end{definition}
The next lemma follows from these definitions.
% found in Appendix \ref{app:kv}
\begin{lemma} \label{kv}
Fix an arbitrary index $j \in [\ell]$, and let $v$
be a $j$-level bag.
~(1) If $v$ is small, then $K(v) = Q(v)$.
~(2) If $v$ is large,  then $|K(v)| \ge \ell$.
%\end{enumerate}
\end{lemma}
\begin{proof}
We prove both assertions of the lemma by induction on $j$.
\\\emph{Basis: $j = 1$.} In this case $B(v) = K(v) = Q(v)$.
Also, if $v$ is large, then $|K(v)| = |Q(v)| \ge \ell$.
\\\emph{Induction Step: Assume the correctness of the statement for all smaller
values of $j, j \in [2,\ell]$, and prove it for $j$.}

We start with proving the first assertion of the lemma, i.e., we assume that $v$ is small
%  Instead: We start with proving the first assertion of the lemma, i.e. we assume that $v$ is small
and show that $K(v) = Q(v)$.
%  I changed your proof a bit, there were a few inaccuracies]]
Recall that $Q(v)$  is given by
%$Q'(v) = \bigcup_{z \in (\cS(v) \cup \cJ(v))} Q(z)$
$Q(v) = \bigcup_{z \in \chi(v)} Q(z)$.
Moreover, $K'(v) \subseteq Q(v)$, and thus $|K'(v)| \le |Q(v)| <\ell$.
By definition,
$K(v) ~=~ \bigcup_{z \in \chi(v)} K(z).$
Observe that for every $z \in \chi(v)$, $Q(z) \subseteq Q(v)$.
Hence all bags $z \in \chi(v)$ are small as well.
The first assertion of the induction hypothesis implies that $K(z) = Q(z)$, for each $z \in \chi(v)$. Hence $K(v) = Q(v)$.

Next, we prove the second assertion of the lemma, i.e., we assume that $v$ is large
and show that $|K(v)| \ge \ell$.

Suppose first that $|K'(v)| = |\bigcup_{z \in \cS(v)} K(z)| \ge \ell$.
In this case $K(v) = K'(v)$, and so $|K(v)| \ge \ell$.
% = \bigcup_{z \in \cS(v)} K(z)$.
%If one of the surviving children $z \in \cS(v)$ is large, then the second assertion of the induction hypothesis
%implies that $|K(z)| \ge \ell$, and thus $|K(v)| \ge |K(z)| \ge \ell$ as well.
%Otherwise (all bags $z \in \cS(v)$ are small), by the first assertion of the induction hypothesis,
%$K'(v) = \bigcup_{z \in \cS(v)} K(z) = \bigcup_{z \in \cS(v)} Q(z) = Q'(v)$.
%Hence, in this case $|K(v)| \ge |K'(v)| = |Q'(v)| \ge \ell$, as required.

We are now left with the case that $|K'(v)| < \ell$.
%In this case $K(v) = K'(v) \cup \bigcup_{z \in \cJ(v)} K(z)$.
%Also, it holds that $Q(v) = Q'(v) \cup \bigcup_{z \in \cJ(v)} Q(z)$.
In this case $K(v) = \bigcup_{z \in \chi(v)} K(z)$.

%Thus
%each $z \in \cS(v)$ is a small bag, which implies that $Q(z) = K(z)$, for every bag $z \in \cS(v)$.
%It follows that $$K'(v) ~=~ \bigcup_{z \in \cS(v)} K(z)
%~=~ \bigcup_{z \in \cS(v)} Q(z) ~=~ Q'(v).$$

If there exists a large bag $z \in \chi(v)$, 
%such that $|Q(z)| \ge \ell$,
then the second assertion of the induction hypothesis yields $|K(z)| \ge \ell$,
which implies that $|K(v)| \ge |K(z)| \ge \ell$.

Otherwise,   all bags $z \in \chi(v)$ are small. The first assertion
of the induction hypothesis yields $K(z) = Q(z)$, for all bags $z \in \chi(v)$,
and thus $$K(v) ~=~ \bigcup_{z \in \chi(v)} K(z) ~=~ \bigcup_{z \in \chi(v)} Q(z) = Q(v).$$
Hence $|K(v)| = |Q(v)| \ge \ell$.
\QED
\end{proof}

As  mentioned above, for every index $j \in [\ell]$,
%the point sets
%$Q(v)$ of bags $v \in \cF_j$ define a partition of the point set $Q$,
%i.e.,
$Q = \bigcup_{v \in \cF_j} Q(v)$, and for any pair $u,v$ of distinct $j$-level bags,
$Q(u) \cap Q(v) = \emptyset$.
It can also be readily verified that $Q(v) = \emptyset$ iff  $B(v) = \emptyset$.
%  What's the latex symbol
%for disjoint union?]]
%These assertions can be readily verified.    % on the index $j$.

\subsection{The Base Edge Set} \label{disregard}
%\subsection{Base Edges} \label{disregard}
%  Base edges section]]
The algorithm will also maintain a set of edges $\cal B$, which we
call the \emph{base edge set} of the spanner.
%Before we describe how this edge set is formed,
%we need to provide a few definitions.
%This edge set is formed in the following way.
%Fix  an index $j, j \in [\ell]$.

For each non-empty bag $v \in {\cal F}$, the base edge set $\cB$ will connect
the base point set $B(v)$ of $v$ via a simple path $P(v)$. That is, if we denote the points
of $B(v)$ from left to right (w.r.t.\ the order relation $\prec_{\cal L}$)
by $p_1,\ldots,p_k$, then $P(v) = (p_1,\ldots,p_k)$.
%Having such a path $P(v)$ for each bag $v$
%is essential for   bounding the stretch and diameter.
We will show
%in order to guarantee that the degree and lightness are in check, we will require
that $\Delta(\cB) \le 2$ and $\Psi(\cB) = O(\ell)$.
%The description of how the base edge set $\cB$ is formed, and the analysis of its degree and lightness, are deferred to Appendix \ref{disregard}.

%Next, we describe how the base edge set $\cB$ is formed,
%and analyze its properties.

%In addition to the point set $N(v)$ of a bag $v$, the improved algorithm will also
%maintain a subset $B(v)$ of $N(v)$ which we call the \emph{base point set} of $v$.
%This point set is formed in the following way.
Fix an index $j, j \in [\ell]$. For each non-empty bag $v \in \cF_j$,
let $x(v)$ (respectively, $y(v)$)
denote the leftmost (resp., rightmost) (with respect to $\prec_{\cal L}$) point in
the base point set $B(v)$ of $v$.
The next observation, which follows easily from the definition of $B(v)$ ($B(v) = \bigcup_{z \in \cS(v)} B(z))$,
implies that the order relation $\prec_{\cal L}$ can be used in the obvious way to define a total order on the non-empty bags of ${\cal F}_j$.
\begin{observation} \label{obviousorder}
For any pair $u,v$ of distinct non-empty bags in
${\cal F}_j$, either $x(u) \preceq_{\cal L} y(u) \prec_{\cal L} x(v) \preceq_{\cal L} y(v)$
or $x(v) \preceq_{\cal L} y(v) \prec_{\cal L} x(u) \preceq_{\cal L} y(u)$ must hold.
With a slight abuse of notation, we will write $u \prec_{\cal L} v$ in the former case and $v \prec_{\cal L} u$ in the latter.
\end{observation}
We may henceforth assume without loss of generality that,
for each bag
$v \in {\cal F}_j$, with $j \ge 2$,
its surviving children $c^{(1)}(v),c^{(2)}(v),\ldots,c^{(h)}(v)$
are ordered such that $c^{(1)}(v) \prec_{\cal L} c^{(2)}(v) \prec_{\cal L} \ldots \prec_{\cal L} c^{(h)}(v)$.

Next, we turn to a detailed description of the way
%describe how
that the base edge set $\cal B$ is constructed.

On the bottom-most level ($j=1$), for each bag $v \in {\cal F}_1$,
we order all points of $B(v) = N(v)$ from left to right, according to their respective order
in $\cal L$. In other words, write $B(v) = (p_1,p_2,\ldots,p_{|B(v)|})$,
where $p_1 \prec_{\cal L} p_2 \prec_{\cal L} \ldots \prec_{\cal L} p_{|B(v)|}$.
The $(|B(v)| -1)$ edges $(p_1,p_2),\ldots,(p_{|B(v)|-1},p_{|B(v)|})$ form the \emph{base edge set} ${\cal B}(v)$
of the bag $v$.
The union ${\cal B}_1 = \bigcup_{v \in {\cal F}_1} {\cal B}(v)$
is the \emph{1-level base edge set}.

For $j \ge 2$, the \emph{base edge set} ${\cal B}(v)$
of a $j$-level bag $v$ is
formed
%as follows.
in the following way.
Recall that $c^{(1)}(v),c^{(2)}(v),\ldots,c^{(h)}(v)$ denote the surviving children of $v$ from left to right (w.r.t.\ $\prec_{\cal L}$),
and denote by
$x^{(i)}(v)$ (respectively, $y^{(i)}(v))$  the left-most (resp., right-most)
point in the base point set $B(c^{(i)}(v))$ of $c^{(i)}(v)$, for each index $i \in [h]$.
%Some of them might become step-children of other $j$-level nodes $u, u \ne v$.
%We say that these children of $v$ \emph{disappear}. Let $c_{i_1}(v),c_{i_2}(v),
%\ldots,c_{i_h}(v)$ denote those of them that did \emph{not} disappear.
Then the \emph{base edge set} ${\cal B}(v)$ of $v$ will be the edge set
%given by
${\cal B}(v) = \{(y^{(1)}(v),x^{(2)}(v)),(y^{(2)}(v),x^{(3)}(v)),\ldots,(y^{(h-1)}(v),x^{(h)}(v))\}$.
%By Observation \ref{
Given the base edge sets of all $j$-level bags $v \in \cF_j$,
the \emph{$j$-level base edge set} ${\cal B}_j$ is formed as their union,
i.e., ${\cal B}_j = \bigcup_{v \in {\cal F}_j} {\cal B}(v)$.
Finally, the \emph{base edge set} $\cal B$ is formed as the union ${\cal B} = \bigcup_{j=1}^\ell {\cal B}_j$.
(See Figure \ref{fig9} for an illustration.)
\begin{figure*}[htp]
\begin{center}
\begin{minipage}{\textwidth} %{5in}
\begin{center}
\setlength{\epsfxsize}{5.5in} \epsfbox{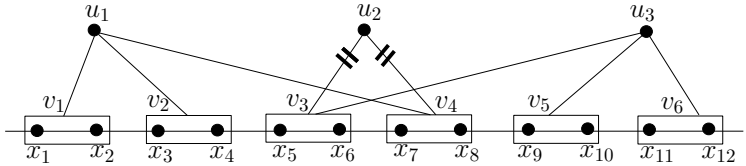}
\end{center}
\end{minipage}
\caption[]{ \label{fig9} \sf In this example $\cL = (x_1,x_2,\ldots,x_{12})$.
In level 1 there are 6 bags, $v_1,\ldots,v_6$, with $v_i = (x_{2i-1},x_{2i})$,
for $i \in [6]$.
In the forest $\cF$, $u_i$ is the parent of $v_{2i-1}$ and $v_{2i}$, for $i \in [3]$.
In the forest $\hat \cF$, $v_4 \in \cJ(u_1)$ (i.e., $v_4$ is a step-child of $u_1$),
and $v_3 \in \cJ(u_3)$. The bag $u_2$ becomes empty. The base edge sets are
$\cB_1 = \{(x_1,x_2),(x_3,x_4),\ldots,(x_{11},x_{12})\}$ and
$\cB_2 = \{(x_2,x_3),(x_{10},x_{11})\}$.}
%The recursive edge sets are $\hat \cB_1 = \{(x_1,x_2),(x_3,x_4),\ldots,(x_{11},x_{12})\}$ and
%$\cB_2 = \{(x_2,x_3),(x_{10},x_{11})\}$.}
\end{center}
\end{figure*}

We also define the \emph{recursive base edge set} $\hat \cB(v)$ 
%  I think we should use another symbol instead of hat,
%because we're using hat in a different context]] 
of a $j$-level
bag $v$ in the following way.
For $j = 1$, $\hat \cB(v) = \cB(v)$.
For $j \in [2,\ell]$, the recursive base edge set $\hat \cB(v)$ of $v$ is
defined as the union of the recursive base edge sets
$\hat \cB(c^{(1)}(v)),\ldots,\hat \cB(c^{(h)}(v))$
of its surviving children $c^{(1)}(v),\ldots,c^{(h)}(v)$, respectively,
union with the base edge set $\cB(v)$ of $v$.
In other words, $\hat \cB(v) = \cB(v) \cup \bigcup_{i=1}^h \hat \cB(c^{(i)}(v))$.
%  We should mention that for every bag $v \in {\cal F}$, the base edge set ${\cal B}$ (or more specifically, the union of
%all edge sets ${\cal B}(u)$, taken over all descendants $u$ of $v$ in $\cal F$) connects the base point set $B(v)$ of $v$ via a simple path
%which we denote by $P(v)$; we say that $P(v)$ is the \emph{base path} of $v$; it is important that the base path $P(v)$ contains
%no points outside $B(v)$.]]
%  I don't think that we should provide the proof of the following lemma;
%we should just say that it follows from the construction by a straightforward induction]]
The following lemma follows from the construction by a straightforward induction.
%  I removed the proof, OK?]]
\begin{lemma} \label{recur}
Fix an arbitrary index $j \in [\ell]$, and let $v$ be an arbitrary non-empty $j$-level bag.
Let $B(v) = (p_1,\ldots,p_{|B(v)|})$ be the base point set of $v$, ordered according to
$\prec_\cL$. (In other words, $p_1 \prec_\cL p_2 \prec_\cL \ldots \prec_\cL p_{|B(v)|}$.)
Then the recursive base edge set $\hat \cB(v)$ is the edge set given by
$\hat \cB(v) = \{(p_1,p_{2}),\ldots,(p_{|B(v)|-1},p_{|B(v)|})\}$.
\end{lemma}

Consider the path  $P(v) = ((p_1,p_2),\ldots,(p_{|B(v)|-1},p_{|B(v)|}))$.
%$the \emph{base path} $P(v)$ of the bag $v$.
By Lemma \ref{recur},
the edge set of the path $P(v)$ is equal to the recursive base
edge set $\hat \cB(v)$ of $v$.
%(The base point set $B(v)$ is equal to $\{p_1,p_2,\ldots,p_{|B(v)|}\}$,
%with $p_1 \prec_\cL p_2 \prec_\cL \ldots \prec_\cL p_{|B(v)|}$.)

For a point $p$ and an index $j \in [\ell]$, we say that a bag $v \in \cF_j$ (if exists) is
the \emph{$j$-level base bag} of $p$ if $p \in B(v)$.
Recall that for a pair $u,v \in \cF_j$ of distinct bags, $B(u) \cap B(v) = \emptyset$.
Hence for any point $p$ and index $j \in [\ell]$, there is at most one $j$-level base bag.
Moreover, for any point $p$ there exists a 1-level base bag. However,
on subsequent levels the base bag of $p$ may not exist;
this happens when the base bag $v$ of $p$ becomes a step-child of some other bag $u$ different from its parent $\pi(v)$ in $\cF$.
In other words, for any point $p$, there exists an index $j = j(p) \in [\ell]$
such that there exist $i$-level base bags for $p$, for all indices $1 \le i \le j$,
and there are no $i$-level base bags for $p$, for all indices $j+1 \le i \le \ell$.
We will say that the base bags of $p$ in levels $1,\ldots,j-1$ are \emph{surviving},
and the base bag of $p$ in level $j$ is \emph{disappearing}.  % in iteration $j$.
%  This seems like an important paragraph, and it is ``buried'' in this boring section...]]

%The proof of the following lemma can be found in Appendix \ref{disregard}.
%\begin{lemma} \label{baselemma}
%Next, we argue that the maximum degree $\Delta(\cB)$ of the base edge set $\cB$ is at most 2,
%and its lightness $\Psi(\cB)$ satisfies $\Psi(\cB) = O(\ell)$.
%\end{lemma}

\ignore{
Consider again the Hamiltonian path $\cal L$
and the hierarchy $\cal I$ of intervals, which we defined
in the beginning of Section \ref{first}.
Specifically, we defined $n_0 = | {Q}  | = n$.
%$P_0 =  P, n_0 = |P_0| = n$.
Also, for each index $j \in [\ell]$, we defined $\xi_j = \rho^{j-1} \cdot \frac{L}{n}$.
(Recall that $L = \omega({\cal L})$.)
For each index $j \in [\ell]$ we partitioned the path $\cal L$ into
$n_j = \frac{c \cdot L}{\xi_j} = \frac{c \cdot n}{\rho^{j-1}}$ intervals
of length $\mu_j = \frac{\xi_j}{c}$ each, denoted from left to right
by $I_j^{(1)},\ldots,I_j^{(n_j)}$; we called these intervals the $j$-level intervals.
}

Next, we argue that the maximum degree $\Delta({\cal B})$ of the base edge set $\cal B$ is at most 2,
and that its lightness $\Psi(\cal B)$ is $O(\ell)$. 

We start with analyzing $\Delta(\cal B)$.
For each point $p \in {Q}$ and any index $j \in [\ell]$, we say that a point $q \in {Q}$ is a \emph{left neighbor}
(respectively, \emph{right neighbor}) of $p$ in $\cB_j$ if the edge $(p,q)$
belongs to $\cB_j$ and $q \prec_{\cal L} p$ (resp., $p \prec_{\cal L} q$).
%If $q$ is a left
%Similarly, we say that a point $q \in {Q}$
In addition, we will say that $q$ is a left neighbor (respectively, right neighbor)
of $p$ in $\cB$, if there exists an index $j \in [\ell]$, such that $q$ is is a left neighbor
(resp., right neighbor) of $p$ in $\cB_j$.
%then
The \emph{left degree} (resp., \emph{right degree}) of $p$ in $\cal B$, denoted
$leftdeg_{\cal B}(p)$ (resp., $rightdeg_{\cal B}(p)$) is the number of left neighbors (resp., right neighbors)
$q$ of $p$ in $\cal B$.

Next, we argue that for every point $p \in {Q}$, $rightdeg_{\cal B}(p) \le 1$. Symmetrically,
it also holds that $leftdeg_{\cal B}(p) \le 1$. We will conclude that $deg_{\cal B}(p) = leftdeg_{\cal B}(p) + rightdeg_{\cal B}(p) \le 2$,
and thus $\Delta({\cal B}) \le 2$.
%  def'n: for a point $x \in Q$ and index $j \in [\ell]$,
%let $v_j(x)$ be the \emph{host node} of $x$ on level $j$,
%i.e., the $j$-level node in $\cal F_j$, such that $x \in Q(v_j(x))$.]]
%Let $h, h \in [\ell]$, denote the minimum index such that $x$
%is not the rightmost point of the base point set $B(v_h(x))$
%of the host node $v_h(x)$.
%Now we are ready to prove the estimate for point degrees in ${\cal B}$.
\begin{lemma}
For every point $p \in {Q}$, $rightdeg_{\cal B}(p) \le 1$.
\end{lemma}
\begin{proof}
Suppose first that $p$ is not the rightmost point of a base point set $B(v)$,
for some 1-level bag $v \in F_1$.     % then its right degree is 1.
Denote by $p'$ the right neighbor of $p$ in ${\cal B}_1$.
In this case, by construction, for every bag $v \in {\cal F}$ such
that $p \in B(v)$, it also holds that $p' \in B(v)$. Therefore,
$p$ will not be the rightmost point of $B(v)$, for any bag $v \in {\cal F}$.
Hence $p$ will not have any right neighbor in $\bigcup_{j=2}^\ell {\cal B}_j$, and so $rightdeg_\cB(p) = 1$.

Suppose now that $p$ is the right-most point of a base point set $B(v)$,
for a 1-level bag $v \in {\cal F}_1$. Let $h, h \in [\ell]$, denote the
maximum level such that $p$ is the right-most point of a base point
set $B(v)$, for an $h$-level bag $v \in {\cal F}_h$. By construction,
%for every
%index $i \in [h]$,
$p$ will not have any right neighbor in $\bigcup_{j=1}^h {\cal B}_j$.
% $rightdeg_{{\cal B}_i}(p) = 0$.
Denote by $v_i \in {\cal F}$ the base bag
of $p$ on level $i$ (if exists), for each index $i \in [\ell]$.
(In other words,
$p \in B(v_i)$, for each index $i$ as above.)
%Note that a bag $v_i$ as above may not exist for some values of $i$.
%Specifically, the host bag $v_1$ of $p$ (which always exists) may become
%a step-child of some other bag rather than its parent in $\cal F$.
%(In this case we say that $v_1$ is disappearing.)
%In this case $p$ has no host bags on levels $i \ge 2$.
%More generally,
Recall that there exists an index $j \in [\ell]$ such that the bags $v_1,v_2,\ldots,v_{j-1}$
are surviving, but the bag $v_{j}$ is disappearing.
%,v_{j+2},\ldots,v_\ell$ are not surviving nodes.
%(In this case $v_j$ is a disappearing node.)
It holds, however, that $h \le j$.
If the bag $v_{h}$ is disappearing (i.e., if $h = j$), then the point $p$ acquires no
right degree on levels $h+1,h+2,\ldots,\ell$.
Hence, in this case $rightdeg_{\cal B}(p) = 0$. Otherwise the point $p$ acquires exactly
one right neighbor on level $h+1$. From that moment on, however, $p$ will no
longer be the rightmost point of the base point sets $B(v_i)$ of its host bags.
Hence it acquires no additional right neighbors on subsequent levels.
In this case $rightdeg_{\cal B}(p) = 1$.
%Denote by $v_i$ the host node   needs def'n]] of $x$ on level $i$.
\QED
\end{proof}

%To summarize:
\begin{corollary} \label{degbase}
$\Delta({\cal B}) \le 2$ and $|\cB| \le n$.
\end{corollary}

%Since the maximum degree of the base edge set $\cB$ is 2,
%Corollary \ref{degbase}
%implies that there are at most $n$ base edges in $\cB$.
%\begin{corollary} \label{numedgebase}
%$|{\cal B}| \le n$.
%\end{corollary}

Next we analyze the lightness $\Psi(\cB)$ of the base edge set $\cB$.
Observe that for each index $j \in [\ell]$, the $j$-level base edge set ${\cal B}_j$ is a collection of vertex-disjoint paths.
% [[S: this observation is incorrect (see Figure 1 for example), and redundant]].
%[[S: I suggest to turn the following lemma into an observation, and omit its proof (which is boring and easy).]]
By the triangle inequality, the weight $\omega({\cal B}_j)$ of
${\cal B}_j$
is bounded above by the weight $L = \omega({\cal L}) = O(t \cdot \omega(MST(M[Q])))$ 
%[[S: $L$ might be greater by a factor of $t$ than the MST]] 
of the  Hamiltonian path $\cL$,
for each index $j \in [\ell]$.
We conclude that the weight $\omega(\cB)$ of the base edge set $\cB = \bigcup_{j=1}^\ell {\cal B}_j$ satisfies $\omega(\cB) ~=~ \omega(\bigcup_{j=1}^\ell {\cal B}_j)
~\le~ \sum_{j=1}^\ell \omega({\cal B}_j) ~\le~ \ell \cdot O(t \cdot \omega(MST(M[Q]))) ~=~ O(\ell \cdot t) \cdot \omega(MST(M[Q]))$.
\begin{corollary} \label{weightbase}
$\Psi(\cB) = O(\ell \cdot t)$.
\end{corollary}
{\bf Remark:} If the metric is low-dimensional Euclidean or doubling, then $\omega(\cL) \le 2 \cdot \omega(MST(M[Q]))$,
and so $\Psi(\cB) = O(\ell)$.

%\vspace{-0.05in}
\subsection{Zombies and Incubators} \label{thealg}
%\vspace{-0.05in}
%In this section we describe how Algorithm $LightSp$
%determines the sets $\cS(v)$ (of surviving children of $v$)
%and the sets $\cJ(v)$ (of joining step-children of $v$).
%We will also furnish other details of Algorithm $LightSp$ that
%were left unspecified in previous sections.
%^^^^^^^^^^^^^^^^^
%Algorithm $LightSp$
%Let $\tilde G_0 = (Q_0,\tilde E_0)$.
Algorithm $LightSp$
starts with computing
the path-spanner $H = (Q,E_H)$ and the base edge set $\cB$.
% and inserting them into the (initially empty) spanner
%At the end of the algorithm we will return $\tilde G$ as our ultimate spanner.
%The algorithm proceeds by invoking
Next, it invokes Algorithm $BasicSp$ to build
a $t$-spanner $G'_0 = (Q_0,E'_0)$ for the sub-metric
$M[Q_0]$ of $M$ induced by  $Q = Q_0$.
Define $\tilde E_0$ to be the edge set obtained by
\emph{pruning} $E'_0$, i.e., removing all edges
of weight greater than the \emph{0-level threshold} $\tau_0$.
%for a constant $0 \le  \le 1$, and a parameter
%$\tau_0 ~=~ \frac{L}{n} \cdot t \cdot (1+\frac{1}{c}) ~=~ \xi_0 \cdot \rho \cdot t \cdot (1+\frac{1}{c})
%= \mu_0 \cdot \rho \cdot t \cdot (c+1).$
%Then it computes a $t$-spanner $G'_0 = (Q_0,E'_0)$
%for $M[Q_0]$, and removes from $G'_0$ all edges of weight greater
%than $\tau_0 = (1+\) \cdot \tau_0$, where $\tau_0$ is given by Equation (\ref{tau}),
%and $0 \le \ \le 1$ is a small universal constant, to be fixed
%in the sequel.
%The resulting edge set is denoted $\tilde E_0$,
%and its
The corresponding graph $\tilde G_0 = (Q_0,\tilde E_0)$ is called the \emph{0-level auxiliary
spanner}.
%It is then inserted into our ultimate spanner $\tilde G$.
%More generally, we define the the $j$-level threshold
%$\tau_j = (1+\) \cdot \tau_j$,
%with $\tau_j ~=~ \rho^j \cdot \frac{L}{n} \cdot t \cdot (1+\frac{1}{c}) ~=~ \xi_j \cdot \rho \cdot t \cdot (1+\frac{1}{c})
%~=~ \mu_j \cdot \rho \cdot t \cdot (c+1)$,
%for every $j \in [0,\ell]$.
%Our ultimate spanner $\tilde G$ contains the auxiliary base edge set $\cB$,
%the auxiliary path-spanner $H$ and the  0-level auxiliary spanner $\hat G_0$.
%all constructed as was described above.
In a similar way (details will be provided in Section \ref{s:jlevel}), the algorithm
%the spanner $\tilde G$
builds auxiliary $j$-level spanners $\tilde G_j = (Q_j,\tilde E_j)$, for each $j \in [\ell]$.
The spanner $\tilde G_j$ is a graph over the set of representatives of the non-empty $j$-level bags. The representatives are determined according to rules that will be specified in Section \ref{sec24}.
%$G^*_j$ and $\hat G_j$.
%These spanners are build in a way that is similar to the way the 0-level auxiliary spanner is built.
%However, they are pruned w.r.t.\ the
%$j$-level threshold
%$\tau_j = (1+\) \cdot \tau_j$, $j \in [\ell]$,
%where $\tau_j ~=~ \rho^j \cdot \frac{L}{n} \cdot t \cdot (1+\frac{1}{c}) ~=~ \xi_j \cdot \rho \cdot t \cdot (1+\frac{1}{c})
%~=~ \mu_j \cdot \rho \cdot t \cdot (c+1)$.
%Let $\tilde G_j$ denote the union of $G^*_j$ and $\hat G_j$;
%it will be referred to as the \emph{$j$-level auxiliary spanner}. %Finally, it will contain the base edge set $\cB$.
The union $\cB \cup E_H \cup \bigcup_{j=0}^\ell \tilde E_j$
%of all these
%auxiliary spanners
is the ultimate spanner
$\tilde G = (Q,\tilde E)$
 that Algorithm $LightSp$ returns.

%The $j$-level threshold $\tau_j$ will be used for constructing the $j$-level
%auxiliary spanner $\tilde G_j$. (In Algorithm $SimpleLightSp$ the threshold
%$\tau_j$ was used in a similar fashion.)

As  discussed above, during the algorithm a $j$-level bag $z$
may join as a step-child of some $(j+1)$-level bag $v$, $v \ne \pi(z)$.
We now take a closer look at this process.

Each bag $v \in \cF$ may hold a \emph{label} of exactly one of two types,
a \emph{zombie} and an \emph{incubator}. Initially, all bags
are unlabeled. As Algorithm $LightSp$ proceeds, some bags may be assigned
with labels.
%Our algorithm will make sure to assign each bag with at most one label.
%In other words, each bag can be a zombie or an incubator, but not both.
%(This will be proved in the sequel??, see Lemma \ref{tat}.)

It may happen that an $i$-level bag $v$ is \emph{abandoned} by its
parent $\pi(v)$, and is
\emph{attached} to an $i$-level bag $u$.
It must hold that $\pi(v) \ne \pi(u)$.
We also say that $v$ is \emph{adopted} by $\pi(u) \in \cF_{i+1}$.
We denote the \emph{attachment} of $v$ to $u$ by $\cA(u,v)$.
We also call it an \emph{adoption} of $v$ by $\pi(u)$.
However, the attachment and adoption come with a suspension
period, henceforth, \emph{incubation period}. Specifically, there is a positive integer  $\gamma$, 
which determines the  length of the incubation period. The $(i+\gamma-1)$-level
ancestor $v'$ of $v$
will actually be \emph{disintegrated} from its parent $\pi(v')$,
and \emph{join}
the $(i+\gamma)$-level ancestor $u'$ of $u$. The bag $u'$ is referred to as the \emph{actual adopter}.
It will be shown later (see Corollary \ref{tat4} in Section \ref{zom:app}) that adoption rules (which we still did not finish to specify) imply that $\pi(v') \ne u'$.
We remark that attachments occur only for $i \le \ell - \gamma$.   %   $i \le \ell - \gamma$]].
%This guarantees that $v'$ and $u'$ are well-defined.
%The integrations and disintegrations themselves occur according to the rules
%specified in Section  2.2. In other words,
The sets $B(u'),K(u'),Q(u')$
and $B(\pi(v')),K(\pi(v')),Q(\pi(v'))$ are computed
according to the rules specified in Section \ref{sec21}.
%(Specifically, see Equations (\ref{twovertast}), (\ref{threevertast}), (\ref{fourvertast}), and (\ref{excmark}) in Section \ref{sec53}.)
%Next, we specify how the sets $\cS(v)$ and $\cJ(v)$ are computed.
%  I only see a vague specification.]]

The $\gamma-1$ immediate ancestors of $v = v^{(0)}$, namely, the bags $v^{(1)} = \pi(v), v^{(2)} = \pi(v^{(1)}),
\ldots,  v^{(\gamma-1)} = \pi(v^{(\gamma -2)}) = v'$,
%  note that I stopped two units
%before you did, at $v^{(\gamma-2)}$ rather than $ v^{(\gamma)}$]],
change
their status
%label   that's not very good, though, because they didn't have any label (and even if they did,
%we are not changing labels, just assigning them); perhaps we can write ``changed their status'' instead.]]
as a a result of this attachment. They will be now labeled as  \emph{zombies}.
%  Note that $v$ itself is not labeled as a zombie.]]
%remains unlabeled, unless it was labeled before this adoption; actually, this is impossible)]]
The bag $v'$ is called a \emph{disappearing zombie}, because it joins $u'$ rather than its original parent $\pi(v')$.
We will refer to $v$ as an \emph{attached bag}.
%  Observe that the bag $\pi(v')  = v^{(\gamma)}$
%is empty, and thus it is different than ]]
Similarly, the $\gamma-1$ immediate ancestors of $u = u^{(0)}$,
 namely, the bags $u^{(1)}= \pi(u),u^{(2)}= \pi(u^{(1)}),
\ldots,  u^{(\gamma-1)} = \pi(u^{(\gamma -2)})$,
change their status as well.
%(Note that $u$ is an $(i+1)$-level bag, and not an $i$-level bag like $v$.)
%Also, note that $u
%  Here too... note that $u'$ is on level $i+\gamma$;
%I assumed that $u$ is on level $i+1$; it's unclear if $u'$ itself should be
%refereed to as an adopter, probably yes]]
%  not necessarily, they could have been adopters before]].
They %bags $u^{(1)},u^{(2)},\ldots,u^{(\gamma-1)}$ 
will be now labeled as \emph{incubators}.
%  Note that $v$ itself is not (necessarily) labeled as an incubator.]]
The $(i+\gamma)$-level bag $u' = u^{(\gamma)}$ is \emph{not} labeled as an incubator.
This bag is called the \emph{actual adopter}.
%  I changed,
%you referred to $u'$ rather than $\pi(u')$.   IMPORTANT: It's a bit strange to call all these bags adopters,
%when the only actual adopter $\pi(u')$ (which is on level $i+\gamma$) is not an adopter.]]
We remark that the same bag may become an adopter (and incubator) of several different
descendants. Note also that, since $i \ge 1$,
for a $j$-level bag $u' \in \cF_j$ to be an actual adopter, it must
hold that
%\begin{equation} \label{stam}
$j ~=~ i + \gamma ~\ge~ \gamma + 1$.
%\end{equation}
The $i$-level bag $u = u^{(0)}$ will be referred to as the \emph{initiator} of the attachment $\cA(u,v)$.
(See Figure \ref{fig2} for an illustration.)
\begin{figure*}[htp]
\begin{center}
\begin{minipage}{\textwidth} %{5in}
\begin{center}
\setlength{\epsfxsize}{5in} \epsfbox{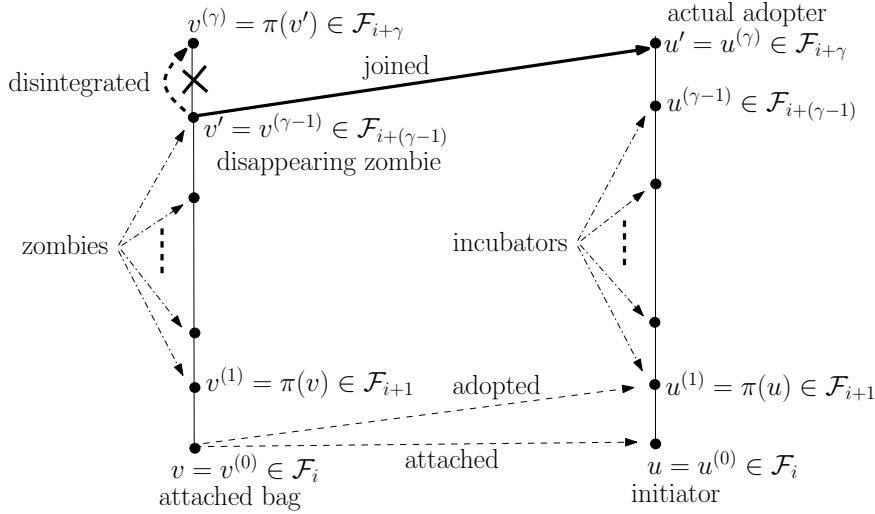}
\end{center}
\end{minipage}
\caption[]{ \label{fig2} \sf An illustration of an attachment $\cA(u,v)$.}
\vspace{-0.2in}
\end{center}
\end{figure*}

%\vspace{-0.05in}
\subsection{Representatives} \label{sec24}
In this section we specify how Algorithm $LightSp$ selects representatives for bags.

For a point $p \in Q$ and an index $j \in [\ell]$, denote by $v_j(p)$
the $j$-level \emph{host bag} of $p$, i.e., the unique bag $v_j(p)$
that satisfies $p \in Q(v_j(p))$, $v_j(p) \in \cF_j$.
(Recall that $\{Q(v) ~\vert~ v \in \cF_j\}$ is a partition of $Q$, for every index $j \in [\ell]$.)
%\begin{definition} \label{growbag}
%\end{definition}

%Algorithm $LightSp$ starts with constructing the path-spanner $H$
%and the 0-level auxiliary spanner $\tilde G_0$.  These auxiliary spanners
%do not overload the degree bound, hence when selecting the representatives they are disregarded.   rewrite]]
The algorithm
maintains a few
load indicators and counters for every point $p \in Q$. For each index $j \in [\ell]$,
the load indicator $load_j(p)$ is equal to 1 if the point $p$ is not isolated (i.e., it has at least one neighbor)
in the
$j$-level auxiliary spanner $\tilde G_j$.
%(Recall that $\tilde G_j$ is constructed during
%the $j$-level processing, and it is the union of $G''_j$ and $\hat G_j$.)
Otherwise, $load_j(p)$ is set to 0.
The \emph{load counter} $load\_ctr_j(p)$ is defined by $load\_ctr_j(p) = \sum_{i=1}^j
load_i(p)$.
Algorithm $LightSp$ also maintains three more refined load counters for every
point $p$. Specifically, the \emph{small counter} $ctr_j(p)$
(respectively, \emph{large counter} $CTR_j(p)$) is the number of
indices $i$, $1 \le i \le j$, such that the point $p$ is not isolated in
%$the $i$-level auxiliary spanner
$\tilde G_i$
\emph{and} its host bag $v_i(p)$ is small (resp., large). (See definition \ref{smalllarge}.)
 Note that
$load\_ctr_j(p) = ctr_j(p) + CTR_j(p)$.
The algorithm also counts the number of indices $i, 1 \le i \le j$,
such that the point $p$ is not isolated in $\tilde G_i$ \emph{and}
its host bag $v_i(p)$ satisfies $Q(v_i(p)) = \{p\}$. This counter is
referred to as the \emph{single counter} of $p$, and is denoted
$single\_ctr_j(p)$.
It also maintains the complementary counter $plain\_ctr_j(p) = ctr_j(p)
- single\_ctr_j(p)$,
which is referred to as the \emph{plain counter} of $p$.
(For convenience, all counters with index 0 are set as 0, i.e., $load\_ctr_0(p) = CTR_0(p) = ctr_0(p) = plain\_ctr_0(p) = single\_ctr_0(p) = 0$.)

A point $p \in Q$ may have edges incident on it in the $j$-level auxiliary spanner
$\tilde G_j$ only if it is a representative of a $j$-level bag $v \in \cF_j$.
Hence we generally make an effort to select a representative with as
small counter as possible. The specific way in which
Algorithm $LightSp$ selects representatives at the beginning of the $j$-level
processing, $j \in [\ell]$, is the following one.

The representative $r(v)$ of a non-empty 1-level bag
$v \in \cF_1$ is selected arbitrarily from $K(v) = Q(v)$.
Next, consider a non-empty $j$-level bag, $j \in [2,\ell]$.
The bag $v$ is said to be a \emph{growing bag}
if $|\chi(v)| \ge 2$, i.e., if $v$ is obtained as a result of a merge of two or more
non-empty $(j-1)$-level bags. Otherwise, the bag $v$ is called \emph{stagnating}.	
If $v$ is a stagnating bag, then necessarily $\cJ(v) = \emptyset$, and $|\cS(v)| = 1$.
(See the paragraph between Lemma \ref{tat} and Claim \ref{tate1}.)
%(We will later show
%  give reference to lemma and section??? (I don't think that such a lemma exists)]] that an actual adopter $v$ must have $\cS(v) \ne \emptyset$. Hence if %$\cS(v) = \emptyset$,
%it follows that the bag $v$ is empty, and $\cJ(v) = \emptyset$ as well.)

If $v$ is large (i.e., $|Q(v)| \ge \ell$) then Algorithm $LightSp$ appoints a point $p \in K(v)$
with the smallest large  counter $CTR_{j-1}(p)$ as its representative $r(v)$.

If $v$ is small (i.e., $1 \le |Q(v)| < \ell$) then the algorithm checks whether it is a growing bag or a
stagnating one.
If $v$ is a stagnating bag then $\cS(v) = \{w\}$, for some
$(j-1)$-level bag $w$. In this case Algorithm $LightSp$ sets the representative $r(v)$
of $v$ to be equal to the representative $r(w)$ of $w$, i.e., $r(v) = r(w)$.
Otherwise, $v$ is a growing small bag. In this case Algorithm $LightSp$ appoints
a point $p \in K(v)$ with the smallest plain counter $plain\_ctr_{j-1}(p)$ as
its representative $r(v)$.
%\vspace{-0.08in}

\subsection{Procedure $Attach$} \label{app:attachalg}
In this section we present a simple graph procedure, called Procedure $Attach$, which will be used as a building
block by Algorithm $LightSp$.
More specifically, this procedure is used to determine which bags will participate in attachments 
in the $j$-level processing. (The $j$-level processing will be described in Section \ref{s:jlevel}.)
It accepts as input an $n$-vertex graph $G=(V,E)$,
whose
vertices are labeled by either \emph{safe} or \emph{risky}. (The meaning of these labels will become clear in Section
\ref{s:jlevel}.)
The procedure returns a \emph{star forest}, i.e., a collection  $\Gamma$ of
vertex disjoint stars,
that satisfies the following two conditions.
%Recall that  Procedure $Attach$ accepts as input an $n$-vertex graph $G = (V,E)$, whose vertices
%are labeled by either safe or risky. The procedure returns as output a star forest that
%satisfies some desired conditions (see Corollary \ref{attachcor} below).
\begin{enumerate}
\item $\bigcup_{S \in \Gamma} V(S)$ contains the set $R \subseteq V$ of
vertices which are not isolated\footnote{A vertex $v$ in a (possibly directed) graph $G$ is called \emph{isolated} if 
no edge in $G$ is 
%there are no edges 
incident on $v$.}
%its degree in the graph is equal to 0,
%i.e., $deg_G(v) = 0$. In a directed graph, $deg_G(v) = indeg_G(v) + outdeg_G(v)$,
%where $indeg_G(v)$ (respectively, $outdeg_G(v)$) stands for the indegree (resp., outdegree) of $v$ in $G$.}
in $G$, and labeled as risky.
\item Each star $S \in \Gamma$ contains a \emph{center} $s \in V$
labeled as either safe or risky, and one or more leaves $z_1,\ldots,z_k \in V$
labeled as risky. The edge set $E(S)$ of a star $S$ is
given by $E(S) = \{(z_i,s) ~\vert~ i \in [k]\}$.
\end{enumerate}
See Figure \ref{fig10} for an illustration.
%\begin{figure}
%\centering
%\epsfig{file=fig10stoc.eps, height=1.1in, width=4.3in}
%\caption{\label{fig10} A star forest. Safe bags are depicted by squares, while risky ones are depicted by cycles.}
%\end{figure}
\begin{figure*}[htp]
\begin{center}
\begin{minipage}{\textwidth} %{5in}
\begin{center}
\setlength{\epsfxsize}{5in} \epsfbox{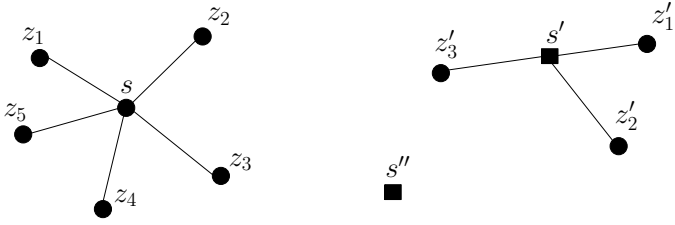}
\end{center}
\end{minipage}
\caption[]{ \label{fig10} \sf A star forest. Safe bags are depicted by squares, while risky ones are depicted by cycles.}
\end{center}
\end{figure*}

Intuitively, Procedure $Attach$ attaches each risky vertex to some
other vertex. Each star of $\Gamma$ will eventually be merged into a
single super-vertex in a
certain supergraph in our algorithm. This will be,
roughly speaking, our way to "get rid" of risky vertices.
It is instructive to view each star center $s$ as an attachment initiator,
and leaves of the  star centered at $s$ as
zombies that will eventually join an appropriate ancestor of $s$ as its step-children.

Procedure $Attach$
starts with forming the attachment digraph $\mathcal G$
%[[S: that's not good, because in Section 5.4
%we use this notation $\cG_j$ to refer to the input graph $G$ of this section]]
as follows:
for every non-isolated vertex $z \in R$,
we pick an arbitrary neighbor $x \in V$ of $z$ in $G$, and insert
the arc $\langle z,x \rangle$ into $\mathcal G$.

It is easy to see that each vertex
$z \in R$ has out-degree one
in $\mathcal G$, and each vertex $s \in V \setminus R$
(in particular, each vertex labeled as safe) has out-degree zero in $\mathcal G$.

Procedure $Attach$ proceeds in two stages.
The first stage is carried out iteratively.
At each iteration the procedure picks an arbitrary non-isolated vertex $z$ in the attachment digraph $\cG$ with in-degree zero,
and handles it as follows.
Since $z$ is non-isolated, it must have an outgoing neighbor $s$; thus $z$ must be labeled as risky.
The procedure removes
the edge $\langle z,s \rangle$ from $\cG$.

Next, suppose that $s$ is the center of some existing star $S'$ in $\Gamma$.
In this case the procedure adds the vertex $z$ as well as the recently removed edge $\langle z,s \rangle$
into the star $S'$. The vertex $z$ is designated as a leaf of $S'$.

Otherwise, $s$ does not belong yet to any star in $\Gamma$.
In this case the procedure forms
a new star $S$  and adds it to $\Gamma$.
It then adds the vertices $z$ and $s$ as well as the recently removed edge $\langle z,s \rangle$ into $S$.
The vertex $s$ is designated as the center of $S$ and the vertex $z$ is designated as a leaf of $S$.
Moreover, if $s$ has an outgoing neighbor $s'$ in $\cG$, the procedure removes
the edge $\langle s,s' \rangle$ from $\cG$. (By removing the edge $\langle s,s' \rangle$ from $\cG$,
we guarantee that $s$ will not be added to any other star in subsequent iterations.)

%If as a result $s$ or $s'$ (or both) become
%isolated, they are removed from $\cG$ as well.
%This process is carried out iteratively, until
The first stage terminates when all the non-isolated vertices in $\cG$
have in-degree at least one.
Let $V'$ be the set of non-isolated vertices in $\cG$ at the end of the first stage,
and denote by $\cG' = \cG[V']$ the subgraph of $\cG$ induced by the vertex set $V'$.
Denote by $E'$ the edge set of $\cG'$.
If $E'$  is empty, then the procedure $Attach$ terminates.
Otherwise, the second stage of the procedure starts.

%Denote by $V'$ the set of non-isolated vertices at the beginning of the second stage.
%Denote by $\cS_1$ the set of edges that are added during the first stage.
%We argue that $\cS_1$ is a star forest.
%At the end of this process, $z$ and $s$ will belong to the same star in $\Gamma$,
%where $z$ is a leaf in this star and $s$ is the center of that star.
Notice that all vertices of $\cG'$ have in-degree at least one, and so $|E'| \ge |V'|$.
Also, the out-degree of each vertex in $\cG'$ is at most one, and so $|E'| = |V'|$.
It follows
that both the in-degree and the out-degree of each vertex of $\cG'$ must be equal to one.
This, in turn, means that all the vertices of $\cG'$ are labeled as risky (i.e., $V' \subseteq R$). Moreover, the graph $\cG'$ is comprised of a collection
 $\cC$ of directed vertex disjoint  cycles. Consider a cycle $C = (v_0,\ldots,v_{g-1},v_g = v_0) \in \cC$, for some positive integer $g \ge 2$.
If $g$ is even then the procedure forms $\frac{g}{2}$ stars $\{\langle v_0,v_1 \rangle\},\ldots,\{\langle v_{g-2},v_{g-1} \rangle\}$,
each containing a single arc. Otherwise (if $g$ is odd), the procedure forms $\frac{g-1}{2}$ stars $\{\langle v_0,v_1 \rangle\},\ldots,\{\langle v_{g-3},v_{g-2} \rangle,
\langle v_{g-1},v_{g-2} \rangle\}$. (Each of these stars except for the last one contains one arc, and the last contains
two arcs. Note that the orientation of the arc $\langle v_{g-2},v_{g-1} \rangle$ gets inverted.)
In both cases each of these $\lfloor \frac{g}{2} \rfloor$ stars is added to $\Gamma$.

Finally, we ignore the orientation of edges, that is, the resulting star forest $\Gamma$ is viewed as an undirected graph.
%[[S: perhaps we should change the description of the procedure accordingly, so that no directions whatsoever are involved]]

This completes the description of the procedure $Attach$.

It is easy to verify that the graph $\Gamma$ constructed by Procedure $Attach$ is a star forest that satisfies the
two conditions listed above. Also, it is straightforward to implement this procedure in time $O(|V|)$.
%We summarize this appendix in the following corollary.
\begin{corollary} \label{attachcor}
Procedure $Attach$, given a graph $G = (V,E)$ whose
vertices are labeled by either safe or risky, produces a star
forest that satisfies conditions 1 and 2 above (listed in the beginning of this section).
The running time of this procedure is $O(|V|)$.
\end{corollary}

\subsection{$j$-level processing} \label{s:jlevel}
The routine that performs $j$-level processing (henceforth, Procedure $Process_{j}$), for $j \in [\ell]$,
accepts as input the forest $\hat \cF$
that was processed by $Process_1, Process_2,\ldots,Process_{j-1}$.
That is, for each $j$-level bag $w \in \cF_{j}$, Procedure $Process_{j}$ accepts as input the sets $B(w),K(w)$
and $Q(w)$, %Also, for each bag $w \in \cF_{j}$, the procedure accepts as input
and the representative $r(w) \in K(w) \subseteq Q$ of $w$. It is also known to
the procedure whether this bag is (labeled as) a zombie or an incubator, and whether this bag is a disappearing zombie or an actual adopter.

%(These marks are not the labels that were discussed above. In particular,
%they are used only within the execution of Procedure $Process_j$.)

Denote by $Q_{j} = \{r(w) ~\vert~ w \in \cF_{j}, Q(w) \ne \emptyset\}$
the set of representatives of the non-empty $j$-level bags.
Observe that %exactly as in Algorithm $BasicLightSp$
%\begin{equation} \label{qjnj2}
$|Q_j| ~\le~ \min\{n,n_j\} =  \min\left\{n,\frac{c\cdot n}{\rho^{j-1}}\right\}$.
%and that the sequence $|Q_1|,|Q_2|,\ldots,|Q_\ell|$ decays geometrically.

Procedure $Process_{j}$ consists of three parts.
\emph{Part I} of Procedure $Process_j$
invokes Algorithm $BasicSp$ for the metric
$M[Q_{j}]$. The algorithm constructs a $t$-spanner $G'_{j} = (Q_{j},E'_{j})$.
%[[S: why not $G'_{j}$ as with algorithm $SimpleLightSp$; this spanner remains unchanged]]
%with small degree and diameter.
It then \emph{prunes} $G'_{j}$, i.e.,
it removes from it all edges $e$ with $\omega(e) > \tau_{j}$. Denote the resulting pruned
graph $G^*_{j} = (Q_{j},E^*_{j})$.
%Henceforth we say that the procedure \emph{prunes}
%the graph $G'_j$ to obtain $G^*_j$.
The edge set $E^*_{j}$
is inserted into the spanner $\tilde G$.
This completes the description of Part I of Procedure $Process_j$.
%[[S: I removed stuff like $E = E^*_{j}$,
%let's write $E^*_{j}$ rather than $E$ (keep subscript $j$)]]
%We remark that this step of the algorithm is very similar to the analogous
%step in Algorithm $SimpleLightSp$. The only difference is that here we
%use the threshold $\tau_{j}$, while in Algorithm $SimpleLightSp$ we
%used the threshold $\tau_{j}$.
%  If $j > \gamma - \ell$, the
%Procedure $Process_{j}$ terminates; in this case the auxiliary
%$j$-level spanner will simply be set as $\tilde G_j = G^*_{j}$.
%The rest of this description is relevant only for $j \le \gamma - \ell$.]]

While the spanner $G'_{j}$ is connected, some points of
$Q_{j}$ may be isolated in $G^*_{j}$. Denote by $Q^*_{j}$
the subset of $Q_{j}$ of all points $q \in Q_{j}$ which are not isolated in $G^*_{j}$.
%Observe that all points of $G^*_{j}$ are labeled as either safe or risky.

If $j \le \ell -\gamma$ then Procedure $Process_{j}$ enters
\emph{Part II}, which is the main ingredient of Procedure $Process_j$.
(Otherwise, Part II is skipped.) We need to introduce some more definitions before proceeding.

A bag $v$ is called \emph{useless} if it is
either empty or a zombie. Otherwise it is called \emph{useful}.

For a bag $v \in \cF_i$, $i \le \ell - \gamma$, its $(i + \gamma)$-level
ancestor $v^{(\gamma)}$
is called the \emph{cage-ancestor} of $v$.
The set of all $i$-level descendants of $v^{(\gamma)}$, denoted $\cC(v)$,
is called the \emph{cage} of $v$.
If $v$ is the only useful bag in its cage, it is called a \emph{lonely} bag;
otherwise it is called a \emph{crowded} bag.

A non-empty bag $v$ is called \emph{safe} if it satisfies at least one of the following
three conditions:
(1) $v$ is large, (2) $v$ is crowded, (3) $v$ is an incubator or a zombie.
Otherwise $v$ is called \emph{risky}.
Note that for $v$ to be risky it must be small (i.e., $|Q(v)| < \ell$),
lonely, and neither an incubator nor a zombie.
A representative $r(v)$ of a safe (respectively, risky; useful; zombie) bag $v$ is called \emph{safe}
(resp., \emph{risky}; \emph{useful}; \emph{zombie}) as well.

Intuitively, for a safe bag $v$ there is no danger that one of the points $p \in Q(v)$ will become overloaded,
i.e., that its degree in the spanner will be too large.
Indeed, if $v$ is a large bag, then it contains many points which can share the load. If $v$ is a crowded bag
or an incubator or a zombie, then it will soon be merged with at least one other non-empty bag $u$, and through this merge it will
acquire additional points that can participate in sharing the load.
If $v$ is a crowded bag, then $u$ is a bag that belongs to $v$'s cage.
%  There may be additional bags that $v$ will be merged with, but never mind...]]
%Otherwise,
%which is neither an incubator nor a zombie
%If $v$ is a crowded bag
%  what do you mean? It's possible
%that $v$ is a crowded bag but also an incubator, and so it will also get merged with bags outside its cage]].
Otherwise $v$ is an incubator or a zombie.
In this case $u$ does not belong to $v$'s cage.
However, in either case, we will argue later that any point in $Q(v)$ is quite close to any point in $Q(u)$ in the original metric $M$.
Moreover, these points will also stay close in the spanner.

Part II of Procedure $Process_{j}$ starts with marking each bag $w \in \cF_{j}$ (and its representative $r(w)$)
as either useful or useless, and as either safe or risky.
Denote by $\hat Q_{j}$ the subset of $Q^*_j$ which contains only useful
representatives;   %, i.e., representatives of useful bags.
note that $\hat Q_{j}$ contains all points of $Q^*_{j}$, except for zombie
representatives.      %, i.e., representatives of zombie bags.)
(See Figure \ref{fig5} for an illustration.)
\begin{figure*}[htp]
\begin{center}
\begin{minipage}{\textwidth} %{5in}
\begin{center}
\setlength{\epsfxsize}{2in} \epsfbox{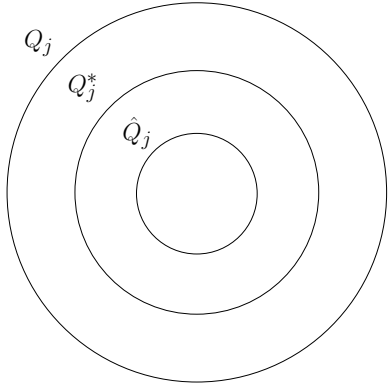}
\end{center}
\end{minipage}
\caption[]{ \label{fig5} \sf Subsets of $Q_j$ which are used during the $j$-level processing.
$Q^*_j$ is the subset of $Q_j$ of points with positive degree in $G^*_j$,
and $\hat Q_j$ is the subset of $Q^*_j$ of non-zombie representatives.}
\end{center}
\end{figure*}
Then it invokes Algorithm $BasicSp$,
%(recall that Algorithm $BasicSp$ was invoked also in the first part of this procedure),
this time with input $M[\hat Q_{j}]$. As a result, a graph $\check G_j = (\hat Q_j, \check E_j)$ is constructed.
Next, it prunes  $\check G_{j}$, i.e.,
it removes from it all edges $e$
 with $\omega(e) > \tau_{j}$.
Denote by $\hat G_{j} = (\hat Q_{j}, \hat E_{j})$
the resulting pruned graph.    % (with small degree and diameter).   %   the diameter may be large, it is not even a spanner]]).
The edge set $\hat E_{j}$ is also inserted into the output spanner $\tilde G$.
Let the \emph{$j$-level auxiliary spanner} $\tilde G_{j} = (Q_{j},\tilde E_{j})$
denote the graph obtained as a union of the graphs $G^*_{j} = (Q_{j},E^*_{j})$
and $\hat G_{j} = (\hat Q_{j},\hat E_{j})$. (In case $j > \ell - \gamma$, we take the $j$-level auxiliary spanner $\tilde G_j$ to be $G^*_{j}$.)
%\end{equation}
%  perhaps we should write that
%its edge set is inserted into the spanner $\tilde G$, instead of inserting the
%edge sets of the two graphs $G^*_{j}$ and $\hat G_{j}$ separately]]
%  What about specifying $\tilde G_{j-1}$ as $G^*_{j-1} \cup \hat G_{j-1}$?]]

Next, %if $j \le \ell -\gamma$   this condition is also relevant for the previous paragraph;
%I mean, the description of the previous paragraph should be carried out only if $j \le \ell -\gamma$]],
 Part II of Procedure $Process_{j}$ constructs the $j$-level \emph{attachment graph}
$G_{j} = (\hat Q_{j},\cE_j)$, which is
the restriction of the $j$-level auxiliary spanner $\tilde G_j$ to the set $\hat Q_j$,
i.e., $\cE_{j} =  \tilde E_j(\hat Q_j) = \hat E_{j} \cup E^*_j(\hat Q_j)$.
 %$G_{j} = (\hat Q_{j},\cE_{j})$ in the following way.
%It inserts into $\cE_{j}$ all edges of $\hat E_{j}$, and also the edges of $E^*_{j}$
%restricted to the point set $\hat Q_{j}$. In other words,
%$\cE_{j} = \hat E_{j} \cup E^*_j(\hat Q_j)$.
(Note that all points of $Q_{j}$, and so  all vertices of $G_j$, are labeled as either safe or risky.)
%Note that all vertices of $G_{j}$ are labeled as either safe or risky.

Part II of Procedure $Process_{j}$ now invokes Procedure $Attach$ on the graph $G_{j} = (\hat Q_{j},\cE_{j})$.
%  in Section 5.1 they are given a graph $G = (V,E)$ as input rather than $\cG$, and
%they transform $G$ into $\cG$ within the Procedure $Attach$; we should change notation...]].
By Corollary \ref{attachcor} (see Section \ref{app:attachalg}), this procedure  returns a star forest
$\Gamma_j$
% = \{S_1,\ldots,S_h\}$
%whose properties are summarized in Section \ref{attachalg}.
that covers $R_j$ (i.e., $R_{j} \subseteq \bigcup_{S \in \Gamma_j} V(S)$), where
$R_{j}\subseteq \hat Q_{j}$ is the set of all risky points in $\hat Q_{j}$.
(Note that no point of $R_j$ is isolated in $G_j$.)
Also, each star $S\in \Gamma_j$ is centered at a center $s \in \hat Q_{j}$,
which is either safe or risky. The star $S$ also contains one or more
leaves $q_1,\ldots,q_k \in \hat Q_{j}$, $k \ge 1$, which are all risky.

%Recall that all points in $\hat Q_{j}$ are representatives of useful $j$-level
%bags. For a point $q \in \hat Q_{j}$, denote by $v_{j}(q)$ the $j$-level bag
%that it represents. ^^^^^^^^^^^^^^^^^^^^^^^^^
%  perhaps also define $v^{(i)}(q)$ as the $i$th ancestor of $v(q)$;
%I think it was more or less already defined before, so perhaps there's no need]]

Intuitively, risky bags cannot be left on their own, because the degrees of their points will inevitably explode. (See Section 1.3.)
Hence the algorithm merges them either with one another, or with some safe bags. 
%For safe bags there is a different restriction--we cannot attach them to other bags, otherwise
%we would obtain ``safe chains''   this is important; this explains why the leaves of stars must be risky]]. The latter would result in an uncontrolled growth %of the diameter.
The attachment graph $G_j$ is used to determine which bags will merge. 
A special care is taken to exclude zombie representatives from $G_j$.
Recall that a zombie bag is already on its way to be merged with some other bag. 
If another bag were merged into a zombie bag,
%(or if a zombie bag would be merged into another
%bag other than its adopter), 
this would ultimately lead to the creation of ``zombie paths'', i.e., paths $(z_1,z_2,\ldots,z_h)$ of zombie bags, where $z_1$ merges into $z_2$,
$\ldots$, $z_{h-1}$ merges into $z_h$.
This would result  in an uncontrolled growth of the diameter.
%  In particular, this is why the attachment graph must contain additional edges besides $E^*_j(\hat Q_j)$;
%because if a non-zombie representative is incident in $E^*_j$ to a zombie representative, it may be isolated in the attachment graph.
%To control the degree of such points, it is imperative to add additional edges to $G_j$ that are not incident to zombie representatives;
%this is important, let's rewrite and add this.]]
%\{(q,q') \in E^*_{j} ~\vert~ q,q' \in \hat Q_{j}\}$.

Next, Part II of Procedure $Process_{j}$ performs attachments. Specifically, for each
star $S \in \Gamma_j$ with center $s$ and leaves $q_1,\ldots,q_k$,
the host bags $v(q_1),\ldots,v(q_k)$ of $q_1,\ldots,q_k$, respectively, are
attached to the host bag $v(s)$ of $s$.
As a result
the parent bag $\pi(v(s)) = v^{(1)}(s)$ of $v(s)$ adopts
%  IMPORTANT: needs rewriting? we want to say
%that the actual adopter (the one in level $j+\gamma$) adopts (the $j$-level ancestors of) these bags, right?
%But I think that this phrasing is all right too (we can use both phrasings without a contradiction;
%first there is a treaty for adoption, but it comes with an incubation period, and at its end
%there is the actual adoption]]
the bags $v(q_1),\ldots,v(q_k)$.
%   but they are at the same level $j$,
%and adoptions as you described it are between $j$-level and $(j+1)$-level bags]].
In other words,
the $\gamma - 1$ immediate ancestors 
$v^{(1)}(s) = \pi(v(s)),\ldots,v^{(\gamma-1)}(s)$ of $v(s)$ (in $\cF$)  are labeled
as incubators,   %  let's label also $v(s)$ itself as an adopter, otherwise it may become risky]],
and the  $\gamma -1$ 
immediate ancestors $v^{(1)}(q_i),\ldots, v^{(\gamma-1)}(q_i)$ of $v(q_i)$ (in $\cF$),
for each $i \in [k]$,
are labeled as zombies. Note that $v^{(\gamma-1)}(q_i)$ is a
disappearing zombie, for each $i \in [k]$, and $v^{(\gamma)}(s)$ is the actual adopter.
%(See the beginning of this section for more details.)
We say that $v(s)$ \emph{performs $k$ attachments}
$\cA(v(s),v(q_1)),\cA(v(s),v(q_2)),\ldots,\cA(v(s),v(q_k))$.
The bag $v(s)$ is the initiator of all these $k$ attachments.
%$\cA(v(s),v(q_i))$, for each index $i \in [k]$.
% where
% denotes the adoption of $v(q_i)$ by $v(s)$   $\pi(v(s))$]],
%Consider also the
The edges $\{(s,q_i) ~\vert~ i \in [k]\}$ that
connect the center $s$ of the star $S$ with the leaves $q_i$ of $S$, $i \in [k]$,
%(Note that $s$ is the representative of $v(s)$, and $q_i$ is the representative of $v(q_i)$, for each index $i\in [k]$.)
%These edges
belong to the attachment graph $G_j = (\hat Q_j,\cE_j)$,
and they are inserted into the auxiliary spanner $\tilde G_j$, and consequently, into the spanner $\tilde{G}$.
For each  $i \in [k]$, we say that the spanner edge $(s,q_i) = (r(v(s)),r(v(q_i)))$
is a
\emph{representing edge} of the attachment $\cA(v(s),v(q_i))$.
%representative $s$ of the center
This completes the description of Part II of Procedure $Process_j$.
Next, if $j \in [\ell-1]$, Procedure $Process_{j}$ moves to \emph{Part III} of Procedure $Process_j$.
(If $j = \ell$, part III is skipped.)
%  I suggest that the wrap-up will be another procedure
%(that comes after $Process_{j-1}$),
%in which we ``prepare'' for the next iteration (just like in my algorithm)]]
Specifically, it computes the sets $\cS(v)$ and $\cJ(v)$ of surviving
children and step-children, respectively, for every bag $v \in \cF_{j+1}$. This is done
according to the set of attachments which were computed in previous
levels. In particular, a child $w$ of $v$ which joins some other
$(j+1)$-level vertex $u, u \ne v$, is excluded from $\cS(v)$. Such a bag $w$
is a disappearing zombie, and a step-child of $u$.
%  $w$ is a step-child of $u$, not a disappearing zombie of $u$;
%there is no such notion: a zombie of $u$]])
Similarly, a bag $z \in \cF_{j}$ with
$\pi(z) \ne v$, which is a  step-child
%and a disappearing zombie of some other $j$-level vertex $y$]]
of $v$, joins the set $\cJ(v)$.
Given the sets $\cS(v)$ and $\cJ(v)$, Part III of Procedure $Process_{j}$
computes the sets $B(v),K(v),Q(v)$ according to the rules specified in Section \ref{sec21}, and computes the representative $r(v)$ of $v$
according to the rules specified in Section \ref{sec24}.
%(See Equations (\ref{twovertast}), (\ref{threevertast}), (\ref{fourvertast}), and (\ref{excmark}) in Section \ref{sec53}.)
%Specifically, it computes sets $B(u),K(u),Q(u)$
%for all $(j+1)$-level bags $u \in \cF_{j+1}$. In addition, it computes
%their representatives $\{r(u) ~\vert~ u \in \cF_{j+1}\}$.
%These computations are performed according to the rules described in
%Section \ref{sec53}. (See Equations (\ref{twovertast}), (\ref{threevertast}), (\ref{fourvertast}), and (\ref{excmark}) therein.)
%  In the previous paragraph you mentioned that vertices are labeled
%as adopters/zombies/disappearing zombies. This is a concise description, but may be good enough. Instead, we can write explicitly that
%every zombie maintains a counter
%that is initialized to the  $\gamma$, and decremented in each iteration. When the counter reaches 1, the zombie
%becomes a disappearing zombie. Similarly for adopters.]]
This completes the description of Part III (the last part) of Procedure $Process_{j}$.
%  changed]]

%for $j \in [\ell-1]$.
%So far we have given the description of  for $j \in [\ell-1]$.
%Recall that Algorithm $LightSp$ invokes Procedure $Process_{j-1}$ in every iteration
%$j$, for $j = 2,3,\ldots,\ell$.
%$Process_1,Process_2,\ldots,Process_\ell$.
%Once the execution of
%Procedure $Process_\ell$ is completed, the algorithm returns the spanner
%$\tilde G$ as the output spanner and terminates.

Observe that for $j \in [\ell - \gamma]$,
all three parts of Procedure $Process_j$ are executed.
Also, for $j \in [\ell-\gamma+1,\ell-1]$,
just Parts I and III of Procedure $Process_j$ are executed, and
Part II is skipped.
Finally, Procedure $Process_\ell$ (i.e., the $\ell$-level processing)
executes just Part I, and
%is a somewhat degenerate variant of the procedure that we have just described.
%Specifically, it only executes the first part of Procedure $Process_j$, and
skips Parts II and III.
%  changed]]
%  Let's remove all that follows, specifically:
%``The execution of Procedure $Process_\ell$ (i.e., the $\ell$-level processing)
%is a somewhat degenerate variant of the procedure that we have just described.
%Specifically, it only executes the first part of Procedure $Process_j$, and
%skips the second and third parts.
%In other words, Procedure $Process_\ell$ invokes Algorithm $BasicSp$ on $M[Q_\ell]$.
%Then it prunes (by removing edges of weight greater than $\tau_\ell$) the
%resulting spanner $G'_\ell$ to obtain a graph $G^*_\ell = (Q_\ell,E^*_\ell)$, and
%defines $\tilde G_\ell = G^*_\ell$ to be the auxiliary $\ell$-level spanner.
%Finally, the procedure inserts the edges of $\tilde G_\ell$ into the output spanner $\tilde G$.''
%]]

%To summarize, Algorithm $LightSp$ invokes Procedure $Process_j$, for $j = 1,2,\ldots,\ell$.
%Once the execution of Procedure $Process_\ell$ is completed, the algorithm
%adds the set $\cB$ of base edges  to the graph $\tilde G$,
%and returns it as the output spanner.

\section{Analysis} \label{section3}
This section is devoted to the analysis of the spanner $\tilde G$ constructed by Algorithm $LightSp$.

Section \ref{zom:app} focuses on properties of zombies and incubators. In Section \ref{app:rig} we analyze the number of edges and  lightness of our spanners, as well as the running time of our algorithm.
In Section \ref{stchdiam} we analyze the stretch and diameter of our spanners, and in Section \ref{deg:app} we study their degree.

The incubation period $\gamma$ is set as $\gamma = c_0 \cdot (\lceil \log_\rho t\rceil + \lceil \log_\rho c\rceil + 1)$, for a  sufficiently large constant $c_0$.
(Recall that $c = \Theta(t/\eps)$, hence $\gamma = \Theta(\log_\rho (t/\eps))$.)

\subsection{Zombies and Incubators} \label{zom:app}
In this section we   prove a few basic properties of
labels (zombies and incubators) used in our algorithm.

A $j$-level bag $v$ is an \emph{attached bag}
if it is unlabeled, and is adopted
during the execution of (part II of) Procedure $Process_j$.
Observe that an attached bag $v$ must be lonely, i.e.,
the cage $\cC(v)$ does not contain any useful bags. In other words,
all the non-empty bags in that cage are labeled as zombies.

When Procedure $Process_j$ creates an attached bag $v$, it labels
$\gamma-1$
of its immediate ancestors in $\cF$ as zombies. %  Note that $v$ is not labeled, though.]].
We remark, however, that Procedure $Process_j$ does not label $v$ itself as a zombie.
Moreover, for $v$ to become an attached bag, it must be unlabeled
at the beginning of the $j$-level processing.
(Recall that the only possible labels are ``incubator'' and ``zombie''. On the other hand, a bag labeled
as an incubator or a zombie is safe, and therefore will not be attached.)
Hence an attached bag
$v$ is \emph{never} labeled by the algorithm.
Thus for any zombie, there is (at least one) path in $\cF$
of hop-distance at most $\gamma-1$ leading down to an attached bag.
(It will be shown in Lemma \ref{tat3} that there exists exactly one such path.)

\begin{lemma} \label{tat}
Fix an arbitrary index $j \in [\ell]$, and let $v$ be a non-empty $j$-level bag.
Then:
%\begin{enumerate}
%\item
\\(1) If $v$ is not labeled as a zombie, then there is a path $\Upsilon_v$ of non-empty bags
which are not labeled as zombies, leading down from
$v$ to some 1-level bag in $\cF$.
%(In particular, this means that there is
%a non-empty $i$-level descendant for $v$ in $\cF$ which is not labeled as a zombie,
%for all $i = 1,2, \ldots, j$.)
\\(2) $v$ cannot be labeled as both a zombie and an incubator.
%\end{enumerate}
\end{lemma}
{\bf Remark:} The second assertion of this lemma implies that the
distinction between useful and useless bags is well-defined.
\begin{proof}
The proof of both assertions of the lemma is by induction on $j$. The basis $j=1$ is trivial.
%\\\emph{Basis: $j=1$.}
%The first assertion holds vacuously.
%To see why the second assertion holds, note that all $1$-level bags
%are unlabeled.
%\\Clearly, $v$ could not have been labeled prior to iteration 1.
%Hence it may be labeled as either a zombie or an incubator
%(not both), and the specific label is
%determined by the
%output of
%by Procedure $Process_1$, but not both.
%(Its label is determined in Procedure $Attach$.)
%The basis: $j=1$
%The first assertion is trivial.
%Next, we prove  the second assertion.
%\\Clearly, $v$ could not have been  labeled prior to iteration 1.
%Hence it will be labeled as either a zombie or an incubator, according to the
%output of Procedure $Attach$   needs rewriting]].
\\\emph{Induction Step: Assume the correctness for all smaller values of $j, j \ge 2$, and prove it for $j$.}
%Let $j \ge 2$, and $v \in \cF_j$.  %prove it for $j$.}
%Suppose that $v$ is labeled as a zobmie. (Otherwise, both assertion
%of the lemma follow immediately.)

We start with proving the first assertion.
Let $v$ be a non-empty $j$-level bag which is not labeled as a zombie.

Suppose for contradiction that all
non-empty children of $v$ in $\cF$ are labeled as zombies.
Since $v$ is not labeled as a zombie, it follows that all its zombie
children are, in fact, disappearing zombies. By construction,
these disappearing zombies become step-children of other $j$-level bags $u, u \ne v$.
Moreover, by the second assertion of the induction hypothesis, none of these disappearing
zombies can be labeled as an incubator. Hence, by construction, $v$ cannot be a step-parent of any $(j-1)$-level bag.
It follows that $v$ is empty, a contradiction.

Therefore, there must be a non-empty child $z$ of $v$ that is not labeled as a zombie.
By the first assertion of the induction hypothesis, there is a path $\Upsilon_z = (z = v_1,\ldots,v_k),k \ge 1$, of non-empty bags which
are not labeled as zombies, leading down from $z$ to some 1-level bag $v_k$ in $\cF$.
The path $\Upsilon_v = (v=v_0,z=v_1,\ldots,v_k) = (v) \circ \Upsilon_z$ obtained by concatenating the singleton path $(v)$ with $\Upsilon_z$ satisfies the conditions of the
first assertion of the lemma.
%\QED
%\end{proof}

Next, we prove the second assertion.
Suppose that $v$ is labeled as a zombie, and
consider any path that leads down to a $j'$-level attached bag $v'$.
%with $i
%  Since the length of this path is at most $\gamma-1$, it follows
%that $j' \ge j - \gamma + 1$; I'm not using this fact]]
Note that $v'$ is a lonely bag, hence
all the non-empty $j'$-level bags in the cage $\cC(v')$ are useless (i.e., they are all labeled as zombies).

% -- they are either empty or zombies.
Suppose for contradiction that $v$ has a non-empty child $z$ that is not labeled as a zombie.
Consider the path $\Upsilon_z$ that is guaranteed by the first assertion of the induction hypothesis. This path contains a $j'$-level
non-empty bag $z'$ which is not labeled as a zombie. However, $z'$ belongs to $\cC(v')$, yielding
a contradiction. See Figure \ref{fig6} for an illustration.
\begin{figure*}[htp]
\begin{center}
\begin{minipage}{\textwidth} %{5in}
\begin{center}
\setlength{\epsfxsize}{2.7in} \epsfbox{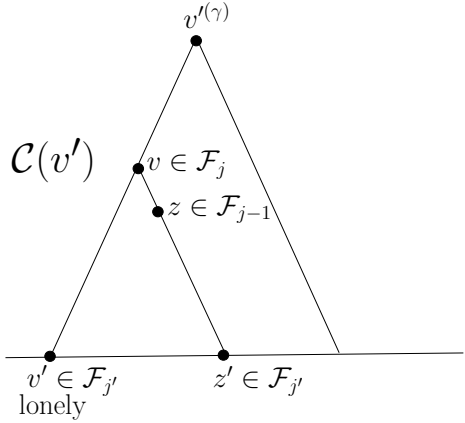}
\end{center}
\end{minipage}
\caption[]{ \label{fig6} \sf $z$ and $z'$ are not zombies, and $v',z' \in \cC(v')$. Hence $v'$ is not
lonely, a contradiction.}
\end{center}
\end{figure*}
Therefore, all the non-empty children of $v$ must be labeled as zombies, and by the induction hypothesis,
they cannot be labeled as incubators. By construction (by the label assignment rules), no child $u$
of $v$ may become the initiator of any attachment (since an attachment initiator cannot be labeled as a zombie).
Thus, $v$ cannot be labeled as an incubator, and we are done.
\QED
\end{proof}

Consider a stagnating bag $v$. If $\cS(v) = \emptyset$, then all its children are disappearing zombies, and thus, by Lemma \ref{tat},
none of them is an incubator. Hence $v$ is not an actual adopter, i.e., $\cJ(v) = \emptyset$ as well. This contradicts the assumption
that $v$ is a stagnating bag (i.e., $|\chi(v)| = 1$). Therefore, if $v$ is a stagnating bag then $|\cS(v)| = 1, \cJ(v) = \emptyset$.

We use the next claim to prove Lemma \ref{tat3}.
\begin{claim} \label{tate1}
Fix an arbitrary index $j \in [\gamma,\ell]$, and let $v$ be a non-empty $j$-level bag.
%Then there is a path $\tilde \Upsilon_v$ of non-empty bags, leading down from
%$v$ to some 1-level bag in $\cF$.
Then there is a useful $(j-\gamma+1)$-level descendant $u$ for $v$ in $\cF$.
\end{claim}
\begin{proof}
%By construction, neither $z'$ nor $u'$ can be zombies (i.e.,
%have zombie children of their own), otherwise
%they will disappear before iteration $j-1$ starts. Hence they are both useful, and
First, we argue that $v$ has a useful $j'$-level descendant $v'$ in $\cF$,
%(at the beginning of iteration $j'$),
for some index $j-\gamma+1 \le j' \le j$.
If $v$ is useful, then we can simply take $v' = v$, $j' = j$.
We henceforth assume that $v$ is a zombie, and consider the path leading down
to an attached $j'$-level bag $v'$. As the hop-distance of this path is at most $\gamma -1$,
we have $j' \ge j-\gamma+1$. By construction, to become an attached bag,
$v'$ must be useful, as required.
%at the beginning of iteration $j'$.

Consequently, Lemma \ref{tat} implies that there is a path $\Upsilon_{v'}$ of useful bags, leading down from $v'$
to some 1-level bag in $\cF$. The claim follows.
\QED
%Consider a useless $(j-\gamma$)-level descendant $u$ for $v$ in $\cF$.
\end{proof}

%  I strengthened the following lemma]]
In the next lemma we show that a zombie cannot have ``brothers'' or ``step-brothers''.	
\begin{lemma} \label{tat3}
Fix an arbitrary index $j \in [\gamma,\ell-1]$, and let $v$ be a non-empty $(j+1)$-level bag.
If $v$ has a zombie child $z$, then all its other children are empty and
it has no step-children, i.e.,  $\cS(v) = \{z\}, \cJ(v) = \emptyset$.
\end{lemma}
\begin{proof}
Suppose for contradiction that $v$ has a non-empty child $u$
in addition to its zombie child $z$.
Both $z$ and $u$ are $j$-level bags.
Set $j' = j-\gamma+1$. Let $z'$ (respectively, $u'$) be a useful $j'$-level descendant of
$z$ (resp., $u$) in $\cF$ that is guaranteed by Claim \ref{tate1}. Observe that the cage-ancestor of $z'$
and $u'$ is $v$, and so
$z'$ and $u'$ belong to the same cage $\cC(z') = \cC(u')$.
It follows that $z'$ and $u'$ are not lonely, and so they are safe and do not
%get adopted. Hence they do not
become attached bags during the $j'$-level processing.
% (and become attached bags) in iteration $j'$.
More generally, note that the least common ancestor of $z'$ and $u'$ in $\cF$ is $v$. Hence, for each index $i = j',j'+1,\ldots,j$,
the $i$-level ancestors of $z'$ and $u'$ in $\cF$ belong to the same cage, and so they are safe and
%do not get adopted, and
do not become attached bags.
Any other useful $i$-level descendant of $v$ is not lonely, and thus it is safe as well,
for each index $i = j',j'+1,\ldots,j$,
%hence it does not get adopted, and
hence it does not become an attached bag.
%in
%all iterations $j',j'+1,\ldots,j-1$, there will be at least two useful bags in the cage that contains
%their respective ancestors in $\cF$. Hence, by construction,
It follows that the $j$-level ancestor $z$ of $z'$ in $\cF$ will not become a zombie, a contradiction.

For the bag $v$ to have a step-child, at least one of the children of $v$ in $\cF$ must be an incubator.
However, we have showed that all children of $v$ besides the zombie $z$ are empty.
%However, we have proved that the only
Hence $\cS(v) = \{z\}, \cJ(v) = \emptyset$.
%  I'll try to simplify this proof tomorrow]]
%  written badly, needs rewriting]]
\QED
\end{proof}

%  $\tilde w^{(2)}$]]
Lemma \ref{tat3} implies the following corollary.
%  The previous lemma became stronger,
%hence the proof of this corollary became trivial and it's omitted]]
\begin{corollary} \label{tat4}
Fix an arbitrary index $j \in [\ell-1]$, and let $v$ be a (non-empty) $j$-level bag
which is a disappearing zombie. (Notice that $j \ge \gamma$.) Then the parent $\pi(v)$ of $v$ in $\cF$ is empty, and
therefore is different than its step-parent $v'$ (in other words, the bag that adopts $v$), i.e., $\pi(v) \ne v'$.
%has a zombie child, then all its other children are empty.
\end{corollary}
%\begin{proof}
%By Lemma \ref{tat3}, $v$ is the only non-empty child of $\pi(v)$.
%By construction, since no child of $\pi(v)$ is an incubator, $\pi(v)$ will not have any step-children.
%It follows that $\pi(v)$ will be empty. However, $v'$ is non-empty, and we are done.\QED
%\end{proof}

Let $w$ be a bag, and $w'$ be an ancestor of $w$ in $\cF$. We say
that $w$ and $w'$ are \emph{identical bags} if $Q(w) = Q(w')$.
%  OK, we remove this remark, but should make sure that we don't use it later:
%``(We remark that identical bags $w$ and $w'$ also satisfy
%$K(w) = K(w')$ and $B(w) = B(w')$.)'']]
%   we can remove the last remark, it is not required for out analysis]]
%This property, however, will not be required for our analysis.
%,K(w) = K(w'),B(w) = B(w')$.
%  I suggest to define that $w$ and $w'$ are identical if $Q(w) = Q(w')$ (that's good enough for our purposes);
%besides, it can be proved that if $Q(w) = Q(w')$ then it must also hold that $B(w) = B(w')$ and $K(w) = K(w')$,
%but that's not the point.]]
%  I changed the statement and proof of the following lemma]]
\begin{lemma} \label{newl}
Let $w \in \cF_j$ be a disappearing zombie. (Hence $j \ge \gamma$.)
Then there exists a unique useful descendant $\tilde w \in \cF_{j-(\gamma-1)}$
of $w$ in $\cF$, and it is an attached bag.
The disappearing zombie $w = \tilde w^{(\gamma-1)}$
is identical to the attached bag $\tilde w = \tilde w^{(0)}$.
%The bag
More generally,
each of the $\gamma-1$ zombie bags ${\tilde w}^{(i)} \in \cF_{j-(\gamma-1)+i}$
%  it's obvious that they are zombies from the construction, no need to prove it]]
along the path between
%the attached bag
$\tilde w^{(1)}$ and
%the disappearing one
$w  = \tilde w^{(\gamma-1)}$, $i \in [\gamma-1]$,
is identical to $\tilde w = \tilde w^{(0)}$.
%is a zombie, and it
%the unique useful descendant of $w$ in $\cF_{j-(\gamma-1)+i}$;
%moreover, ${\tilde w}^{(i)}$ and $\tilde w$ are identical.
(All these $\gamma$ bags are identical.)
\end{lemma}
\begin{proof}
Since $w$ is a disappearing zombie, there exists an attached bag
$\tilde w \in \cF_{j-(\gamma-1)}$, such that $w = {\tilde w}^{(\gamma-1)}$,
i.e., $\tilde w$ is a $(j-(\gamma-1))$-level descendant of $w$ in $\cF$.
For the bag $\tilde w$ to become an attached bag, it must be risky,
and therefore lonely in its cage $C(\tilde w)$. Hence all other
$(j-(\gamma-1))$-level descendants of ${\tilde w}^{(\gamma)} = \pi(w)$
(and therefore, of $w$) are useless. %The fact that $\tilde w$ is identical to itself is obvious.

Next, we prove by induction on the index $i, i \in [0,\gamma-1]$,
that each bag $\tilde w^{(i)}$ is identical to the attached bag $\tilde w = \tilde w^{(0)}$.
%Denote $w = v^{(\gamma-1)}$.
The basis $i=0$ is obvious, as $\tilde w$ is identical to itself.
\\\emph{Induction Step: Assume the correctness of the statement for all smaller values of
$i, i\ge 1$, and prove it for $i$.}
%Write $i' = j - (\gamma-1) + i$.
By  the induction hypothesis, the bag $\tilde w^{(i-1)}$ is identical to $\tilde w$,
i.e., $Q(\tilde w^{(i-1)}) = Q(\tilde w)$.
%, K(\tilde w^{(i-1)}) = K(\tilde w), B(\tilde w^{(i-1)}) = B(\tilde w)$.
Also, since  the bag $\tilde w^{(i-1)}$ is a zombie child of $\tilde w^{(i)}$, Lemma \ref{tat3} yields
$\cS(\tilde w^{(i)}) = \{\tilde w^{(i-1)}\}, \cJ(\tilde w^{(i)}) = \emptyset$.
By construction,
%(See Equations (\ref{twovertast}), , (\ref{fourvertast}), and (\ref{excmark}) in Section \ref{sec53}.)
%(\ref{twovertast})-
$$Q(\tilde w^{(i)}) ~=~ \bigcup_{z \in (\cS(\tilde w^{(i)}) \cup \cJ(\tilde w^{(i)}))} Q(z) ~=~ Q(\tilde w^{(i-1)}) ~=~ Q(\tilde w).$$
%Similarly, Equations (\ref{fourvertast}) and (\ref{excmark}) yield
%$K(\tilde w^{(i)}) ~=~  K(\tilde w^{(i-1)}) ~=~ K(\tilde w)$, and Equation (\ref{twovertast})
%yields $B(\tilde w^{(i)}) ~=~  B(\tilde w^{(i-1)}) ~=~ B(\tilde w)$.
We conclude that the bags $\tilde w^{(i)}$
%is equivalent to the attached bag
and $\tilde w$ are identical. \QED
%and we are done. \QED
%Hence $\cS(\tilde w^{(i)}) = \tilde w^{(i-1)}$. Moreover,
%and thus, it is not useful.
%By the induction hypothesis, the bag $v^{
%The proof procceds
%Next, we prove  by induction on the index ,
%that each bag $\tilde w^{(i)}$ is the unique non-empty descendant
%are identical to $\tilde w = {\tilde w}^{(0)}$
\end{proof}
For a disappearing zombie $w \in \cF_j,j \ge \gamma$, and an
index $i$, such that $j-(\gamma-1) \le i \le j$, we refer to the $i$-level descendant
of $w$ (which is, by Lemma \ref{newl}, identical to $w$) as the
\emph{$i$-level copy} of $w$. We also call it the \emph{$i$-level
copy} of $\tilde w$, where $\tilde w \in \cF_{j-(\gamma-1)}$  is the
unique non-empty $(j-(\gamma-1))$-level descendant of $w$.

\subsection{Number of Edges, Weight, and Running Time} \label{app:rig}
In this section we analyze the number of edges, weight and running time,
of the spanner $\tilde G= (Q,\tilde E)$ computed by Algorithm $LightSp$.
%\subsection{Number of Edges}

\subsubsection{Auxiliary Statements}
We start with providing a few auxiliary lemmas.
They will be used in the analysis of the number of edges,
weight and running time of our construction. %We start with making the following observation.
(See the beginning of Section \ref{section2} for the definitions of $n_j,\tau_j,c,\ell,\rho$ and $L$.)
%\subsection{
\begin{observation} \label{basic}
Let $f$ be a monotone non-decreasing convex function that vanishes at zero,
and let $n'_1,n'_2,\ldots,n'_\ell$ be a sequence of positive numbers
that satisfy that $n'_j \le \min\{n,n_j\}$, for each index $j \in [\ell]$.
%\begin{itemize}
Then for each index
$j \in [\ell]$, $f(n'_j) \le \frac{c}{\rho^{j-1}}\cdot f(n)$.
Moreover,
for each index $1 \le j < \log_\rho c +1$, $f(n'_j) \le  f(n) < \frac{c}{\rho^{j-1}}\cdot f(n)$.
\end{observation}
\begin{proof}
Suppose first that $1 \le j < \log_\rho c +1$; in this case, we have
$\frac{c}{\rho^{j-1}} > 1$,
and so $n < \frac{c \cdot n}{\rho^{j-1}} = n_j$. It follows that $n'_j \le n$,
%, $,
which yields $f(n'_j) \le f(n) < \frac{c}{\rho^{j-1}} \cdot f(n)$.
%SpTm(|Q_j|) \le SpTm(n) < \frac{c}{\rho^{j-1}} \cdot SpTm(n)$.  % < \frac{c}{\rho^{j-1}}\cdot f(n)$.
We henceforth assume that $\log_\rho c +1 \le j \le \ell$.
\\In this case, we have $\frac{c}{\rho^{j-1}} \le 1$, and so
$n \ge \frac{c \cdot n}{\rho^{j-1}}$. It follows that $n'_j \le n_j = \frac{c \cdot n}{\rho^{j-1}}$.
Also, the assumptions about $f$ imply that $f(\frac{c \cdot n}{\rho^{j-1}}) \le \frac{c}{\rho^{j-1}}\cdot f(n)$.
We conclude that $f(n'_j) \le f(\frac{c \cdot n}{\rho^{j-1}}) \le \frac{c}{\rho^{j-1}}\cdot f(n)$.
\QED
\end{proof}

Recall that $|Q_j| ~\le~ \min\{n,n_j\} =  \min\left\{n,\frac{c\cdot n}{\rho^{j-1}}\right\}$.
Observation \ref{basic} implies the following corollary.
\begin{corollary} \label{basic22}
For any monotone non-decreasing convex function $f$ that vanishes at zero:\\
%and any sequence $n'_1,n'_2,\ldots,n'_\ell$ as above:
%it holds that:
%\begin{enumerate}
%\item
(1) $\sum_{j=1}^{\ell} f(|Q_j|) = O(f(n) \cdot \log_\rho (t/\eps))$,
%\sum_{j=1}^{\ell} SpTm(|) = O(SpTm(n) \cdot \log_\rho (t/\eps))$.
and (2) $\sum_{j=1}^{\ell} f(|Q_j|) \cdot \tau_j = O(\frac{f(n)}{n} \cdot \rho \cdot \log_\rho n \cdot t^2/\eps) \cdot L$.
%\end{enumerate}
\end{corollary}
\begin{proof}
We start proving the first assertion.
%\begin{enumerate}
%\item
%We will prove the upper bound for $f$ only. The bound for $g$ can be proved in the same way.
Recall that $c = O(t/{\eps})$. Hence
\begin{eqnarray*}
\sum_{j=1}^{\ell} f(|Q_j|) &=&
%\\ &=& O(n) + f(n) + \sum_{j=1}^{\log_\rho n} f\left(\min\left\{n,\frac{c \cdot n}{k^{j-1}}\right\}\right)
\sum_{1 \le j < \log_\rho c +1, j \in \mathbb N} f(|Q_j|)  + \sum_{\log_\rho c +1 \le j \le \ell, j \in \mathbb N} f(|Q_j|)
 \\ &\le& (\log_\rho c +1) \cdot f(n)  + \sum_{\log_\rho c +1 \le j \le \log_\rho n, j \in \mathbb N} \frac{c}{\rho^{j-1}}\cdot f(n)
 \\ &\le& (\log_\rho c +1) \cdot f(n)  + \sum_{j=0}^{\infty} \frac{1}{\rho^{j}}\cdot f(n)
  %\\ &\le&  (\log_\rho c +1) \cdot f(n) + 2\cdot f(n)
  ~=~ O(f(n) \cdot \log_\rho (t/\eps)).
\end{eqnarray*}

Next, we prove the second assertion.
 For each  $j \in [\ell]$,
  $\tau_j = 2 \cdot \rho^j \cdot \frac{L}{n} \cdot t \cdot (1+\frac{1}{c}) = O(\rho^j \cdot \frac{L}{n} \cdot t)$.
Hence,
\begin{eqnarray*}
\sum_{j=1}^{\ell} f(|Q_j|) \cdot \tau_j &\le&
  \sum_{j=1}^{\ell}  \frac{c}{\rho^{j-1}}\cdot f(n) \cdot O\left(\rho^j \cdot \frac{L}{n} \cdot t\right)
%~=~ \sum_{j=1}^{\ell}  O\left(\rho \cdot \frac{f(n)}{n} \cdot c \cdot t \right) \cdot L
~=~ O\left(\frac{f(n)}{n} \cdot \rho \cdot \log_\rho n \cdot t^2/\eps \right) \cdot L. \inQED
%\\ &=& O\left(\frac{f(n)}{n} \cdot \min\{\Lambda(n),\log n\} \cdot n^{1/\Lambda(n)}\right) \cdot \omega(MST(M[])). \inQED
\end{eqnarray*}
%sdf
%\end{enumerate}
%\QED
\end{proof}

%Hence, Corollary \ref{basic22} yields:
%\begin{corollary} \label{basic2}
%For any monotone non-decreasing convex function $f$ that vanishes at zero,
%we have
%(1) $\sum_{j=1}^{\ell} f(|Q_j|) = O(f(n) \cdot \log_\rho (t/\eps))$, and
%(2) $\sum_{j=1}^{\ell} f(|Q_j|) \cdot \tau_j = O(\frac{f(n)}{n} \cdot \rho \cdot \log_\rho n \cdot t^2/\eps) \cdot L$.
%\end{corollary}

\subsubsection{Number of Edges}
In this section we bound the number of edges in $\tilde G$.
%As shown in Section \ref{disregard}, the base edge set $\cB$ contains at most $n$ edges.
%there are at most $n$ base edges in $G^*$,
%Hence we may disregard the base edge set in this analysis.
%Write $\tilde G = (Q,\tilde E)$.
\begin{lemma} \label{edgebound2}
$|\tilde E| = O(SpSz(n) \cdot \log_\rho (t/\eps))$.
\end{lemma}
\begin{proof}
By construction, the edge set $\tilde E$ of $\tilde G$ is  the union of
the path-spanner $H = (Q,E_H)$, the base edge set $\cB$, and all the $j$-level auxiliary spanners $\tilde G_j = (Q_j,\tilde E_j)$, $j \in [0,\ell]$,
i.e., $\tilde E ~=~ E_{H} \cup \cB \cup \bigcup_{j=0}^\ell \tilde E_j.$
%\tilde E_0 \cup \tilde E_1 \cup \ldots \cup \tilde E_\ell$.
%(The edge sets $E_{H}$ and $E''_0 = \$ are also present in the spanner construction $G''$ of Section \ref{first}.)

%Recall that the edge set $E_H$ of the path-spanner $H$ and the edge set $\tilde E_0 = E''_0$
%of the graph  $\tilde G_0 = G''_0$ also belong to the spanner that is computed by Algorithm $Sp$.
The path-spanner $H$ contains at most $O(n)$ edges, i.e., $|E_{H}| = O(n)$,
and the graph $\tilde G_0 = (Q_0 = Q,\tilde E_0)$ contains at most $SpSz(n)$ edges, i.e.,
$|\tilde E_0| \le SpSz(|Q_0|) = SpSz(n)$.
%Also, we have $|E''_0| \le SpSz(|Q_0|) = SpSz(n)$.
%Finally, note that  $|E_{H}| = O(n)$.
Also, as shown in Section \ref{disregard} (see Corollary \ref{degbase}), the base edge set $\cB$ contains at most $n$ edges.

For each   $j \in [\ell-\gamma]$, we have $\tilde G_j = G^*_j \cup \hat G_j$,
where $G^*_j = (Q_j,E^*_j)$ and $\hat G_j = (\hat Q_j,\hat E_j)$.
Also, for each   $j \in [\ell-\gamma+1,\ell]$, we have $\tilde G_j = G^*_j = (Q_j,E^*_j)$.
%where $G^*_j = (Q_j,E^*_j)$ and $\hat G_j = (\hat Q_j,\hat E_j)$.
%(E^*_1 \cup \hat E_1) \cup \ldots \cup (E^*_{\ell}\cup \hat E_{\ell})$.
%On the other hand, the edge sets $E^*_1,\ldots,E^*_{\log_\rho n}$ do not belong $G''$;
%instead, each edge set $E^*_i$ is a modification of the edge set $E''_i$ from $G''$.)
Observe that $|E^*_j| \le SpSz(|Q_j|)$, for every   $j \in [\ell]$.
We also have $|\hat E_j| \le SpSz(|\hat Q_j|) \le SpSz(|Q_j|)$, for every   $j \in [\ell-\gamma]$.
It follows that $|\tilde E_j| = |E^*_j \cup \hat E_j| \le |E^*_j| + |\hat E_j| \le 2\cdot SpSz(|Q_j|)$,
for every index $j \in [\ell-\gamma]$,
and $|\tilde E_j| = |E^*_j| \le SpSz(|Q_j|)$, for every index $j \in [\ell-\gamma+1,\ell]$.
Finally, recall that $SpSz(\cdot)$ is a monotone non-decreasing convex function that vanishes at zero.
%Also, we have $|E''_0| \le SpSz(|Q_0|) = SpSz(n)$.
%Finally, note that  $|E_{H}| = O(n)$.
Consequently,
\begin{eqnarray*}
|\tilde E| &=& |E_{H}| + |\cB| + \sum_{j=0}^{\ell} |\tilde E_j|
%|\tilde E_0| + |\tilde E_1| + \ldots + |\tilde E_{\ell}|
~=~ |E_{H}| + |\cB| + |\tilde E_0| + \sum_{j=1}^{\ell-\gamma} |E^*_j \cup \hat E_j| +
\sum_{j=\ell-\gamma+1}^\ell  |E^*_{j}|
%\\ &=& |E_{H}| + |E''_0| + |E''_1| + \ldots + |E''_{\log_\rho n}|
%~\le~ O(n) + SpSz(|Q_0|) + \sum_{j=1}^{\ell} 2SpSz(|Q_j|)
\\ &\le& O(n) + SpSz(n) + 2 \cdot \sum_{j=1}^{\ell} SpSz(|Q_j|)
%\\ &=& O(n) + SpSz(n) + \sum_{j=1}^{\log_\rho n} f\left(\min\left\{n,\frac{c \cdot n}{\rho^{j-1}}\right\}\right)
~\le~ O(n) + SpSz(n) + O(SpSz(n) \cdot \log_\rho (t/\eps)) \\ &=& O(SpSz(n) \cdot \log_\rho (t/\eps)).
\end{eqnarray*}
(The last inequality follows from the first assertion of Corollary \ref{basic22}.) \QED
%We used the fact that
\end{proof}

\subsubsection{Weight} \label{subsecwt}
In this section we bound the weight of $\tilde G$.
%As shown in Section \ref{disregard}, the total weight
%of the base edges of $\tilde G$ is bounded above by $\log_\rho n \cdot L$,
%and so we may disregard the base edges in this analysis.
\begin{lemma} \label{wtver2}
%$\omega(\tilde G) = O(\frac{SpSz(n)}{n} \cdot \rho \cdot \log_\rho n \cdot t^2/\eps) \cdot L$.
$\omega(\tilde G) = O(\frac{SpSz(n)}{n} \cdot \rho \cdot \log_\rho n \cdot t^2/\eps) \cdot L$.
%\omega(MST(M[Q]))$.
%O\left(\frac{SpSz(n)}{n} \cdot \min\{\Lambda(n),\log n\} \cdot n^{1/\Lambda(n)}\right) \cdot \omega(MST(M[]))$.
\end{lemma}
%is a spanner for $G$ with the desired properties.
%The set $Q_j$ is referred to as the
%set of \emph{representatives} of the $j$th iteration.
\begin{proof}
First, note that the weight $\omega(H)$ of the path-spanner $H$ satisfies $\omega(H) = O(\rho \cdot \log_\rho n) \cdot L$.
As shown in Section \ref{disregard} (see Corollary \ref{weightbase}),
the weight $\omega(\cB)$ of the base edge set $\cB$ satisfies $\omega(\cB) = O(\log_\rho n) \cdot L$.
Also, observe that the maximum edge weight in the graph $\tilde G_0$
is at most $\tau_0$, and so
$$\omega(\tilde G_0) ~\le~ |\tilde E_0| \cdot \tau_0 ~\le~ SpSz(|Q_0|) \cdot \tau_0 ~=~
 SpSz(n) \cdot 2 \cdot \frac{L}{n} \cdot t \cdot \left(1+\frac{1}{c}\right)
~=~
 SpSz(n) \cdot O\left(\frac{L}{n} \cdot t\right).$$
%Finally, note that     % and $\omega(E''_0) \le SpSz(|Q_0|) \cdot \tau_0$.
%Recall that $k = \lceil n^{1/\Lambda(n)} \rceil$, implying that $k \le n^{1/\Lambda(n)} + 1, \log_\rho n \le \min\{\Lambda(n),\log n\}$
%Observe that for each index $j \in [\ell-\gamma]$,
Next, observe that
the maximum edge weight in the graph $G^*_j$ (for every index $j \in [\ell]$) and
the graph $\hat G_j$ (for every index $j \in [\ell-\gamma]$)
is bounded above by the $j$-level threshold $\tau_j$.
% ~=~  \cdot \rho^j \cdot \frac{L}{n} \cdot t \cdot (1+\frac{1}{c}) ~=~ O(\rho^j \cdot \frac{L}{n} \cdot t).$$
% (The last equality follows since $\ \le 1$.)    %; see the third paragraph of Section \ref{s:anal}).
In other words,  for every index $j \in [\ell]$, the maximum edge weight in the graph $\tilde G_j$ is bounded above by $\tau_j$,
and so $$\omega(\tilde G_j) ~\le~ |\tilde E_j| \cdot \tau_j ~\le~ 2 \cdot SpSz(|Q_j|) \cdot \tau_j.$$
%$\omega(\tilde G_j) ~=~ \omega(G^*_j) + \omega(\hat G_j) \le (|E^*_j| + |\hat E_j|) \cdot \tau_j
Finally, we have $\omega(\tilde G) ~=~ \omega(H) + \omega(\cB)+  \omega(\tilde G_0) + \sum_{j=1}^\ell \omega(\tilde G_j)$.
It follows that
\begin{eqnarray*}
%\omega(G'') &=& \omega(E'') ~=~ \omega(E_{H}) + \omega(E''_0) + \omega(E''_1) + \ldots + \omega(E''_{\log_\rho n}) \\
\omega(\tilde G) &\le&
%\omega(H) + \omega(\cB)+
%\sum_{j=0}^\ell \omega(\tilde G_j) ~=~
%\omega(H) + \omega(\cB)+  \omega(\tilde G_0) + \sum_{j=1}^\ell \omega(\tilde G_j) \\
%\omega(G''_0) + (\omega(G^*_1) + \omega(\hat G_1)) + \ldots + (\omega(G^*_{\ell}) + \omega(\hat G_{\ell})) \\
%&=& \omega(H) + \omega(\cB)+\omega(G''_0) + (\omega(G^*_1) + \omega(\hat G_1)) + \ldots + (\omega(G^*_{\ell}) + \omega(\hat G_{\ell})) \\
%&\le& O(\rho \cdot \log_\rho n) \cdot L + |E''_0| \cdot \tau_0 + \sum_{j=1}^{\log_\rho n} |E''_j| \cdot \tau_j
%\\ &\le&
O(\rho \cdot \log_\rho n) \cdot L + O(\log_\rho n) \cdot L + SpSz(n) \cdot O\left(\frac{L}{n} \cdot t \right) + \sum_{j=1}^{\ell} 2\cdot SpSz(|Q_j|) \cdot \tau_j
%\\ &\le& O(\rho \cdot \log_\rho n) \cdot L + SpSz(n) \cdot O\left(\frac{L}{n} \cdot t\right) + 2\cdo \cdot \sum_{j=1}^{\ell}  SpSz(|Q_j|) \cdot  %\tau_j
\\ &\le& O(\rho \cdot \log_\rho n) \cdot L + SpSz(n) \cdot O\left(\frac{L}{n} \cdot t\right) +  O\left(\frac{SpSz(n)}{n} \cdot \rho \cdot \log_\rho n \cdot t^2/\eps\right) \cdot L
%\\ &\le& O(\rho \cdot \log_\rho n) \cdot L + SpSz(n) \cdot O\left(\frac{L}{n}\right) + \sum_{j=1}^{\log_\rho n}  \frac{c}{\rho^{j-1}}\cdot SpSz(n) \cdot O\left(\rho^j \cdot %\frac{L}{n}\right)
%\\ &=& O(\rho \cdot \log_\rho n) \cdot L + O\left(\frac{SpSz(n)}{n} \cdot t \right) \cdot L + \sum_{j=1}^{\log_\rho n}  O\left(\rho \cdot \frac{SpSz(n)}{n} \cdot t \cdot c %\right) \cdot L
\\ &=& O\left(\frac{SpSz(n)}{n}
\cdot \rho \cdot \log_\rho n \cdot t^2/\eps \right) \cdot L.
%O\left(\rho \cdot \log_\rho n + \frac{SpSz(n)}{n} \cdot t + \frac{SpSz(n)}{n} \cdot \rho \cdot \log_\rho n \cdot t^2/\eps \right) \cdot L
%~=~
%\\ &=& O\left(\frac{SpSz(n)}{n} \cdot \min\{\Lambda(n),\log n\} \cdot n^{1/\Lambda(n)}\right) \cdot \omega(MST(M[Q])). \inQED
\end{eqnarray*}
(The last inequality follows from the second assertion of Corollary \ref{basic22}.)
\ignore{
Finally,  we bound the weight $L$ of the path $\cL$.
Recall that Algorithm $LightSp$ starts with computing an approximate MST $T$ for $M[Q]$.
%Recall that $L \le 2 \cdot \omega(T)$, where $T$ is some light-weight spanning tree of $M[Q]$.
The tree $T$ is built in the following way.

First, we use Algorithm $BasicSp$ to build a $t$-spanner
$\cH$ for $M[Q]$.
% with $SpSz(n)$ edges. Then
Then we run Prim's  Algorithm for $\cH$ to get a $t$-approximate MST for $M[Q]$.
%  in total
%$O(SpTm(n) + SpSz(n) \cdot \log n)$ time.
It follows that $\omega(T) \le t \cdot \omega(MST(M[Q]))$.

Recall that $L \le 2 \cdot \omega(T)$, which implies that
$L = O(t) \cdot \omega(MST(M[Q]))$, and we are done.
}
\QED
\end{proof}

In Theorem \ref{ourresult} we assumed the existence of Algorithm $BasicSp$ that constructs a $t$-spanner
for any sub-metric $M[Q]$ of $M$ (including $M$ itself) with certain properties. This algorithm can be used to construct
a $t$-approximate MST for $M$. Specifically, running Prim's MST Algorithm over this $t$-spanner results in a $t$-approximate
MST. The weight $L = \omega(\cL)$ of the Hamiltonian path $\cL$ which is computed in this way is $O(t \cdot \omega(MST(M[Q])))$,
and the running time of this computation is $SpTm(n) + O(SpSz(n) + n \cdot \log n) = O(SpTm(n) + n \cdot \log n)$.
Alternatively, we may add to the statement of  Theorem \ref{ourresult} an assumption that it is provided with
Algorithm $LightTree$, which computes an $O(1)$-approximate MST for $M[Q]$ within time $TrTm(n)$.
In particular, in low-dimensional Euclidean and doubling metrics running Prim's Algorithm over an $O(1)$-spanner
results in a routine that computes an $O(1)$-approximate MST within $O(n \cdot \log n)$ time. In these cases (i.e., if we are
provided with Algorithm $LightTree$ or if the input metric is doubling), $L = \omega(\cL) = O(\omega(MST(M[Q])))$.
% I think this paragraph (at least parts of it) is somewhat redundant--my impression is that we already discussed parts of it before]]

To summarize:
\begin{corollary} \label{wtcorver2}
In the variant of Algorithm $LightSp$ which employs
Algorithm $LightTree$ (or if $M$ is a low-dimensional Euclidean or doubling metric), $\omega(\tilde G) = O(\frac{SpSz(n)}{n} \cdot \rho \cdot \log_\rho n \cdot t^2/\eps) \cdot \omega(MST(M[Q]))$.
In the variant of Algorithm $LightSp$ which does not employ it, $\omega(\tilde G) = O(\frac{SpSz(n)}{n} \cdot \rho \cdot \log_\rho n \cdot t^3/\eps) \cdot \omega(MST(M[Q]))$.
\end{corollary}
%\begin{proof}
%Recall that in the first case $L = O(\omega(MST(M[Q])))$, while in the second case $L \le t \cdot O(\omega(MST(M[Q])))$.
%The corollary now follows from Lemma \ref{wtver2}. \QED
%\end{proof}

\subsubsection{Running Time} \label{time:app}
In this section we analyze the running time of Algorithm $LightSp$.
%  I wrote this section (the proof of Lemma \ref{time2}) in a rush, but I think it is OK.
%I plan to go over the proof of this lemma (and polish it, if needed) while you read the rest of the text]]

%  I made some changes to this lemma, please check]]
\begin{lemma} \label{time2}
The variant of Algorithm $LightSp$ that invokes (respectively, does not invoke)
Algorithm $LightTree$ can be implemented in time $O(SpTm(n) \cdot \log_\rho(t/\eps) +  TrTm(n))$
(resp., $O(SpTm(n) \cdot \log_\rho(t/\eps) + n \cdot \log n$)).
\end{lemma}
\begin{proof}
The tree $T$ can be built within $O(TrTm(n))$ time by Algorithm $LightTree$,
and within $O(SpTm(n) + n \cdot \log n)$ without it.
In the former case $T$ will have constant lightness. In the latter a $t$-approximate MST $T$ for $M[Q]$
is built as outlined in Section \ref{subsecwt}.
%as follows. First we invoke Algorithm $BasicSp$ to build a $t$-spanner
%$\cH$ for $M[Q]$ with $SpSz(n)$ edges in time $O(SpTm(n))$. Then,
%we run Prim's  Algorithm over $\cH$ to get $T$ in another time
%$O(SpSz(n) + n \cdot \log n) = O(SpTm(n) + n \cdot \log n)$.
In another $O(n) = O(TrTm(n))$ time we can compute the preorder traversal of $T$,
thus obtaining the Hamiltonian path $\cL$.
The 1-spanner $H_\cL$ can be built in $O(n)$ time. The path-spanner $H$ can be obtained from $H_\cL$ in another $O(n)$ time.
Also, it is easy to see that the graph $\tilde G_0$ can be built within $O(SpTm(n))$ time.

%The Hamiltonian path $\cL$ is built in $O(n)$ time
%by computing the preorder traversal of $T$.
%The 1-spanner $H_\cL$ can be built in $O(n)$ time. The path-spanner $H$ can be obtained from $H_\cL$ in another $O(n)$ time.
%Also, it is easy to see that the 0-level auxiliary spanner $\tilde G_0$ can be built within $O(SpTm(n))$ time.
%In each iteration $j=1,2,\ldots,\ell$, the set $Q_j$ can be built in linear time in the obvious way,
%hence the time needed to build all $\ell$ sets $Q_1,Q_2,\ldots,Q_{\ell}$ is $O(n \log_\rho n)$.
%\end{lemma}
%\begin{proof}
%The Hamiltonian path $\cL$ can be built within $O(TrTm(n))$ time. Indeed, the spanning tree $T$ is built in $TrTm(n)$ time using
%Algorithm $LightTree$, and another $O(n) = O(TrTm(n))$ time is needed to compute the preorder traversal of $T$.
%The 1-spanner $H_\cL$ can be built in $O(n)$ time. The path-spanner $H$ can be obtained from $H_\cL$ in another $O(n)$ time.

In each level $j = 1,\ldots,\ell$, we spend $SpTm(|Q_j|)$ time to build the $t$-spanner $G'_j$.
%By going over the edges of $G'_j$ and removing those of weight greater than $2 \cdot \tau_j$, we obtain the
Then Algorithm $LightSp$ prunes $G'_j$ to obtain the
graph $G^*_j$.
%this can be carried out in at most $O(SpSz(|Q_j|))$ time, which is clearly bounded above by $O(SpTm(|Q_j|)$ time.
Since there are at most $SpSz(|Q_j|)$ edges in $G'_j$ and $SpSz(|Q_j|) \le SpTm(|Q_j|)$,
%Clearly, $SpSz(m) \le SpTm(m)$, for every $m \ge 1$, and so
the graph $G^*_j$ can be obtained from $G'_j$ in $O(SpSz(|Q_j|)) = O(SpTm(|Q_j|))$ time  (see Part I of Procedure
$Process_j$ in Section \ref{s:jlevel}).
By similar considerations, the graph $\hat G_j$ can be built in time $O(SpTm(|\hat Q_j|)) = O(SpTm(|Q_j|))$;
also, the $j$-level attachment graph $G_j$ can be built within another amount of $O(SpSz(|Q_j|)) = O(SpTm(|Q_j|))$ time (see Part II of Procedure
$Process_j$).
Updating the load indicators and counters of representatives  requires another $O(SpSz(|Q_j|)) = O(SpTm(|Q_j|))$ time.
By Corollary \ref{attachcor} (see Section \ref{app:attachalg}), Procedure $Attach$ runs in time that is linear in the number of vertices of the attachment graph
$G_j$,
i.e., in time $O(|Q_j|) = O(SpTm(|Q_j|))$.
The number of non-empty bags in the forest $\cF$ is $O(n)$.
%  More accurately, it is $O(n_1) = O(n \cdot c) = O(n \cdot (t/\eps))$.]]
Each time one of these bags is processed, at most $O(\gamma)$
bags are labeled as zombies or incubators. Hence the total time
required for labeling bags is $O(n \cdot \gamma) = O(n \cdot \log_\rho (t/\eps))$.
%; setting the zombie-time of all the (necessarily new) zombie bags to $\gamma$
%as well as (re)setting the leader-time of all the (possibly old) leader bags to $\gamma$ requires another $O(SpSz(|Q_j|)) = O(SpTm(|Q_j|))$ time.
%Indeed,
%The attachment graph of iteration $j$ can be built in linear time, i.e., in time that is linear in the number
%of edges in
%it takes $g_2(|Q_j|) \le SpTm(|Q_j|)$ to build the $t$-spanner of Algorithm $\mathcal A$ in step 3 of the algorithm,
%and another $O(f_2(|Q_j|)) = O(SpSz(|Q_j|)) = O(SpTm(|Q_j|))$ time is required to remove the edges of weight greater than $2 \cdot \tau_j$.
%then we remove the
%Since there are at most $f_2(|Q_j|) \le SpSz(|Q_j|)$ edges in this spanner and $SpSz(|Q_j|) \le SpTm(|Q_j|)$,
%it follows that the attachment
%Similarly,
It follows that the time needed to build all $2\ell$ graphs $G^*_1,\hat G_1,\ldots,G^*_{\ell},\hat G_{\ell}$
(and updating the load indicators and counters of the involved representatives accordingly),
as well as executing Procedure $Attach$ and labeling bags
%performing the attachment algorithm
%and signing the resulting treaties for attachments (and updating the zombie-time and leader-time of bags)
throughout all $\ell$ levels
%$1,2,\ldots,\ell$
is at most
%\begin{eqnarray*}
$\sum_{j=1}^{\ell} O(SpTm(|Q_j|)) + O(n \cdot \log_\rho (t/\eps))
%&=& O(SpTm(n)) + \sum_{j=1}^{\log_\rho n} O\left(g\left(\frac{c n}{k^{j-1}}\right)\right)
~\le~ O(SpTm(n)\cdot \log_\rho (t/\eps)).$
% ~=~ O(SpTm(n)\cdot \log_\rho (t/\eps)).
%\end{eqnarray*}
(Recall that $SpTm(\cdot)$ is a monotone non-decreasing convex function that vanishes at zero,
and see the first assertion of Corollary \ref{basic22}.)

On level $j$, determining which bags are crowded requires
$O(|Q_j|)$ time. Hence, by Corollary \ref{basic22}, in all $\ell$
levels altogether this step requires $\sum_{i=1}^\ell O(|Q_j|) =
O(n \cdot \log_\rho (t/\eps))$ time.
%Hence the overall running time of this step is
%$O(SpTm(n) \cdot \log_\rho (t/\eps))$.

%Altogether, Algorithm $LightSp$ takes  time
%$O(TrTm(n)+n \log_\rho n+SpTm(n) \cdot \log_\rho (t/\eps))$,
%O(\max\{SpTm(n)\cdot \log_\rho (t/\eps),TrTm(n),n\cdot \log_\rho n\})$,
%and we are done.
%Altogether,  Algorithm $LightSp$ takes  time $O(SpTm(n)\cdot \log_\rho(t/\eps) + SpSz(n) \cdot  \log n)$.
Altogether, the variant of Algorithm $LightSp$ that uses Algorithm $LightTree$ takes  time $O(SpTm(n)\cdot \log_\rho(t/\eps) + TrTm(n))$. The variant of Algorithm $LightSp$ that does not use Algorithm $LightTree$ takes  time $O(SpTm(n)\cdot \log_\rho(t/\eps) + n \cdot  \log n)$.
\QED
\end{proof}
{\bf Remark:} In the case of low-dimensional Euclidean or doubling metrics, $TrTm(n) = O(n \cdot \log n)$,
and so the running time becomes $O(SpTm(n) \cdot \log_\rho (t/\eps) + n \cdot \log n)$.

%Also,  set $\ = \frac{1}{\rho \cdot t}$,
%and define $\ = 2\z$.
%Note that $\ = \frac{1}{\rho \cdot t} \le 1/2$, hence $\ =2\a \le 1$.
%Also, we define $c = \lceil \frac{2\cdot(1+\) \cdot (t+1)}{\eps} \rceil$
%(instead of $c = \lceil \frac{2\cdot(t+1)}{\eps} \rceil$ as in Section \ref{first}).   % which gives $\eps \ge \frac{(1+) \cdot 2(t+1)}{c}$.
%Note that $(1+\) \le 2 < c$.

\subsection{Stretch and Diameter} \label{stchdiam}
In this section we analyze the stretch and diameter of the spanner $\tilde G$.

The following lemma is central in our analysis. 
%It is used not only for bounding the stretch (and the diameter)
%of the spanner, but also for bounding its maximum degree.
%For a bag $v$ in $\mathcal F_j$, we denote by $l(v) = \mu_j$
%the length of the base interval $I(v)$ of $v$.
%If $v$ is a $j$-level bag, then we have $l(v) = \mu_j = \frac{\xi_j}{c}$.
%  I changed the statement of this lemma and the proof; it had two assertions, now just one]]
%  You suggested to remove the second assertion from the statement of this lemma -- because we don't need it
%outside of this lemma, and we don't need to prove it inductively. You're right. I'll apply this change...]]
\begin{lemma} \label{key}
Fix any index $j \in [\ell]$, and let $v \in \cF_j, Q(v) \ne \emptyset$.
Then
%\begin{enumerate}
%\item
for every $p \in Q(v)$, there is a path $\Pi_j(p)$ in the spanner $\tilde G$ that
leads to a point $b_j(p)$ in the base point set $B(v)$ of $v$, having
weight at most $\frac{1}{2} \cdot \mu_j$ and at most $3\ell$ edges.
Moreover, if $p \in K(v)$, then $\Pi_j(p)$ consists of at most $2\ell$ edges.
All points in $\Pi_j(p)$ belong to the point set $Q(v)$ of $v$.
%(Hence $\Pi_j(p)$ consists of at most $|Q(v)|-1$ edges. In particular, if
%$v$ is a small bag, then $\Pi_j(p)$ consists of at most $\ell -2$ edges.)
(The point $b_j(p)$ is called the $j$-level \emph{base point} of $p$.)
%\item
%There is a (simple) path $\Pi_j(p,q)$ in the spanner $\tilde G$ between every pair $p,q$ of points in $Q(v)$,
%having weight at most $(1 + 2\) \cdot )$.
%All points of $\Pi_j(p,q)$ belong to the point set $Q(v)$ of $v$. (Hence $\Pi_j(p,q)$ consists of at most $|Q(v)|-1$ edges;
%in particular, if $v$ is a small bag, then $\Pi_j(p,q)$ consists of at most $\ell -2$ edges.)
%\end{enumerate}
\end{lemma}
\begin{proof}
The proof is by induction on $j$. The
%We prove the two assertions of the lemma simultaneously by induction on $j$.
basis $j=1$ is immediate.
%.} For a $1$-level bag $v$, we have $B(v) = K(v) = Q(v)$. In this case we can simply take
%$b_1(p) = p, \Pi_1(p) = \emptyset$.
\\\emph{Induction Step:}
%We assume that the statement holds for all smaller values of $j, j \ge 2$, and prove it for $j$.}
Let $j \ge 2$, and $p \in Q(v)$.
%We start with proving the first assertion.
%Let $p$ be an arbitrary point of $Q(v)$, and recall that $Q(v) = \bigcup_{z \in \chi(v)} Q(u)$.
Also, let $u \in \chi(v) \subseteq \cF_{j-1}$ be the $(j-1)$-level host bag of $p$.
% i.e., the unique
%Thus $p$ belongs to $Q(u)$, for some
%$(j-1)$-level bag such that $p \in Q(u)$.
%let $u_1,\ldots,u_g$ be
%the non-empty $(j-1)$-descendants of $v$; thus $Q(v) = \bigcup_{i=1}^g Q(u_i)$, and $p$ belongs to
%some $Q(u_i)$, $i \in [g]$.   a side remark: if one of the children of $v$ is a zombie (or a disappearing zombie),
%this means (as was discussed before) that all other $\rho-1$ children of $v$ are empty; in fact, the only zombie child of $v$ cannot be
%a disappearing one because this would mean that $v$ is empty, and then $p$ would not belong to $v$.
%Thus if one of the children of $v$ is a zombie, then its zombie-time is at least 2, and $v$ would be a zombie as well,
%with the same point set, kernel and base point set as those of
% its zombie child (in this case $v$ will not have any
%disappearing zombies attached to it).
%Otherwise, all the children of $v$ in $\mathcal F$ are either empty or useful, and some of the $(j-1)$-descendants
%of $v$ may be children of some other $j$-level bag in $\mathcal F$ (in this case these descendants are disappearing
%zombies that are attached to $v$, and $v$ is their leader; note that this case can occur only if $j \ge \gamma+1$)]]

Suppose first that $u \in \cS(v)$, i.e., $u$ is a surviving child of $v$ in $\cF$.
%is a child of $v$ in $\mathcal F$. (Note that $u_i$ is non-empty.)
In this case,  $B(u) \subseteq B(v), K(u) \subseteq K(v),
Q(u) \subseteq Q(v)$. (See Section \ref{sec21} to recall the basic properties of these point sets.)
%(See Equations (\ref{twovertast}), (\ref{threevertast}), (\ref{fourvertast}), and (\ref{excmark}) in Section \ref{sec53}.)
Consider the path $\Pi_{j-1}(p)$ between $p$ and its $(j-1)$-level base point $b_{j-1}(p) \in B(u) \subseteq B(v)$
 guaranteed by the  induction hypothesis for $u$. Its weight is at most $\frac{1}{2} \cdot \mu_{j-1} = \frac{1}{2} \cdot \frac{\mu_j}{\rho} < \frac{1}{2} \cdot \mu_j$
and it consists of at most $3 \ell$ edges.
Also, all points of $\Pi_{j-1}(p)$ belong to $Q(u) \subseteq Q(v)$.
%By construction, the kernel set $K(v)$ of $v$ contains the surviving kernel set
%(see Equation (\ref{fourvertast}) $K'(v) ~=~ \bigcup_{z \in \cS(v)} K(z)$.
Moreover, suppose now that $p \in K(v)$.
Recall that $u$ is the unique $(j-1)$-level bag
such that $p \in Q(u)$. Since $K(v) \subseteq \bigcup_{z \in \chi(v)} K(z)$
and each kernel set $K(z)$ is contained in $Q(z)$,
%$z \in \cS(v) \cup \cJ(v)$,
% \subseteq \bigcup_{u \in \cS(v) \cup \cJ(v)} Q(u)$,
%it follows that
it follows that
$p \in K(u)$. By the induction hypothesis, $\Pi_{j-1}(p)$
consists of at most $2\ell$ edges.
% and so $\Pi_{j-1}(p)$ consists of at most $2\ell$ edges.
%(that is, any point of $u_i$ that belongs to the kernel $K(v)$ of $v$
%also belongs to the kernel $K(u)$ of $u$).
%the bag $u_i$, hence they also belong to the bag $v$.
Thus, we set $\Pi_j(p) = \Pi_{j-1}(p)$ and $b_j(p) = b_{j-1}(p)$.    % which proves the first  in this case.

We henceforth assume that $u$ is a disappearing zombie, i.e., $u \in \cJ(v)$ is a step-child of $v$.
% that is attached to $v$ in iteration $j$.
In this case,  since $v \in \cF_j$ is an actual adopter, it must hold that $j \ge \gamma + 1$.
%The bag $v$ is the actual adopter of the zombie $u$.
%  I stop here, because I
%think that we should define some more notions (e.g., attachment edge),
%and I do not want to do it myself]]
%Since the actual attachment between $u$ and $v$ occurs in iteration $j$,
%the treaty for attachment was signed $\gamma$ iterations earlier.
%Write $j' = j - \gamma$;
% and consider the adop
%this treaty was signed between the (single non-empty) twin $j'$-descendant $\tilde u$ of $u$ and some non-empty $j'$-descendant $x$ of $v$,
%through the attachment edge $(r_{j'}(\tilde u),r_{j'}(x))$, where $r_{j'}(\tilde u)$ and $r_{j'}(x)$
%are the representatives of $\tilde u$ and $x$ in iteration $j'$, respectively.
%$\\$  The following claim is within the proof of Lemma \ref{key}]]
%  We should probably define $j' = j - \gamma$ and use $j'$]]
For each index $i \in [0,\gamma-1]$, let $y^{(i)}$
denote the $(j-1 - (\gamma-1) + i) = (j-\gamma+i)$-level \emph{copy} of $u$. (That is,
for each of these identical copies $y^{(i)}$, we have $Q(y^{(i)}) = Q(u)$.
%  no -- we just have $Q(y^{(i)}) = Q(u)$ (otherwise, we need to keep the remark that follows the definition
%of identical bags)]]
See Lemma \ref{newl}.)
%  (See Lemma \ref{newl}.)]]
In particular, $u = y^{(\gamma-1)}$ is a disappearing zombie,
and $y^{(0)} = y$ is an attached bag.
Observe that an attachment $\cA(x,y)$, for some $(j-\gamma)$-level  bag
$x = x^{(0)}$, occurs during the $(j-\gamma)$-level processing. As a result
of this attachment, $y = y^{(0)}$ became an attached bag, and the corresponding disappearing zombie is $y^{(\gamma-1)} = u$.
The initiator bag $x$ of this attachment is a descendant of the
actual adopter $v$. The bags $x^{(1)} = \pi(x^{(0)}),x^{(2)} = \pi(x^{(1)}),
\ldots,x^{(\gamma-1)} = \pi(x^{(\gamma-2)})$ are labeled as a result of
this attachment as incubators. Observe that $v = x^{(\gamma)} = \pi(x^{(\gamma-1)})$.
%(The adopter $v = x^{(\gamma)}$ is not labeled as an incubator.)
%Hence $y \in \cF_{j-\gamma}$.
%Similarly, let $x = x^{(0)} \in \cF_{j-\gamma +1}$ be the
%incubator that adopts $y$   no, we want to take its child]]. The bag $x$ is a $(j-\gamma+1)$-level   no, we
%consider the bag on level $j-\gamma$]]
%descendant of $v$.
Recall that the attachment $\cA(x,y)$ is represented
by the edge $(r(x),r(y))$ in the spanner $\tilde G$.
(Recall that $r(x)$ and $r(y)$ are the representatives of the bags $x$ and $y$, respectively.)
See Figure \ref{fig3} for an illustration.
\begin{figure*}[htp]
\begin{center}
\begin{minipage}{\textwidth} %{5in}
\begin{center}
\setlength{\epsfxsize}{5in} \epsfbox{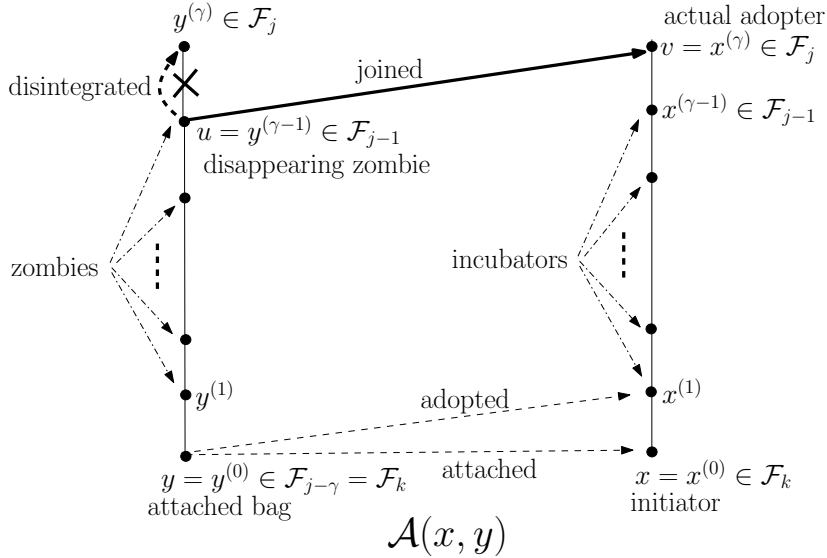}
\end{center}
\end{minipage}
\caption[]{ \label{fig3} \sf $u$ is a disappearing zombie, a step-child of (the actual adopter) $v$.}
\end{center}
\end{figure*}

We will use the following claim to prove Lemma \ref{key}.
%(Its proof can be found in Section \ref{app:inside}.   remove]])
\begin{claim} \label{inside}
%By the second  of the induction hypothesis
Define $k = j - \gamma$.
There is a   path $\Pi(p,r(y))$ in
  $\tilde G$ between $p$ and $r(y)$ that has weight at most
$2 \cdot \mu_{k}$
%= (1+2\) \cdot \frac{\mu_j}{\rho^\gamma}$,
and at most $\ell-2$ edges. Also, all points of $\Pi(p,r(y))$ belong to $Q(y) = Q(u) \subseteq Q(v)$.
\end{claim}
\begin{proof}
Recall that $p \in Q(u)$, and $y$ is a $(j-\gamma)$-level copy of $u$.
Hence both points $p$ and $r(y)$ belong to the $k$-level bag $y$, i.e., $p,r(y) \in Q(y)$.
%Next, we prove the second .  % follows as a corollary of the first .
%Consider an arbitrary pair $p,q$ of points in $Q(v)$, and let
Consider the paths $\Pi_{k}(p)$ and $\Pi_{k}(r(y))$
in $\tilde G$ that are guaranteed by the induction hypothesis for $y$,
having weight at most $\frac{1}{2} \cdot \mu_{k} = \frac{1}{2} \cdot \frac{\mu_j}{\rho^\gamma}$;
% and at most $3\log_\rho n$ edges;
all points of these two paths
belong to $Q(y)$.
The path $\Pi_{k}(p)$ (respectively, $\Pi_{k}(r(y))$) leads to a point $b_{k}(p)$ (resp., $b_{k}(r(y))$)
in the base point set $B(y)$ of $y$.
Recall that the spanner $\tilde G$ contains a path $P(y)$
which connects the base point set $B(y)$ via a simple path.
%(See the end of Section 2.2 and Appendix \ref{disregard}.)
Denote by $\Pi(b_{k}(p),b_{k}(r(y)))$ the sub-path of $P(y)$ between $b_{k}(p)$ and $b_{k}(r(y))$;
by the triangle inequality,  the weight of this path is at most
$\delta_\cL(b_{k}(p),b_{k}(r(y))) ~\le~
 \mu_{k} ~=~  \frac{\mu_j}{\rho^\gamma}.$
We set $\Pi(p,r(y)) = \Pi_{k}(p) \circ \Pi(b_{k}(p),b_{k}(r(y))) \circ \Pi_{k}(r(y))$.
%to be the path that is obtained as the concatenation
%of the path $\Pi_{k}(p)$, the path $\Pi(b_{k}(p),b_{k}(r(y)))$, and the path $\Pi_{k}(r(y))$.
(We assume that $\Pi(p,r(y))$ is a simple path. Otherwise we transform it into such by eliminating loops.)
It is easy to see that $\Pi(p,r(y))$ is a path between $p$ and $r(y)$
in the spanner $\tilde G$ that has weight at most $2 \cdot \mu_k = 2 \cdot \frac{\mu_j}{\rho^\gamma}$.
Moreover, all points of $\Pi(p,r(y))$ belong to $Q(y) = Q(u) \subseteq Q(v)$.
Note that an attached bag $y \in \cF_k$ was necessarily marked as risky
by Procedure $Process_k$. Therefore, $y$ is a small bag. Hence $|Q(y)| \le \ell-1$.
%Finally, we note that $\tilde u$ must be a small bag. Indeed, by the description of the algorithm,
%if $\tilde u$ were a large bag it would not be turned into a zombie in iteration $j'$ and attached to $x$.
%hence all points of $P(p,r_{j'}(u))$ belong to $u$ (and thus to $v$), hence
Since $\Pi(p,r(y))$ is a simple path, 
it consists of at most $\ell -2$ edges,
%Finally, we note that
which completes the proof of
Claim \ref{inside}. See Figure \ref{fig8} for an illustration.
\begin{figure*}[htp]
\begin{center}
\begin{minipage}{\textwidth} %{5in}
\begin{center}
\setlength{\epsfxsize}{6.1in} \epsfbox{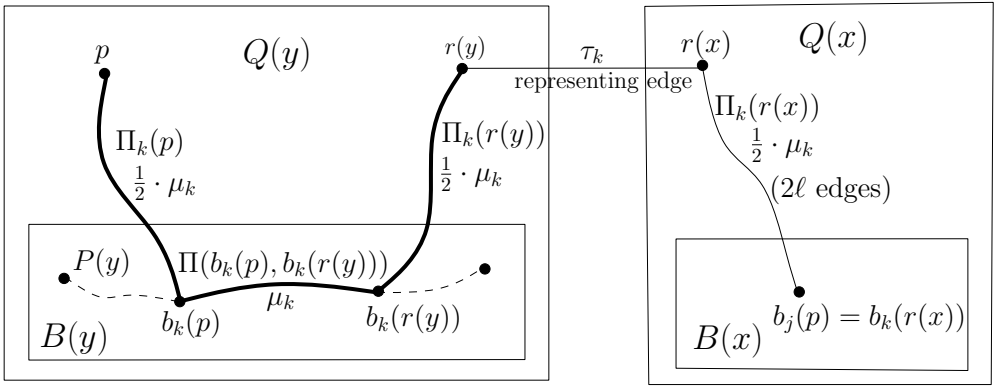}
\end{center}
\end{minipage}
\caption[]{ \label{fig8} \sf 
The path $\Pi(p,r(y))$ is depicted by a bold solid line. It is a sub-path of
the path $\Pi_j(p)$, which connects $p$ with $b_j(p)$, and is depicted by a solid line.}
%An illustration of the path $\Pi(p,r(y))$ (for the proof of Claim \ref{inside})
%and the path $\Pi_j(p)$.
%  caption needs rewriting I think]]}
\end{center}
\end{figure*}
\QED
\end{proof}

%  Recall that we are still in the proof of Lemma \ref{key}.]]
%We now continue proving Lemma \ref{key}.
Observe that at this point we have built a ``good path'' $\Pi(p,r(y)) \circ (r(y),r(x))$ from $p$ to $r(x)$.  % via $\Pi(p,r(y))$ and the edge $(r(x),r(y))$.
We now need to ``connect'' $r(x)$ to a point $b_j(p) \in B(v)$, which will be designated as the $j$-level base point of $p$.
%Recall that $x$ is a descendant of $v$ in $\cF$.
Observe that all bags $x = x^{(0)},x^{(1)} = \pi(x^{(0)}), \ldots,v = x^{(\gamma)} = \pi(x^{(\gamma-1)})$ 
along the path in $\cF$ between the attachment initiator $x$ and the actual adopter $v = x^{(\gamma)}$ (which is an ancestor of $x$ in $\cF$)
are not zombies. 
In particular, none of these bags is a disappearing zombie. 
%Since $x$ is a descendant of $v$ in $\cF$, 
%Notice that all points 
%[[S: perhaps should  explain that the entire path from $x$ to $v$ in $\cF$ does not ``disappear'',
%as $x$ is an attachment initiator and $v$ is the corresponding actual adopter]]
It follows that $B(x) \subseteq B(v), K(x) \subseteq K(v), Q(x) \subseteq Q(v)$.
Also, since a representative
of a bag must belong to its kernel,  we have $r(x) \in K(x)$.
%[[S: changed from: it follows that $r_{j'}(x) \in K(x) \subseteq K(v)$]];
%thus $r_{j'}(x)$ belongs to the kernel of $x$ [[S: redundant: as well as to the kernel of $v$]].
By the induction hypothesis for $x$, there exists
a path $\Pi_{k}(r(x))$ between $r(x)$
and its $k$-level base point $b_{k}(r(x)) \in B(x) \subseteq B(v)$ in the spanner $\tilde G$.
%that is guaranteed b
Moreover, all points of this path belong to $Q(x) \subseteq Q(v)$.   %$ (and thus to $v$).
In addition, the weight of this path is at most
$\frac{1}{2} \cdot \mu_k = \frac{1}{2} \cdot \frac{\mu_j}{\rho^\gamma}$,
%$\ \cdot ) = \ \cdot \frac{v)}{\rho^\gamma}$,
and since $r(x) \in K(x)$,
it consists of at most $2 \ell$ edges.
%Moreover, if $r_{k}(w) \in K(w) \subseteq K(v)$, then $\Pi_{k}(r_{k}(w))$ consists
%of at most $2 \log_\rho n$ edges.
We set $b_j(p) = b_{k}(r(x)) \in B(v)$,
and $\Pi_j(p) = \Pi(p,r(y)) \circ (r(y),r(x)) \circ \Pi_{k}(r(x))$.
(See Figure \ref{fig8}.)
It is easy to see that $\Pi_j(p)$ is a path between $p$ and its $j$-level base point $b_j(p) = b_k(r(x))$,
and that all points of $\Pi_j(p)$ belong to $Q(v)$. Notice that $\omega(r(y),r(x)) \le \tau_k$.
%of the attachment edge $(r(x),r(y))$ is bounded above by the $k$-level threshold $\tau_k$.
%Specifically,
%$$\tau_k ~=~ (1+\) \cdot \tau_{k}
%~=~ (1+\) \cdot \xi_{k} \cdot \rho \cdot t \cdot (1+\frac{1}{c})
%~=~ (1+\) \cdot \mu_k \cdot c \cdot \rho \cdot t \cdot (1+\frac{1}{c})
%~=~ (1+\) \cdot (c+1) \cdot \rho \cdot t \cdot \frac{\mu_j}{\rho^{\gamma}}.$$
%Recall also that $\gamma = c_0 \cdot (\log_\rho t + \log_\rho c + 1)$, for a sufficiently large constant $c_0$;
%thus
%Note that $\rho^\gamma = (c \cdot \rho \cdot t)^{c_0}$.
Therefore, the total weight $\omega(\Pi_j(p))$
%$\omega(P(p,b_{k}(r_{k}(x)))$
of the path
$\Pi_j(p) =  \Pi(p,r(y)) \circ (r(y),r(x)) \circ \Pi_{k}(r(x))$
%$P(p,b_{k}(r_{k}(x))$
satisfies (for sufficiently large $c_0$)
\begin{eqnarray*}
%\omega(P(p,b_{k}(r_{k}(x))))
\omega(\Pi_j(p)) &\le& 2\mu_k + \tau_k + \frac{1}{2} \cdot \mu_k ~=~
%&=&
%\omega(\Pi(p,r(y))) + \omega(r(x),r(y)) + \omega(\Pi_{k}(r(x))) \\
%(1+2\) \cdot \frac{\mu_j}{\rho^\gamma} + (1+\) \cdot (c+1) \cdot \rho \cdot t \cdot \frac{\mu_j}{\rho^{\gamma}} + \ \cdot \frac{\mu_j}{\rho^\gamma}
%\\ &=&
\mu_j \cdot \frac{\left(\frac{5}{2} + 2 \cdot (c+1) \cdot \rho \cdot t\right)}{(c \cdot \rho \cdot t)^{c_0}}
~<~ \frac{1}{2} \cdot \mu_j.
%\mu_j \cdot \frac{1}{\rho \cdot t} ~=~
\end{eqnarray*}
(Recall that $c_0$ is a sufficiently large constant of our choice. Setting $c_0 \ge 8$ is enough here.)
%(The last inequality holds for a sufficiently large constant $c_0$.)   % Specifically, if $c \ge 2$, then $c_0 \ge 4$ is sufficient.)
%In particular, we get that $r_{k}(w) \in K(w) \subseteq K(v)$ is in the kernel
Also,      %The number $|\Pi_j(p)|$ of edges in the path $\Pi_j(p)$ satisfies
%\begin{eqnarray*}
%\label{tryit}
%\nonumber
%|P(p,b_{k}(r_{k}(x)))|
it holds that $|\Pi_j(p)|
~=~
|\Pi(p,r(y))| + 1 + |\Pi_{k}(r(x))| ~\le~ (\ell -2) + 1 + 2\ell ~\le~ 3\ell.
$
%\end{eqnarray*}

Suppose now that $p \in K(v)$. We argue that in this case $x$ is a small bag.
(This case is characterized by $u \in \cJ(v), p \in Q(u) \cap K(v)$.)
%we can improve the bound on $|P(p,b_{j'}(r_{j'}(x))|$ to $2\log_\rho n$
%as follows. The fact that $p \in K(v)$ implies
%it must hold that $x$ is a small bag.
%Indeed, by the description of the algorithm, if $x$ were a large bag then no point of $y$ would belong to the kernel $K(v)$ of $v$, contradicting the fact that %$p$ belongs to $Q(u) = Q(y)$.
Suppose for contradiction otherwise, and consider
%Indeed, if $x$ were a large bag,
the $(j-1)$-level ancestor $x^{(\gamma-1)}$ %would be a large bag
of $x$, which is a surviving child of $v = x^{(\gamma)}$.
%By Equations (\ref{justadded}) and (\ref{threevertast}) (see Section \ref{sec53}),
Observe that
$$K'(v) ~=~  \bigcup_{z \in \cS(v)} K(z) ~\supseteq~ K(x^{(\gamma-1)}) ~\supseteq~ K(x).$$
By Lemma \ref{kv}, $|K'(v)| \ge |K(x)| \ge \ell$.
%as well, and so $Q(x^{(\gamma-1)}) \ge \ell$. Note that $x^{(\gamma-1)}$ ,
%i.e., $x^{(\gamma-1)} \in \cS(v)$.
%By Equation (\ref{twovertast}, $Q'(v) = \bigcup_{z \in \cS(v)} Q(z) \supseteq Q(x)$.
%then we will have $Q(v) \
%Equation (\ref{fourvertast}) thus yields
%and (\ref{excmark}),
%for a large bag $v$,
By construction, $K(v) = K'(v) = \bigcup_{z \in \cS(v)} K(z)$.
Hence the kernel set $K(v)$ of $v$
contains only points from the kernel sets of its surviving children, and contains no points from its 
step-children. However, $p \in  Q(u)$, and $u$ is a  step-child of $v$.
Hence $p \nin K(v)$, a contradiction.
\\Therefore $x$ is a small bag, and so $|Q(x)| < \ell$.
 We may assume that $\Pi_{k}(r(x))$ is a simple path.
Since % is a simple path and
all points of $\Pi_{k}(r(x))$ belong to $Q(x)$,  this path consists of at most $\ell -2$ edges (rather than at most $2\ell$ edges as in the general case).
Hence,
%Substituting $|\Pi_{k}(r_{j'}(x))|  n -1$ in Equation (\ref{tryit}) yields
%$P(p,b_{j'}(r_{j'}(x))| \le 2\log_\rho n$, as required.
%\begin{eqnarray*}
%\nonumber
$|\Pi_j(p)| ~=~ |\Pi(p,r(y))| + 1 + |\Pi_{k}(r(x))| ~\le~ (\ell -2) + 1 + (\ell -2) ~\le~ 2\ell.$  \QED
%\end{eqnarray*}
%This completes the proof of Lemma \ref{key}.\QED
\end{proof}

%We use
Lemma \ref{key} implies the following corollary.
%(The proof is deferred to Appendix \ref{key2:app}.)
\begin{corollary} \label{key2}
Fix an arbitrary index $j \in [\ell]$, and let $v$ be an arbitrary non-empty $j$-level bag.
There is a  path in the spanner $\tilde G$ between every pair of points in $Q(v)$,
having weight at most $2 \cdot \mu_j$ and at most $O(\log_\rho n + \alpha(\rho))$ edges.
In particular, the metric distance between any two points in $Q(v)$ is at most $2 \cdot \mu_j$.
\end{corollary}
\begin{proof}
Consider an arbitrary pair $p,q$ of points in $Q(v)$, and let $\Pi_j(p)$ and $\Pi_j(q)$
be the paths in $\tilde G$ that are guaranteed by Lemma \ref{key},
having weight at most $\frac{1}{2} \cdot \mu_j$ and at most $3\ell = 3 \lceil \log_\rho n \rceil$ edges each.
The path $\Pi_j(p)$ (respectively, $\Pi_j(q)$) leads to a point $b_j(p)$ (resp., $b_j(q)$)
in the base point set $B(v)$ of $v$. The spanner $\tilde G$ contains the path-spanner $H$.
Recall that for any pair $x,y \in Q$ of points,
there is a path $\Pi_H(x,y)$ in the path-spanner $H$
that has weight at most $\delta_\cL(x,y)$ and $O(\log_\rho n +\alpha(\rho))$ edges.
In particular, $H$ contains a path $\Pi_{H}(b_j(p),b_j(q))$
between $b_j(p)$ and $b_j(q)$, having weight at most $\delta_\cL(b_j(p),b_j(q)) \le \mu_j$
 and $O(\log_\rho n + \alpha(\rho))$ edges. Consider the path
$\Pi(p,q) = \Pi_j(p) \circ \Pi_{H}(b_j(p),b_j(q)) \circ \Pi_j(q)$.
%obtained as the concatenation of the path $\Pi_j(p)$, the path $\Pi_{H}(b_j(p),b_j(q))$ and the path $\Pi_j(q)$. 
Note that $\Pi(p,q)$ is a path between $p$ and $q$ in the
spanner $\tilde G$ that has weight at most $\frac{1}{2} \cdot \mu_j + \mu_j + \frac{1}{2} \cdot \mu_j = 2 \cdot \mu_j$, and
at most $3 \lceil \log_\rho n \rceil + O(\log_\rho n + \alpha(\rho)) + 3 \lceil \log_\rho n \rceil = O(\log_\rho n + \alpha(\rho))$ edges.
\QED
\end{proof}

The next lemma implies that $\tilde G$ is a $(t+\eps)$-spanner for $M[Q]$ with diameter $O(\Lambda(n) + \log_\rho n + \alpha(\rho))$.
Recall that $\Lambda(n)$ is an upper bound on the diameter
of the auxiliary spanners, produced by Algorithm $BasicSp$.
(See the statement of Theorem \ref{ourresult}.)
\begin{lemma} \label{lmstretch2}
For any $p,q \in Q$, there is a $(t+\eps)$-spanner path in $\tilde G$ with $O(\Lambda(n) + \log_\rho n + \alpha(\rho))$ edges.
\end{lemma}
\begin{proof}
We start the proof of the lemma with the following observation.
% (see Appendix \ref{obstretch2:app} for the proof).
\begin{observation} \label{obstretch2}
Fix any index $j \in [0,\ell]$.
For any pair $u,v \in \cF_j$ of non-empty $j$-level bags,
such that $\delta(r(u),r(v)) \le \frac{\tau_j}{t}$,
there is a $t$-spanner path in $G^*_j$ between $r(u)$ and $r(v)$
with at most $\Lambda(n)$ edges.
\end{observation}
\begin{proof}
Since $u$ and $v$ are non-empty $j$-level bags,
it holds that $r(u),r(v) \in Q_j$. In addition,
since $G'_j$ is a $t$-spanner for $Q_j$ with diameter at most $\Lambda(n)$, there is a $t$-spanner path $\Pi$ in $G'_j$ between
$r(u)$ and $r(v)$ that consists of at most $\Lambda(n)$ edges. The fact that $\Pi$ is a $t$-spanner path between
$r(u)$ and $r(v)$ implies
that the weight $\omega(\Pi)$ of $\Pi$ satisfies $\omega(\Pi) \le t \cdot \delta(r(u),r(v)) \le \tau_j$.
Clearly, the weight of each edge of $\Pi$ is bounded above by $\omega(\Pi) \le \tau_j$.
By construction (see Section \ref{s:jlevel}), $G^*_j$ contains all the edges of $G'_j$ with weight at most $\tau_j$.
It follows that
all edges of $\Pi$ belong to $G^*_j$. Observation \ref{obstretch2} follows.
\QED
\end{proof}
Next, we continue with the proof of Lemma \ref{lmstretch2}.

Let $p,q \in Q$.
Suppose first that $\delta(p,q) \le \frac{L}{n} < \frac{\tau_0}{t}$. % < (1+) \cdot \frac{\tau_0}{t}$.
Note that the graph $\tilde G_0 = G^*_0$
%(which also belongs to the spanner $G''$ constructed by Algorithm $BasicLightSp$)
belongs to  $\tilde G$.
% and so the case $j=0$ follows immediately
%from Observation \ref{obstretch}. We henceforth assume that $j \in [\log_\rho n]$.
By Observation \ref{obstretch2} for $j = 0$,
there is a $t$-spanner path in $\tilde G_0$  between $p$ and $q$ with at most $\Lambda(n)$ edges.
\\We henceforth assume that $\delta(p,q) > \frac{L}{n}$. Let $j \in [\ell]$ be the index such that $\rho^{j-1} \cdot \frac{L}{n} < \delta(p,q) \le \rho^j \cdot \frac{L}{n}$,
i.e., $\xi_j < \delta(p,q) \le \rho \cdot \xi_j$.
%(Observe that $j \le \lceil \log_\rho n \rceil = \ell$.)
Let $u = v_j(p)$ (respectively, $w = v_j(q)$) be the $j$-level host bag of $p$ (resp., $q$).
% i.e.,
%$u$ (resp., $w$) is the unique $j$-level bag such that $p \in Q(u)$ (resp., $q \in Q(w)$).
%Note that $c > 2$;
By Corollary \ref{key2}, the metric distance between every pair of points in
the same $j$-level bag is at most
$2 \cdot \mu_j < \xi_j.$
(See the beginning of Section \ref{section2} for the definitions of $\mu_j$ and $\xi_j$, and for other relevant notation.)
%(1+2\) \cdot \frac{\xi_j}{c}  ~=~
Since $\delta(p,q) > \xi_j$, it follows that $u \ne w$.
Consider the representative $r(u) \in Q_j$ (respectively, $r(w) \in Q_j$)  of $u$ (resp., $w$);
by Corollary \ref{key2}, $\delta(p,r(u)),\delta(q,r(w)) \le 2 \cdot \mu_j = 2 \cdot \frac{\xi_j}{c}$.
%Recall also that $\rho \ge 2$.
It follows that
\begin{eqnarray} \label{fin2}
\delta(r(u),r(w)) &\le& \delta(p,r(u)) + \delta(p,q) + \delta(q,r(w))
%~\le~ \delta(p,q) + 4 \cdot \frac{\xi_j}{c} \\ \nonumber &\le&
~\le~ \delta(p,q) + 4 \cdot \frac{\xi_j}{c}
\\ \nonumber &\le&
\rho \cdot \xi_j + 4 \cdot \frac{\xi_j}{c}
~=~ 2 \rho^j \cdot \frac{L}{n} \cdot \left(\frac{1}{2} + \frac{2}{\rho \cdot c}\right)
%~\le~ \rho \cdot \xi_j \cdot \left(1+\frac{2}{c}\right)
%\\ \nonumber &\le&
%~\le~
%2 \cdot \rho \cdot \xi_j \cdot \left(1+\frac{1}{c}\right)
~\le~
2  \rho^j \cdot \frac{L}{n} \cdot   \left(1+\frac{1}{c}\right)
=~ \frac{\tau_j}{t}.
\end{eqnarray}
By Observation \ref{obstretch2}, there is a $t$-spanner path between $r(u)$ and $r(w)$ in $G^*_j$ (and thus in $\tilde G$) with at most $\Lambda(n)$ edges; denote this path by $\Pi^*(r(u),r(w))$, and observe that $\omega(\Pi^*(r(u),r(w))) \le t \cdot \delta(r(u),r(w))$.
% \le
%t \cdot dist$.
Also, by Corollary \ref{key2}, the spanner $\tilde G$ contains a path $\Pi(p,r(u))$ (respectively, $\Pi(q,r(w))$) between $p$ and $r(u)$ (resp., between $q$ and $r(w)$) that has weight at most
%$(1+2\) \cdot \frac{\xi_j}{c} =
$2 \cdot \mu_j = 2 \cdot \frac{\xi_j}{c}$
and $O(\log_\rho n + \alpha(\rho))$ edges.

Let $\Pi(p,q) = \Pi(p,r(u)) \circ \Pi^*(r(u),r(w)) \circ \Pi(q,r(w))$.
% obtained from the concatenation of the path $\Pi(p,r(u))$, the path $\Pi^*(r(u),r(w))$, and the path $\Pi(q,r(w))$.
Note that $\Pi(p,q)$ is a path in $\tilde G$ between $p$ and $q$ that has weight $\omega(\Pi(p,q))$ at most
$t \cdot \delta(r(u),r(w)) + 4 \cdot \frac{\xi_j}{c}$ and
 $O(\Lambda(n) + \log_\rho n + \alpha(\rho))$
edges.
%To complete the proof of the lemma, we show that $\Pi(p,q)$ is a $(t+\eps)$-spanner path
%between $p$ and $q$.
By Equation (\ref{fin2}), $t \cdot \delta(r(u),r(w)) ~\le~ t \cdot (\delta(p,q) + 4 \cdot \frac{\xi_j}{c}).$
Also, recall that $c = \lceil \frac{4 \cdot (t+1)}{\eps} \rceil$.
%, we have $\eps \ge \frac{4 \cdot (t+1)}{c}$.
It follows that
\begin{eqnarray*}
\omega(\Pi(p,q)) &\le&
%t \cdot \delta(r_j(u),r_j(w)) + 4 \cdot \frac{\xi_j}{c} ~\le~
t \cdot \left(\delta(p,q) + 4 \cdot \frac{\xi_j}{c}\right) + 4 \cdot \frac{\xi_j}{c}
%\\ &\le& t \cdot \left(\delta(p,q) + 4 \cdot \frac{\delta(p,q)}{c}\right) + 4 \cdot \frac{\delta(p,q)}{c}
~\le~ \left(t + \frac{4 \cdot (t+1)}{c} \right) \cdot \delta(p,q) ~\le~ (t+\eps) \cdot \delta(p,q). 
\end{eqnarray*}
Hence, $\Pi(p,q)$ is a $(t+\eps)$-spanner path in $\tilde G$ between $p$ and $q$
with $O(\Lambda(n) + \log_\rho n + \alpha(\rho))$ edges. \QED
%. The lemma follows.\QED
\end{proof}

\subsection{Degree} \label{deg:app}

In this section we bound the maximum degree of our spanner $\tilde G$. Specifically, we  will show that the degree
of  $\tilde G$ is $O(\Delta(n) \cdot \gamma + \rho)$. Recall that $\gamma = c_0 \cdot (\lceil \log_\rho t \rceil  + \lceil \log_\rho c \rceil + 1)$, for some constant $c_0$;
thus $\gamma = O(\log_\rho (t/\eps))$. In other words, we will get the desired degree bound of $O(\Delta(n) \cdot \log_\rho(t/\eps) + \rho)$.

%Under the natural assumption that $t$ and $c$ are constants, we will get the desired degree bound of $O(\Delta(n) + \rho)$.

As shown in Section \ref{disregard} (Corollary \ref{degbase}), the base edge set $\cB$ increases the degree bound by at most two units,
%re are at most $n$ base edges in $\tilde G$,
and so we may disregard it in this analysis.
%[[S: notice that there is a typo -- I often write $\log_\rho t$ instead of $\log_\rho t$;
%also, see elsewhere where we should use $t$ rather than $t$. Also, I must verify
%all the properties (degree, stretch, hop-diam, number edges and weight), and see
%
%[[S: verify whether parameters like $t$ should not appear within weight bound, diameter bound, etc.]]
We will also disregard the path-spanner $H$
and the 0-level auxiliary spanner $\tilde G_0$, which together contribute  $O(\Delta(n) + \rho)$ units to the degree bound.

The degree analysis is probably the most technically involved part of our proof.
We start with an intuitive sketch and then proceed to the rigorous proof.

Our algorithm makes a persistent  effort to  merge small bags together to form large bags.
Intuitively, a large bag is easy to handle because its kernel set contains enough (at least $\ell$) points to share the load.

If a $j$-level bag $v$ is large, i.e., $|Q(v)| \ge \ell$, then all its $\ell - j$ ancestors are large as well. 
Moreover, by Lemma \ref{kv}, the kernel set $K(v)$ of $v$ contains at least $\ell$ points.
A point $p \in Q(v)$ may get loaded by one of the
auxiliary spanners $\tilde G_j,\tilde G_{j+1},\ldots,\tilde G_\ell$
only if it is a representative of $v$ or of one of its ancestors. (In particular, the only points of $Q(v)$ that may get loaded belong
to the kernel set $K(v)$.)
However, we have at least $\ell$ points in $K(v) \subseteq Q(v)$ that can be used to ``represent $v$'' in at most $\ell$
auxiliary spanners. (One for $v$, and one for each of its ancestors.)
Hence it is not hard to share the load in such a way that each point
$p \in K(v)$ will be loaded by $O(1)$ auxiliary spanners. Consequently,
the maximum degree of points that belong to large bags are small.
(In fact, a point $p$ may, of course, belong to a small bag, and later
join a large bag. However, for the sake of this intuitive discussion one can
imagine that $p$ duplicates itself into $p^{large}$ and $p^{small}$,
where $p^{large}$ (respectively, $p^{small}$) belongs only to
large (resp., small) bags.)

For a small bag $v$, its representative $r(v)$ is loaded
by a $j$-level auxiliary spanner only if $r(v)$ is not isolated in $G^*_j$
(see Section \ref{s:jlevel}). It means that there exists another $j$-level
representative $r(u)$, such that $\delta(r(v),r(u)) \le \tau_j$;
in other words, $r(u)$ is close to $r(v)$.
Intuitively, we will want the bags $v$ and $u$ to merge, as this would
increase the pool of eligible representatives. We cannot merge them
right away, however, because this would blow up the weighted diameters of
the $(j+1)$-level bags. Instead we wait for $\gamma = O(1)$ levels,
and then merge $v$ into the $(j+\gamma)$-level ancestor $u'$ of $u$.
%  I think we should explain that if $v$ is lonely, it is all right to ``wait'' for $\gamma$ levels.
%If not, then we are fine anyway...]]
(Or the other way around, merge $u$ into the $(j+\gamma)$-level
ancestor $v'$ of $v$.) The weighted diameters of the $j$-level bags, are, roughly
speaking, proportional to the length $\mu_j$ of the $j$-level intervals,
i.e., they grow geometrically with the level $j$. Hence when $v$ is merged
into $u'$, it contributes only an  $\exp(-\Omega(\gamma))$-fraction to
the weighted diameter of the $(j+\gamma)$-level bag $u'$. In this way we keep
the weighted diameters of bags in check, while always maintaining sufficiently large
pools of eligible representatives. During the $\gamma$ levels
$j,j+1,\ldots,j+\gamma-1$, points of $v$ do accumulate some extra
degree; however, since $\gamma = O(1)$, they are overloaded by
at most a constant factor.

We next proceed to the rigorous analysis of the degree of the spanner $\tilde G$.
% can be found in Appendix \ref{deg:app}.

Consider an index $j \in [\ell]$.
In the next paragraph we shortly remind how the $j$-level auxiliary spanner $\tilde G_j$ is constructed.
Procedure $Process_j$ builds a spanner $G'_j$ for the set $Q_j$ of
representatives of all non-empty $j$-level bags, including zombies.
This spanner is then pruned to obtain the graph $G^*_j = (Q_j,E^*_j)$.
(By ``pruning'' we mean removing edges of weight greater than $\tau_j$.)
%  The rest happens if $j \le \ell - \gamma$ or something]]
If $j > \ell - \gamma$, then $\tilde G_j = G^*_j$.	
Otherwise, Procedure $Process_j$ constructs the subset $\hat Q_j$ of $Q_j$ of
useful (i.e., non-empty and non-zombie) representatives, which are not isolated in $G^*_j$.
It then constructs the spanner $\check G_j = (\hat Q_j,\check E_j)$
for the set $\hat Q_j$, and prunes it to obtain the graph $\hat G_j
= (\hat Q_j, \hat E_j)$. The union $\tilde G_j = (Q_j,\tilde E_j = E^*_j \cup \hat E_j)$
is the $j$-level auxiliary spanner. See Section \ref{s:jlevel} for more details.

%Observe that the maximum degree $\Delta(\tilde G_0)$ of the $0$-level auxiliary spanner
%is bounded above by $\Delta(n)$ 
%  we said in the previous page that we can safely ignore $\tilde G_0$ from the degree
%analysis]],
Observe that the maximum degree $\Delta(\tilde G_j)$ of the $j$-level auxiliary spanner,
$j \in [\ell]$, is bounded above by $\Delta(|Q_j|) + \Delta(|\hat Q_j|) \le 2 \cdot \Delta(n)$.
For future reference we summarize this observation below.
\begin{observation} \label{auxdeg}
For each index $j \in [\ell]$, $\Delta(\tilde G_j) = O(\Delta(n))$.
\end{observation}
Observation \ref{auxdeg} implies directly that $\Delta(\tilde G) = O(\log n \cdot \Delta(n) + \rho)$.
Such a bound on the maximum degree can, in fact, be achieved by a much simpler construction.
(See Section 1.3 for its outline.) In this section we show that our much more intricate construction
guarantees $\Delta(\tilde G) = O(\Delta(n) \cdot \log_\rho(t/\eps) + \rho)$. As $t$ is typically a constant,
and $\rho$ and $\eps$ can be set as constants, this would essentially imply that $\Delta(\tilde G) = O(\Delta(n))$.
\begin{definition}
A $j$-level bag $v$ is called \emph{active} if its representative
$r(v)$ is not isolated in $G^*_j$. Otherwise it is called \emph{passive}.
\end{definition}
Note that if $v$ is passive, then its representative $r(v)$ \emph{is not loaded}
during the $j$-level processing, i.e., $load_j(r(v)) = 0$. On the other hand,
for an active bag $v$, its representative \emph{is loaded}, i.e., $load_j(r(v)) = 1$.
Recall that $v$ is a growing bag if $|\chi(v)| \ge 2$. (See Section \ref{sec24}.)

Recall also (see the beginning of Section \ref{section2}) that the relation
step-parent - step-child among bags of $\cF$ defines another forest $\hat \cF$
over the same set of bags.
Specifically, %e now define the forest $\hat \cF$   already defined!]] over the same bag set in the following way.
a bag $v$ is a parent of $u$ in $\hat \cF$ iff $u \in \chi(v)$, i.e., $u$ is an extended child of $v$ (either
a surviving child or a step-child of $v$). We denote the parent-child
relation in $\hat \cF$ by $\hat \pi(\cdot)$, i.e.,  we write $v = \hat \pi(u)$.
Note that a bag of level $j$ in $\cF$ has level $j$ in $\hat \cF$ as well.
%  IMPORTANT: I suggest to move this definition to the end of Section 6.1, and to mention
%there once and for all that from the point on, when we say ancestor/descendant,
%we refer to the parent-child relation in $\hat \cF$ (and not in $\cF$);
%I'm not sure this is the right decision, but we should definitely think about it.]]
%  MAGEN-DAVID: Note that the forests $\cF$ and $\hat \cF$ are very similar. The only
%bags that have step-parents are disappearing zombies. In particular,
%if a bag is NOT $\eta$-safe-prospective, then all its $\eta-1$ immediate
%ancestors cannot be zombies by definition, let alone disappearing zombies,
%and thus the path to each one of them is via the original child-parent path in $\cF$.]]

Note that the forests $\cF$ and $\hat \cF$ are very similar.
The only bags $v$ that have step-parents (different from
their parents) are disappearing zombies.
We summarize this observation below.
\begin{observation} \label{notdis}
For a bag $v \in \cF_j, j \in [\ell-1]$,
which is not a disappearing zombie,
$\pi(v) = \hat \pi(v)$.
\end{observation}

We say that a bag $w$ is an \emph{$\cF$-descendant}
(respectively, \emph{$\cF$-ancestor}) of the bag $u$ if it is
a descendant (resp., ancestor) of $u$ in $\cF$.
Similarly, we say that a bag $w$ is an \emph{$\hat \cF$-descendant}
(respectively, \emph{$\hat \cF$-ancestor}) of the bag $u$ if it is
a descendant (resp., ancestor) of $u$ in $\hat \cF$.

%  did we define before the notion of $\cF$-descendant/ancestor?]]
\begin{definition}
For a positive integer parameter $\beta$,
we say that a bag $v \in \cF_j$ is \emph{$\beta$-prospective}, if one of its
$\beta$ immediate $\hat \cF$-ancestors $\hat v^{(1)} = \hat \pi(v),\hat v^{(2)} = \hat \pi(\hat v^{(1)}),
\ldots,\hat v^{(\beta)} = \hat \pi(\hat v^{(\beta-1)})$ is a growing bag.
%  what about $v$? I'm not sure if we should include it or not, probably not]]
For $j > \ell - \beta$, all bags $v \in \cF_j$ are called $\beta$-prospective.
We also use the shortcut \emph{prospective} for $\gamma$-prospective.
%  In particular, for $j > \ell - \gamma$, all bags $v \in \cF_j$ are prospective.]]
\end{definition}

%  Let's write: "The proof of the following lemma is immediate from the construction",
%and omit it.]]
%In Lemma \ref{safepros}
Next, we argue that a small safe bag is necessarily prospective.
Such a bag is either crowded, or a zombie, or an incubator.
We start with the case of a crowded bag.
\begin{lemma} \label{crow}
Let $j\in[\ell]$ be an arbitrary index.
Any crowded bag $v \in \cF_j$ is prospective.
\end{lemma}
\begin{proof}
The case $j > \ell - \gamma$ is trivial. We henceforth assume that $j \le \ell-\gamma$.
Since $v$ is a crowded $j$-level bag, its cage $\cC(v)$ contains another useful
%(i.e., non-empty and non-zombie)
$j$-level bag $u$.
Note that $u$ is crowded as well,
and thus both $v$ and $u$ are safe.
Let $w \in \cF_k$ be the least common $\cF$-ancestor
of $v$ and $u$. The index $k$ satisfies $j+1 \le k \le j + \gamma$.
Write $v = v^{(0)},u = u^{(0)}$,
and consider the $(k-j)$ immediate $\cF$-ancestors of $v$ and $u$,
$v^{(1)} = \pi(v^{(0)}),\ldots,v^{(k-j)} = \pi(v^{(k-j-1)}) = w$ and
$u^{(1)} = \pi(u^{(0)}),\ldots,u^{(k-j)} = \pi(u^{(k-j-1)}) = w$,
respectively.
By induction on $(k-j)$, it is easy to see that all
these bags, except maybe $w$ itself, are crowded and safe.
Hence none of them is a zombie, and so for each index
$i, 1 \le i \le k-j, v^{(i-1)} \in \cS(v^{(i)}) \subseteq \chi(v^{(i)}),
u^{(i-1)} \in \cS(u^{(i)}) \subseteq \chi(u^{(i)})$.
It follows that $v^{(k-j-1)}$ (respectively, $u^{(k-j-1)}$) is
an $\hat \cF$-ancestor of $v$ (resp., $u$),
and $v^{(k-j-1)},u^{(k-j-1)} \in \cS(w) \subseteq \chi(w)$.
Hence $w$ is the least common $\hat \cF$-ancestor of $v$
and $u$, and $|\chi(w)| \ge |\cS(w)| \ge 2$. Thus $w$ is
a growing bag, and $v$ is a prospective one.
See Figure \ref{fig4} for an illustration.
\begin{figure*}[htp]
\begin{center}
\begin{minipage}{\textwidth} %{5in}
\begin{center}
\setlength{\epsfxsize}{2.7in} \epsfbox{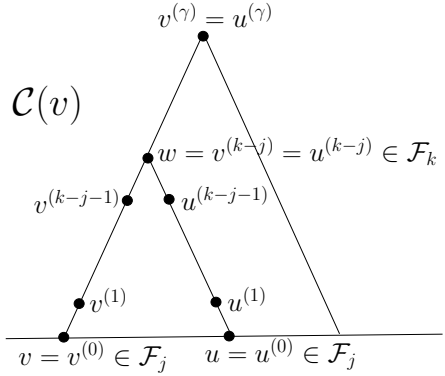}
\end{center}
\end{minipage}
\caption[]{ \label{fig4} \sf The cage $\cC(v)$ of $v$ and $u$.
All vertices $v^{(0)},v^{(1)},\ldots,v^{(k-j-1)},u^{(0)},u^{(1)},\ldots,u^{(k-j-1)}$
are crowded, and thus they are not zombies. Thus $w$ is growing.}
\end{center}
\end{figure*}
\QED
%More generally, all the bags on the paths from $u$ and $v$ to their least common
%ancestor $w$  (except for maybe $w$) are crowded and thus safe
%(This assertion can be proved by a starightfoward induction.)
%It follows that $w$ will have at least two safe children.
%Hence $|\cS(w)| \ge 2$,
%and so $w$ is a growing bag.
%Write  $w= v^{(k)} \in \cF_{j+k}, k \ge 1$.
%Since the cage-ancestor of $v$ and $u$ is a $(j+\gamma)$-level bag,
%it follows that $k \le \gamma$. We conclude that $v$ is a prospective bag.
\end{proof}

Next, we consider the case of a zombie bag.
\begin{lemma} \label{zomb}
Let $j \in [\ell]$.  A zombie bag $v \in \cF_j$ is prospective.
\end{lemma}
\begin{proof}
We only need to prove the assertion for $j \le \ell - \gamma$.
By construction, there exists an attached bag $w = w^{(0)}$, which is an $\cF$-descendant
of $v$. By Lemma \ref{newl}, the bags $w = w^{(0)}, w^{(1)} = \pi(w^{(0)}),
\ldots,w^{(i)} = v,\ldots,w^{(\gamma-1)}$ are identical, where $i \in [\gamma-1]$
is some index. The bag $w^{(\gamma-1)}$ is a disappearing zombie.

Observe that $w = w^{(0)} \in \cF_{j-i}$. There exists a bag $u \in \cF_{j-i}$,
so that the attachment $\cA(u,w)$ took place during the $(j-i)$-level processing.
The bag $u = u^{(0)}$ is the initiator of this attachment.
Denote  by $u^{(1)} = \pi(u^{(0)}),
\ldots,u^{(\gamma)} = \pi(u^{(\gamma-1)})$ the $\gamma$ immediate
$\cF$-ancestors of the initiator $u$. The bags
$u^{(1)},
\ldots,u^{(\gamma-1)}$ are incubators, and $u^{(\gamma)}$ is the
actual adopter. By Lemma \ref{tat}, neither of the incubator bags
$u^{(1)},\ldots,u^{(\gamma-1)}$ is a disappearing zombie.
Hence, by Observation \ref{notdis}, the actual adopter $u^{(\gamma)}$ is an $\hat \cF$-ancestor of
all the bags $u^{(0)},u^{(1)}, \ldots,u^{(\gamma-1)}$.
The initiator bag $u = u^{(0)}$ is non-empty, and thus the
incubators $u^{(1)},\ldots,u^{(\gamma-1)}$ are non-empty as well.
Hence for each index $h \in [0,\gamma-1], u^{(h)} \in \cS(u^{(h+1)}) \subseteq \chi(u^{(h+1)})$.
%   but this already follows from the fact
%that none of them is a disappearing zombie, see MAGEN-DAVID]].

Moreover, the attached bag $w = w^{(0)}$ is non-empty.
Hence the zombie bags $w^{(1)},w^{(2)},\ldots,w^{(\gamma-1)}$
are non-empty as well, and for each
index $h \in [0,\gamma-2], w^{(h)} \in \cS(w^{(h+1)}) \subseteq \chi(w^{(h+1)})$.
Since the bags $w^{(0)},w^{(1)}, \ldots,w^{(\gamma-2)}$ are not disappearing zombies,
Observation \ref{notdis} implies that the disappearing zombie $w^{(\gamma-1)}$ is an $\hat \cF$-ancestor of
all these bags, %$w^{(0)},w^{(1)}, \ldots,w^{(\gamma-2)}$,
and in particular,
of $v = w^{(i)}$. %  but this already follows from the fact
%that none of them is a disappearing zombie, see MAGEN-DAVID]]

Finally, the disappearing zombie $w^{(\gamma-1)}$ is a step-child of
the actual adopter $u^{(\gamma)}$, i.e., $w^{(\gamma -1)} \in \cJ(u^{(\gamma)})
\subseteq \chi(u^{(\gamma)})$. Hence $u^{(\gamma)}$ is an $\hat \cF$-ancestor
of $v$. By Lemma \ref{tat}, $u^{(\gamma-1)} \ne w^{(\gamma-1)}$.
Hence $|\chi(u^{(\gamma)})| \ge 2$, i.e., $u^{(\gamma)}$ is a growing bag.
Thus the bag $v$ is prospective. See Figure \ref{fig7}.
\begin{figure*}[htp]
\begin{center}
\begin{minipage}{\textwidth} %{5in}
\begin{center}
\setlength{\epsfxsize}{5in} \epsfbox{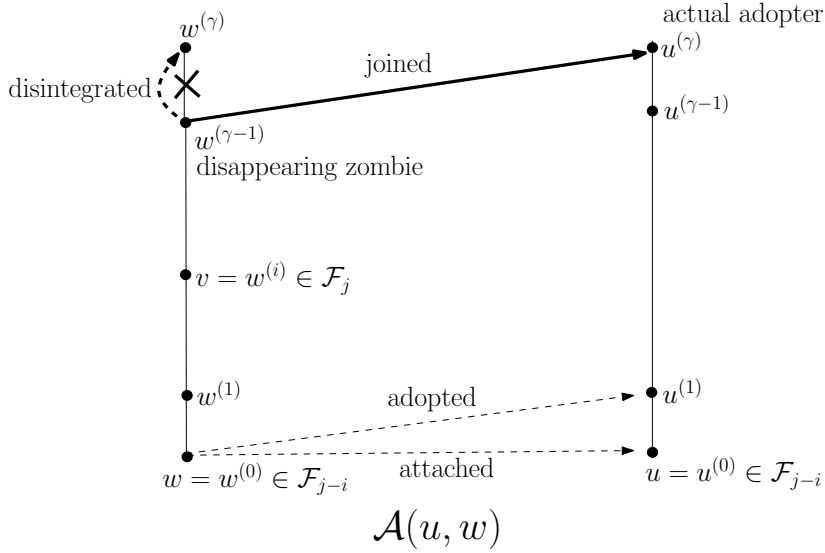}
\end{center}
\end{minipage}
\caption[]{ \label{fig7} \sf The case that $v = w^{(i)}$ is a zombie.
The bags $w^{(\gamma-1)}$ and $u^{(\gamma-1)}$ are two distinct
non-empty extended children of $u^{(\gamma)}$.
Thus $u^{(\gamma)}$ is a growing $\hat \cF$-ancestor of $v = w^{(i)}$.}
\end{center}
\end{figure*}
\QED
\end{proof}
A symmetric argument shows that an incubator bag is prospective as well.
\begin{lemma} \label{inc}
Let $j \in [\ell]$. An incubator bag $v \in \cF_j$ is prospective.
\end{lemma}

Lemmas \ref{crow}, \ref{zomb} and \ref{inc}
imply the following statement.
\begin{corollary} \label{safepros}
Let $j \in [\ell]$. A small safe bag $v \in \cF_j$ is prospective.
\end{corollary}

For a positive integer parameter $\beta$, we say that a bag $v$
is \emph{$\beta$-safe-prospective} if either $v$ or one of
its $\beta$ immediate $\hat \cF$-ancestors is safe.
(We remark that for $v$ to be $\beta$-prospective, one of its $\beta$ immediate
ancestors must be growing, i.e., it is not enough for $v$ to be a growing bag.
This is not the case for a $\beta$-safe-prospective bag. That is, if $v$ is safe,
then it is $\beta$-safe-prospective.)
%  For $j > \ell - \beta$, all bags $v \in \cF_j$ are called $\beta$-safe-prospective.
%In particular, for $j > \ell - \gamma$, all bags $v \in \cF_j$ are safe-prospective.]]

Denote $\kappa = \lceil \log_\rho t \rceil$ and $\eta = 2\kappa + 3$.
Recall that $\gamma = c_0 \cdot (\kappa + \lceil \log_\rho c \rceil + 1)$.
%  here I take $\eta$;
%until here I wrote $\beta$ instead of $\eta$ (because $\eta$ is fixed, and $\beta$ was arbitrary
%for the definitions)]].
Next, we argue that any active small bag is either $\eta$-prospective
or $\eta$-safe-prospective.
Before proving it we shortly outline the main idea of our degree analysis.
Intuitively, large bags are easy to handle because they contain enough points to share the load.
As a result of this load-sharing, no point in a large bag ever becomes overloaded.
To handle small bags we show that once a small bag becomes active (i.e., its points start being loaded),
it will soon get merged into a larger bag. These merges will allow for a more uniform load-sharing,
resulting in a small (constant) load for all points in $Q$.

%[[S: There is a much shorter/simpler proof for the following lemma,
%if the stretch $t$ of the black-box spanners is not arbitrary, but rather close to 1 (i.e., $t = 1+\eps$, for some small
%constant $\eps$). I'll write the simpler proof separately.]]
\begin{lemma} \label{sp}
Let $j \le \ell - \eta$, and $u \in \cF_j$ be an active small bag
that is not $\eta$-prospective. Then $u$ is $\eta$-safe-prospective.
%[[S: of course this is true if $u$ is large, because then it is safe by definition and thus $\eta$-safe-prospective;
%nevertheless, it's more appropriate to state this lemma for small bags.]]
\end{lemma}
\begin{proof}
The bag $u$ is active, and thus non-empty. Since $u$ is
not $\eta$-prospective, it follows that the $\eta$
immediate $\hat \cF$-ancestors of
$u = \hat u^{(0)}$, namely, $\hat u^{(1)}= \hat \pi(\hat u^{(0)}),\ldots,\hat u^{(\eta)} = \hat \pi(\hat u^{(\eta-1)})$
are stagnating bags. Hence all these bags are identical to $u$, and moreover,
they have the same representative as $u$, i.e., $r(u) = r(\hat u^{(0)}) = r(\hat u^{(1)})
= \ldots = r(\hat u^{(\eta)})$.
(See Section \ref{sec24}.)

%[[MAGEN-DAVID]]
Suppose for contradiction that all these bags $\hat u^{(0)},\hat u^{(1)},\ldots,\hat u^{(\eta)}$
are risky. Note that for each index $i \in [0,\eta]$, if
$\hat u^{(i)}$ is a zombie, then it must be safe.
%either $\hat u^{(i)}$ or $\hat u^{(i+1)}$ (or both) are safe.
%(In particular, if $\hat u^{(i)}$ is an attached bag, then $\hat u^{(i+1)}$ is safe.
%Otherwise $\hat u^{(i)}$ is a zombie, but not an appearing one, and so it must be safe.)
%[[S: verify whether usually (not here) we mean that a zombie is not an appearing one]]
Hence the bags $\hat u^{(0)},\hat u^{(1)},\ldots,\hat u^{(\eta)}$
are not zombies, and thus useful.
%[[S: follows from MAGEN-DAVID, but I think this is redundant: ``
By Observation \ref{notdis},
$\hat u^{(1)} = u^{(1)},\ldots,
\hat u^{(\eta)} = u^{(\eta)}$, i.e., the $\eta$ immediate $\hat \cF$-ancestors
of $u$ are its $\eta$ immediate $\cF$-ancestors.

Consider the $j$-level processing. (It is described in Section \ref{s:jlevel}.) Since $u$ is active, the representative
$r(u)$ of $u$ is not isolated in $G^*_j = (Q_j,E^*_j)$. Hence $r(u) \in Q^*_j$.
Moreover, the bag
$u$ is useful, hence $r(u) \in \hat Q_j$.
If $r(u)$ is not isolated in $E^*_j(\hat Q_j)$, then it is not isolated in the
$j$-level attachment graph $G_j = (\hat Q_j,\cE_j),\cE_j = E^*_j(\hat Q_j) \cup
\hat E_j$. However, in this case $r(u)$ belongs to a star $S \in \Gamma_j$ in
the star forest $\Gamma_j$ formed by Procedure $Attach$ (within Procedure
$Process_j$). As a result the bag $u$ becomes either an attachment initiator
or an attached bag, and its parent $u^{(1)} = \hat u^{(1)}$
becomes an incubator or a zombie, respectively. In either case it becomes
safe, a contradiction.

Hence $r(u)$ is isolated in $E^*_j(\hat Q_j)$.
Recall that $\hat Q_j$ is the subset of $Q^*_j$ which contains only non-zombie representatives.
Since $r(u)$ is not isolated in $G^*_j = (Q_j,E^*_j)$, there must exist a zombie
$z$, such that $r(z) \in Q_j \setminus \hat Q_j$ and the edge $(r(u),r(z)) \in E^*_j$.
It follows that
\begin{equation} \label{mark}
\delta(r(u),r(z)) ~\le~ \tau_j.
\end{equation}

Also, the same argument applies for every index $h, h \in [j,j+(\eta -1)]$,
and not only for $h=j$.
If $r(u^{(h-j)})$ is not isolated in the $h$-level attachment graph $G_h$,
then the parent $u^{(h-j+1)} = \hat u^{(h-j+1)}$
of $u^{(h-j)}$ is safe.
Hence in this case $u$ is $\eta$-safe-prospective, a contradiction.

Therefore, from now on we assume that for all indices $h, h \in [j,j+(\eta-1)]$,
the representative $r(u) = r(u^{(h-j)}) = r(\hat u^{(h-j)})$ is isolated in $G_h$.

Next, we argue that $r(u)$ is quite far from any useful representative
on levels $j,j+1,\ldots,j+(\eta-1)$.
\begin{claim} \label{didit}
For any index $h, h \in [j,j+(\eta-1)]$, and any useful bag $w \in \cF_h,
w \ne u^{(h-j)}$, it holds that
%\begin{equation} \label{didit}
$\delta(r(u),r(w)) > \frac{\tau_h}{t}$.
%\end{equation}
\end{claim}
\begin{proof}
Suppose for contradiction that for some index $h \in [j,j+(\eta-1)]$ and a bag
$w$ as above, it %Equation (\ref{didit})
holds that $\delta(r(u),r(w)) \le \frac{\tau_h}{t}$.
It follows that $r(u) = r(u^{(h-j)})$ and $r(w)$ are not isolated in $G^*_h$.
Moreover, since $u^{(h-j)}$ and $w$ are useful bags,
their representatives belong to $\hat Q_h$, 
%[[S: that's incorrect as far as I understand; we should first show
%that these vertices are not isolated in $G^*_h$, and only later conclude that they belong to $\hat Q_h$.
%The reason that they are not isolated in $G^*_h$ is because their distance is at most $\frac{\tau_h}{t}$,
%and so the entire path between them in the non-pruned spanner should survive. BTW, I'm tired now, and may very well 
%miss something...]], 
i.e., $r(u) = r(u^{(h-j)}),r(w) \in \hat Q_h$. Part II of
Procedure $Process_j$ constructs a $t$-spanner $\check G_h = (\hat Q_h,\check E_h)$
for the metric $M[\hat Q_h]$ induced by $\hat Q_h$.
Hence there exists
a $t$-spanner path $\Pi = \Pi(r(u),r(w))$ in $\check G_h$ between
$r(u)$ and $r(w)$. Since $\delta(r(u),r(w)) \le \frac{\tau_h}{t}$,
it follows that $\omega(\Pi) \le \tau_h$.
Therefore all edges of $\Pi$ also have weight at most $\tau_h$. Hence
the path $\Pi$ is contained in the pruned graph $\hat G_h = (\hat Q_h,\hat E_h)$.
Moreover, $\hat E_h \subseteq \cE_h$, where $\cE_h$ is the edge set of the
$h$-level attachment graph $G_h = (\hat Q_h,\cE_h)$.
Hence $\Pi \subseteq \cE_h$ as well.
Therefore $r(u) = r(u^{(h-j)})$ is not isolated in $G_h$, a contradiction.
\QED
\end{proof}
Now we continue to prove Lemma \ref{sp}. Intuitively, we will show that $r(u)$ cannot be close to a zombie
representative $r(z)$ (in the sense of Equation (\ref{mark})), but far from any useful representative $r(w)$ in all levels
$h \in [j,j+(\gamma-1)]$ (see Claim \ref{didit}). This would lead to a contradiction.

Consider again the zombie $z \in \cF_j$, such that $\delta(r(u),r(z)) \le \tau_j$.
Let $i, j- (\gamma-1) \le i < j$, be the index such that an identical descendant
$y$ of $z$ became an attached bag during the $i$-level processing.
More specifically, during the $i$-level processing the bag $y$ was attached to another
$i$-level bag $v$ by an attachment $\cA(v,y)$. As a result of this attachment,
the $(\gamma-1)$ immediate $\cF$-ancestors of $y = y^{(0)}$, i.e.,
$y^{(1)} = \pi(y^{(0)}),y^{(2)} = \pi(y^{(1)}),\ldots,y^{(j-i)} = z,
\ldots,y^{(\gamma-1)} = \pi(y^{(\gamma-2)})$, are labeled as
zombies. Since none of these bags except $y^{(\gamma-1)}$
%(except maybe $z$ itself; in this case $z
are disappearing zombies, Observation \ref{notdis} implies that
$y^{(1)} = \hat \pi(y^{(0)}) = \hat y^{(1)},
\ldots,y^{(\gamma-1)} = \hat \pi(\hat y^{(\gamma-2)})
= \hat y^{(\gamma-1)}$. Moreover, all these bags have the same
representative $r(y) = r(z) = r(y^{(1)}) = \ldots = r(y^{(\gamma-1)})$.
The bag $v \in \cF_i$ is the initiator of the attachment $\cA(v,y)$.
The $(\gamma-1)$ immediate $\cF$-ancestors $v^{(1)},v^{(2)},\ldots,v^{(\gamma-1)}$
of  $v= v^{(0)}$
are labeled as a result of the attachment $\cA(v,y)$ as incubators.
By Lemma \ref{tat}, none of them is a zombie, and, in particular, none
of them is a disappearing zombie. Thus, again by Observation \ref{notdis},
$v^{(1)} = \hat \pi(v^{(0)}) = \hat v^{(1)},
\ldots,v^{(\gamma-1)} = \hat \pi(\hat v^{(\gamma-2)})
= \hat v^{(\gamma-1)}$.

Denote $x = v^{(j-i)} \in \cF_j$. The representing edge of the attachment
$\cA(v,y)$ is the edge $(r(v),r(y))$. Hence $\delta(r(v),r(y)) \le \tau_i$.
The bag $x \in \cF_j$ is an incubator. Hence it is safe.
On the other hand, the bag $u \in \cF_j$ is risky. Hence $u \ne x$.
Therefore, $Q(u) \cap Q(x) = \emptyset$. (Recall that point sets of two
distinct $j$-level bags are disjoint.)

Denote $k = j+ \kappa + 1$. Let $x' = \hat x^{(k-j)}$ denote the $k$-level
$\hat \cF$-ancestor of the bag $x$. By construction, $Q(x) \subseteq Q(x')$.
Since $\kappa+1\le \eta -1$, the bag $u' = \hat u^{(k-j)}$ is identical
to $u$. (Since the $\eta -1$ immediate $\cF$-ancestors, and $\hat \cF$-ancestors,
are all identical to $u$.)
Both bags $u'$ and $x'$ are non-empty $k$-level bags, and $Q(u')= Q(u), Q(x') \supseteq Q(x)$,
and $Q(u) \cap Q(x) = \emptyset$.
Hence $Q(u') \ne Q(x')$, and thus $Q(u') \cap Q(x') = \emptyset$, and
$u'$ and $x'$ are distinct bags. (See Figure \ref{fig1} for an illustration.)

%\ignore{
\begin{figure*}[htp]
\begin{center}
\begin{minipage}{\textwidth} %{5in}
\begin{center}
\setlength{\epsfxsize}{5in} \epsfbox{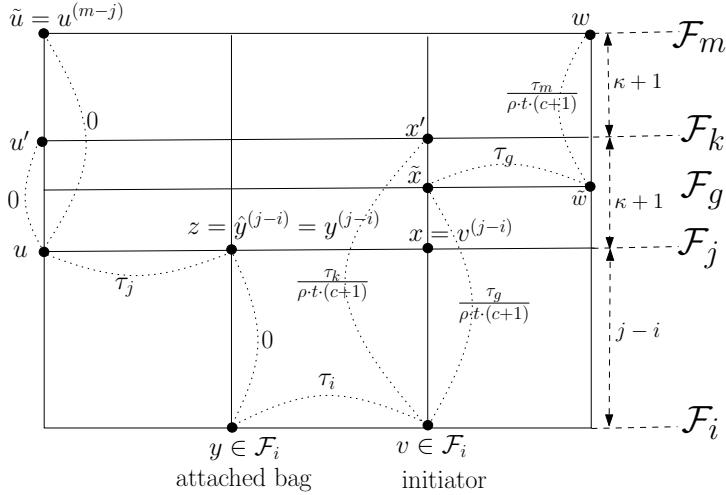}
\end{center}
\end{minipage}
\caption[]{ \label{fig1} \sf A schematic illustration for the proof of Lemma \ref{sp}.
Expressions that appear next to dotted lines
reflect upper bounds on distances between the
representatives of their endpoints.
For example, $r(\tilde u) = r(u)$, and thus 0 appears next
to the dotted line that connects $\tilde u$ and $u$.
Similarly, $\delta(r(v),r(x')) \le \frac{\tau_k}{\rho \cdot t \cdot (c+1)}$,
and thus $\frac{\tau_k}{\rho \cdot t \cdot (c+1)}$ appears next to the
dotted line that connects $v$ and $x'$.}
\end{center}
\end{figure*}

The analysis splits into two cases now, depending on whether the bag $x'  = \hat x^{(k-j)}$
is a zombie or not. (Note that if $k-i \le \gamma-1$ then $x'$ is an incubator, and not a zombie.
But $k-i$ may be larger than $\gamma-1$.)
We start with the case that it is not a zombie.
(It may be an attached bag.) By definition, $x'$ is useful.
We will show that $r(x')$ is prohibitively close to $r(u)$, and this would yield a contradiction.
\\The representative $r(v)$ of  $v$ belongs to $Q(v) \subseteq Q(x) \subseteq Q(x')$.
Hence $r(v),r(x') \in Q(x')$. By Corollary \ref{key2},
$\delta(r(v),r(x')) ~\le~ 2 \cdot \mu_k ~=~ \frac{\tau_k}{\rho \cdot t \cdot (c+1)}.$
Recall that $r(u') = r(u)$ and $r(z) = r(y)$. Hence, by Equation (\ref{mark}), $\delta(r(u'),r(y)) ~=~ \delta(r(u),r(z))
~\le~ \tau_j.$ By the triangle inequality,
\begin{eqnarray*}
\delta(r(u'),r(x')) &\le& \delta(r(u'),r(y)) + \delta(r(y),r(v)) + \delta(r(v),r(x'))
\\ &\le& \tau_j + \tau_i + \frac{\tau_k}{\rho \cdot t \cdot (c+1)}
~=~ \tau_k \cdot \left(\frac{1}{\rho^{k-j}} + \frac{1}{\rho^{k-i}} + \frac{1}{\rho \cdot t \cdot (c+1)}\right).
\end{eqnarray*}
Recall that $k-j = \kappa+1 = \lceil \log_\rho t \rceil + 1$.
Also, $i \le j -1$, and thus $k-i \ge \log_\rho t + 2$.
Since $\rho \ge 2$ and $c \ge 1$, it follows that
\begin{eqnarray*}
\delta(r(u'),r(x')) &\le& \tau_k \cdot \left(\frac{1}{\rho^{\log_\rho t + 1}} + \frac{1}{\rho^{\log_\rho t + 2}} + \frac{1}{\rho \cdot t \cdot (c+1)}\right)
\\ &=& \tau_k \cdot \left(\frac{1}{\rho \cdot t} + \frac{1}{\rho^2 \cdot t} + \frac{1}{\rho \cdot t \cdot (c+1)}\right)
~\le~ \frac{\tau_k}{t} \cdot \left(\frac{1}{2} + \frac{1}{4} + \frac{1}{4}\right) ~=~ \frac{\tau_k}{t}.
\end{eqnarray*}
The bag $x'$ is useful, and $r(u) = r(u')$. Also, $k = j+ (\kappa+1) \in [j,j+(\eta-1)]$,
contradicting Claim \ref{didit}.

%Next, we turn to the case that $x'$ is a zombie. (Note that $x'$ is not an attached bag.)
Next, we turn to the case that $x'$ is a zombie (but not
an attached bag).
%(In other words, the case that $x'$ is labeled as a zombie.)''; if you agree, then
%a similar change should be made for the first case as well.]]
There exists an index $g, j < g < k$,
so that an $\hat \cF$-descendant $\tilde x$ of $x'$
(and an $\hat \cF$-ancestor of $x$ and $v$) became an
attached bag.
Hence there exists an initiator $\tilde w \in \cF_g$, so that the attachment
$\cA(\tilde w, \tilde x)$ occurred during the $g$-level processing.
The representing edge of this attachment is
$(r(\tilde w),r(\tilde x))$. It follows that $\delta(r(\tilde w),r(\tilde x)) \le \tau_g$.
Denote $m = j+2\kappa + 2 = j + (\eta-1)$. Let $w \in \cF_m$ denote the $m$-level
$\cF$-ancestor of $\tilde w$. Observe that $m-g \le m-j = 2\kappa + 2 < \gamma$.
(The constant $c_0$ should be set as $c_0 \ge 3$ for this to hold.)
Hence the bag $w$ is labeled as a result of the attachment $\cA(\tilde w,\tilde x)$
as an incubator. We will show that $r(w)$ is prohibitively close to $r(u)$, yielding a contradiction.
All the $(\gamma-1)$ immediate $\cF$-ancestors of $\tilde w$
are incubators, and thus, by Lemma \ref{tat}, they are not zombies.
In particular, none of them is a disappearing  zombie. Hence, by Observation
\ref{notdis}, for each index $h, g < h \le m$, the $h$-level $\hat \cF$-ancestor
of $\tilde w$ is the same bag as the $h$-level $\cF$-ancestor of $\tilde w$.

Hence the bag $w$ is safe.
On the other hand the $m$-level ancestor $u^{(m-j)} = u^{(\eta-1)}$ of $u$ is risky,
and so $u^{(m-j)} \ne w$. Denote $\tilde u = u^{(m-j)}$.
Since $\tilde x$ is an $g$-level $\hat \cF$-ancestor of $v$, it follows that
$r(v),r(\tilde x) \in Q(\tilde x)$. Hence, by Corollary \ref{key2},
$\delta(r(v),r(\tilde x)) \le \frac{\tau_g}{\rho \cdot t \cdot (c+1)}$.
Similarly, as $w$ is an $\hat \cF$-ancestor of $\tilde w$, and $w \in \cF_m$,
it follows that $\delta(r(w),r(\tilde w)) \le \frac{\tau_m}{\rho \cdot t \cdot (c+1)}$.
Also, $r(\tilde u) = r(u)$, and $r(z) = r(y)$.
Hence $\delta(r(\tilde u),r(y)) ~=~ \delta(r(u),r(z)) ~\le~ \tau_j.$
By the triangle inequality,
\begin{eqnarray*}
\delta(r(\tilde u),r(w)) &\le&
\delta(r(\tilde u),r(y)) + \delta(r(y),r(v)) + \delta(r(v),r(\tilde x))
+ \delta(r(\tilde x),r(\tilde w)) + \delta(r(\tilde w),r(w))
\\ &\le& \tau_j + \tau_i + \frac{\tau_g}{\rho \cdot t \cdot (c+1)}
+ \tau_g + \frac{\tau_m}{\rho \cdot t \cdot (c+1)}
\\ &=& \tau_m \cdot \left(\frac{1}{\rho^{m-j}} + \frac{1}{\rho^{m-i}} + \frac{1}{\rho^{m-g} \cdot \rho \cdot t \cdot (c+1)}
+ \frac{1}{\rho^{m-g}} + \frac{1}{\rho \cdot t \cdot (c+1)} \right)
\\ &\le& \tau_m \cdot \left(\frac{1}{\rho^{2\cdot \log_\rho t + 2}} + \frac{1}{\rho^{2\cdot \log_\rho t + 3}} + \frac{1}{\rho^{\log_\rho t + 2} \cdot \rho \cdot t \cdot (c+1)}
%\\ & & {}
+ \frac{1}{\rho^{\log_\rho t + 2}} + \frac{1}{{\rho \cdot t \cdot (c+1)}}\right)
\\ &=& \tau_m \cdot \left(\frac{1}{\rho^2 \cdot t^2} + \frac{1}{\rho^3 \cdot t^2} + \frac{1}{\rho^3\cdot t^2 \cdot (c+1)} +
\frac{1}{\rho^2 \cdot t} +
\frac{1}{\rho \cdot t \cdot (c+1)}\right) \\ &\le& \frac{\tau_m}{t} \cdot \left(\frac{1}{4} + \frac{1}{8} + \frac{1}{16} + \frac{1}{4} + \frac{1}{4}\right)
~<~ \frac{\tau_m}{t}.
\end{eqnarray*}
Therefore, $\delta(r(\tilde u),r(w)) < \frac{\tau_m}{t}$.
The bag $w$ is useful and $r(\tilde u) = r(u)$.
Also, $\tilde u \in \cF_m$, and $m  = j+ (\eta -1) \in [j,j+(\eta -1)]$.
Hence this is also a contradiction to Claim \ref{didit}.

It follows that at least one of the bags $\hat u^{(0)},\hat u^{(1)},\ldots,\hat u^{(\eta)}$
is safe, and thus $u$ is $\eta$-safe-prospective.
This completes the proof of Lemma \ref{sp}.
%[[S: stopped here]]
\QED
\end{proof}
Next, we combine Corollary \ref{safepros} and Lemma \ref{sp}
to conclude that any active small bag is $(\gamma + \eta)$-prospective.
\begin{lemma} \label{prosper}
Let $j \in [\ell]$.
Any active small bag $v \in \cF_j$ is $(\gamma+\eta)$-prospective.
\end{lemma}
\begin{proof}
If $j > \ell - (\gamma + \eta)$ then the assertion is trivial.
So we henceforth assume that $j \le \ell - (\gamma + \eta)$.
If $v$ is $\eta$-prospective, then we are done. Otherwise, by
Lemma \ref{sp}, it is $\eta$-safe-prospective.
In other words, for some index $i, j \le i \le j + \eta$, the $i$-level
$\hat \cF$-ancestor $\tilde v$ of $v$ is safe. Since $v$ is not $\eta$-prospective,
the bags $v$ and $\tilde v$ are identical, and thus $\tilde v$ is small.
Corollary \ref{safepros} implies that $\tilde v$ is $\gamma$-prospective. It follows that
$v$ is $(\gamma+\eta)$-prospective.\QED
\end{proof}

Recall that the large (respectively, small) counter of a point
$p \in Q$ grows during the $j$-level processing (for some
index $j \in [\ell]$) if $p$ is a representative of some large (resp., small)
$j$-level bag $v$, and if $p$ is not isolated in the $j$-level auxiliary spanner
$\tilde G_j$. (See Section \ref{sec24} for details.)

\begin{observation} \label{ob:small}
Let $j \in [\ell]$, and $v \in \cF_j$ be a small bag.
Then for any point $p \in Q(v)$, it holds that $CTR_j(p) = 0$.
\end{observation}
\begin{proof}
%Consider a small bag $v \in \cF_j$, for some index $j \in [\ell]$.
All $\hat \cF$-descendants of $v$ are small. Also, for a point $p \in Q(v)$,
and an index $i, 1 \le i \le j$, the $i$-level host bag $v_i(p)$ is an
$\hat \cF$-descendant of $v$. Hence any point $p \in Q(v)$ belongs
only to small $i$-level bags, for all $1 \le i \le j$. Hence $CTR_j(p) = 0$.
\QED
\end{proof}
%We summarize this argument in the following observation.
\begin{observation} \label{ker}
Let $j \in [\ell]$, and $v \in \cF_j$ be a large bag. For
every $\hat \cF$-ancestor $v'$ of $v$, $K(v') \supseteq K(v)$.
\end{observation}
\begin{proof}
%Next, consider a large bag $v \in \cF_j$, for some index $j \in [\ell]$.
Only small bags may be labeled as zombies.
Hence $v$ is useful. Observation \ref{notdis} implies that it will not have a step-parent,
i.e., $\hat \pi(v) = \pi(v)$ and $v \in \cS(\pi(v))$.
Also,
we have by construction $K(\pi(v)) \supseteq K(v)$.
Consider now $\pi(v)$,
%Consider the parent $\pi(v) = \hat \pi(v)$ of $v$.
%Equation (\ref{twovertast}) implies that
and notice that $Q(\pi(v)) \supseteq Q(v)$.
Hence $\pi(v)$ will be large as well, and we can apply this argument to $\pi(v)$.
\QED
\end{proof}
% and it will not have a step-parent for the same reason.
%More generally, each of the $\hat \cF$-ancestors $v'$ of $v$ is also an
%$\cF$-ancestor of $v$, and $K(v') \supseteq K(v)$.
%By inductively.
%The following observation follows by a straightforward induction.

Next, we argue that the large counter of any point $p \in Q$ is at most 1.
\\We say that a large bag $v \in \cF_j$ is \emph{atomically large},
for some index $j \in [\ell]$, if all its extended children (if any) are small.
%[[S: instead: if none of its extended children is large]].
In particular, all large 1-level bags are atomically large.
We will use this definition in the proof of the following lemma.
\begin{lemma} \label{large:l}
For each point $p \in Q$, $CTR_\ell(p) \le 1$.
\end{lemma}
\begin{proof}
Suppose for contradiction that there is a point $p \in Q$, with $CTR_\ell(p) \ge 2$,
and let $j,j \in [\ell-1]$, be the index for which $CTR_{j}(p) = 1, CTR_{j+1}(p) = 2$.
By construction, $p$ is the representative of its $(j+1)$-level host bag $v_{j+1}(p)$.
%during the $(j+1)$-level processing.
Moreover, by Observation \ref{ob:small} and by the construction, the $j$-level and $(j+1)$-level host bags $v_j(p)$ and $v_{j+1}(p)$ of $p$,
respectively, are both large.
%Denote by $i$ the smallest index for which the host bag $v_{i}(p)$ of $p$ is large,
%where $1 \le i \le j$.
If $v_j(p)$ is atomically large, set $v = v_j(p)$.
%Also, it is possible that $v = v_j(p)$.
Otherwise, set $v$ to be an arbitrary $\hat \cF$-descendant of $v_j(p)$ that is atomically large.
(Note that $p$ may not belong to $Q(v)$.)
We have $v \in \cF_{g}$, where $1 \le g \le j \le \ell -1$.
By Lemma \ref{kv}, $|K(v)| \ge \ell$.
Write $K(v) = \{q_1,\ldots,q_{k}\}$, where $k \ge \ell$.

Next, we argue that
\begin{equation} \label{ctr}
CTR_{g-1}(q_1) ~=~ CTR_{g-1}(q_2) ~=~ \ldots ~=~ CTR_{g-1}(q_{k}) ~=~ 0.
\end{equation}
Since all counters with index 0 are 0,
Equation (\ref{ctr}) clearly holds if $g =1$.
For $g \ge 2$, all the extended children $z \in \chi(v)$ of $v$ are small by definition.
Also, by construction, $Q(v) = \bigcup_{z \in \chi(v)}Q(z)$.
Therefore, by Observation \ref{ob:small}, we have  $CTR_{g-1}(p) = 0$,  for each point $p \in Q(v)$.
Equation (\ref{ctr}) now follows as $K(v) \subseteq Q(v)$.

%Consider some $j'$-level ancestor $v'$ of $v$, with $j < j' \le \ell$.
Consider the $j-g+1$ immediate $\hat \cF$-ancestors of $v = v^{(0)}$, i.e., $v^{(1)},\ldots,v^{(j-g)} = v_{j}(p),v^{(j-g+1)} = v_{j+1}(p)$.
Observation \ref{ker} implies that for each index $i \in [j-g+1]$,
$K(v^{(i)}) \supseteq K(v) = \{q_1,\ldots,q_{k}\}$.
%$K(v_{j+1}(p)) \supseteq K(v) = \{q_1,\ldots,q_{k}\}$.
For each index $i \in [g,j]$, at most one point from $K(v) = \{q_1,\ldots,q_{k}\}$
is appointed as a representative during the $i$-level processing; that point is the only one from $K(v)$
whose large counter increases during the $i$-level processing.
Since $|K(v)| \ge \ell > j -g + 1$,
there must be at least one point $q \in K(v)$,
with $CTR_{j}(q) = 0$. Also, $q \in K(v) \subseteq K(v_{j+1}(p))$, and the point $p$ is  the representative of $v_{j+1}(p)$.
%during the $(j+1)$-level processing.
Recall that for any large $(j+1)$-level bag $u$, Algorithm $LightSp$ sets its
representative $r(u)$ to be a point $\tilde p \in K(u)$ with the smallest
large counter $CTR_j(\tilde p)$. Hence
% it must hold by construction
%that
$CTR_{j}(p) \le CTR_j(q) = 0$, a contradiction.
%$r(v')$ that is appointed
%as the representative of $v'$ satisfies $CTR_{j'-1}(r(v')) = 0$. Hence, after
%the $j'$-level processing its counter will be incremented to 1.
\QED
\end{proof}

Next, we turn to analyzing single counters of points $p \in Q$.
Recall that for a point $p \in Q$ and an index $j \in [\ell]$, $single\_ctr_j(p)$
counts the number of indices $i \in [j]$ such that the point
$p$ is not isolated in the $i$-level auxiliary spanner $\tilde G_i$ \emph{and}
its host bag $v_i(p)$ is a singleton, i.e., $Q(v_i(p)) = \{p\}$.

%[[S: I changed the argument]]
\begin{lemma} \label{single:l}
For any point $p \in Q$, $single\_ctr_\ell(p) \le \gamma + \eta$.
\end{lemma}
\begin{proof}
For a point $p \in Q$, let $i \in [\ell]$ be the smallest index
such that $p$ is not isolated in the $i$-level auxiliary spanner $\tilde G_i$,
\emph{and} the host bag $v_i(p)$ of $p$ satisfies $Q(v_i(p)) = \{p\}$.
If such an index does not exist, then obviously $single\_ctr_\ell(p) = 0$.
We henceforth assume that the index $i$ exists, and write $v = v_i(p)$.
Notice that $ctr_{1}(p) = \ldots = ctr_{i-1}(p) = 0$, and so
$single\_ctr_{1}(p) = \ldots = single\_ctr_{i-1}(p) = 0$.
%(For $i=0$, we define $ctr_0(q) = single\_ctr_0(q) = 0$, for all points $q \in Q$.)
By definition, the bag $v$ is active.
Thus, Lemma \ref{prosper} implies that $v$ is $(\gamma+\eta)$-prospective.
It follows that there is an index $k, 1 \le k \le (\gamma+\eta)$, such that
the $(i+k)$-level $\hat \cF$-ancestor $\hat v^{(k)}$ of $v$ is a growing bag.
Therefore, the bag $\hat v^{(k)}$ contains at least one point, in addition
to $p$. Moreover,
%Also, for any index $h \ge i+k$, the $h$-level
each $\hat \cF$-ancestor
of $\hat v^{(k)}$ also contains at least one point, in addition
to $p$. Hence $single\_ctr_{i+k}(p) = single\_ctr_{i+k+1}(p) = \ldots = single\_ctr_\ell(p)$.
% for any index $h \ge i+k$.
In other words, the single counter of $p$ may be incremented
only during the $h$-level processing, for $h = i,i+1,\ldots,i+(k-1)$, i.e., for
at most $k \le \gamma + \eta$ times. Therefore $single\_ctr_\ell(p) \le \gamma + \eta$.
\QED
\end{proof}
%We summarize this argument in the following lemma.

%[[S: I need to check this]]
Next, we argue that $plain\_ctr_\ell(p)$ is small as well.
Recall that for a point $p \in Q$ and an index $j \in [\ell]$,
$plain\_ctr_j(p)$ is the number of indices $i \in [j]$,  such that $p$ serves as a
representative of an $i$-level small bag $v$ with $|Q(v)| \ge 2$, \emph{and} $p$ is not isolated
in $\tilde G_i$.
\begin{lemma} \label{leasttwo}
Let $v \in \cF_j$ be a growing small bag, for some index $j \in [2,\ell]$.
Then the kernel set $K(v)$ of $v$ contains at least two points $p,q$
with $plain\_ctr_{j-1}(p) = plain\_ctr_{j-1}(q) = 0$.
\end{lemma}
\begin{proof}
Let $i$ be the minimum index such that there exists an $i$-level growing small
bag $v$. By definition, $i \ge 2$. Also, there are no growing small
$h$-level bags, for any index $1\le h \le i-1$. \\The proof is by induction on $j$.
\\\emph{Basis: $j=i$.}
Consider a growing small bag $v \in \cF_i$. Since it is growing, we have
$|\chi(v)| \ge 2$. Let $u,w \in \chi(v)$ be two distinct
extended children of $v$. Since $v$ is small, both $u$ and $w$
are small too. As there are no growing small bags of level $h, 1 \le h \le i-1$,
it follows that there exist 1-level bags $u'$ and $w'$ such that $u'$
is identical to $u$ (and thus an $\cF$-descendant of it) and $w'$ is identical to $w$ (and an $\cF$-descendant of it).
Moreover, all bags on the path in $\cF$ that connects $u'$ to $u$
(respectively, $w'$ to $w$) are identical to both of them and have the
same representative $r(u)$ (resp., $r(w)$).
%[[S: remove: (Also, since none of these bags is a disappearing zombie, it follows
%that $v$ is an $\hat \cF$-ancestor of both $u$ and $w$,
%and that the paths between $v$ and $u$
%and between $v$ and $w$ in $\cF$ and in $\hat \cF$ are identical.)]] %The same applies to $w$ and $w'$.
%In particular, we conclude that $u,w$ are surviving children of $v$.

If the point set $Q(u)$ of $u$ contains just one single point, i.e.,
$Q(u) = \{r(u)\}$, then $plain\_ctr_{i-1}(r(u)) = 0$. (Its single counter
$single\_ctr_{i-1}(r(u))$ might be larger, but it was taken care of separately. See Lemma \ref{single:l}.)
Otherwise, $|Q(u)| \ge 2$. Hence $Q(u) \setminus \{r(u)\}$ contains at
least one additional point $p(u), p(u) \ne r(u)$. This point satisfies $plain\_ctr_{i-1}(p(u))=0$.
In either case the bag $u$ contains a point $q(u) \in Q(u)$, such that
$plain\_ctr_{i-1}(q(u)) = 0$. The same is true for $w$.
Moreover, $Q(u),Q(w) \subseteq Q(v)$ and $Q(u) \cap Q(w) = \emptyset$, and so
$q(u)$ and $q(w)$ are two distinct points  in $Q(v)$.
%Since $Q(u) \cap Q(w)$
Hence $Q(v)$ contains two distinct points $q(u),q(w)$ such that $plain\_ctr_{i-1}(q(u)) = plain\_ctr_{i-1}(q(w)) = 0$.
By Lemma \ref{kv}, since $v$ is a small bag, $K(v) = Q(v)$, and we are done.
%Since $u$ and $w$ are surviving children of $v$, we have
%$K(v) \subseteq K(u) \cup K(w) \subseteq \{q(u),q(w)\}$, and we are done.
\\\emph{Induction Step: Assume the correctness of the statement for
all smaller values of $j, j \ge i+1$, and prove it for $j$.}
For a growing small bag $v \in \cF_j$, there exist two distinct small
extended children $u,w \in \chi(v) \subseteq \cF_{j-1}$.
Either $u$ is growing, or there exists an extended child $u^{(-1)}$
of $u = u^{(0)}$, which is identical to $u$.
The same argument applies to $u^{(-1)}$.
Hence, there is a sequence of bags $u = u^{(0)},u^{(-1)}, \ldots,u^{(-h)}$,
for some index $h \in [0,j-2]$, with $u^{(-k+1)} = \hat \pi(u^{(-k)})$,
for each $k \in [h]$. The bag $\tilde u = u^{(-h)} \in \cF_{j-h-1}$ is either growing
or belongs to $\cF_1$. Moreover, all bags $u = u^{(0)},u^{(-1)}, \ldots,\tilde u = u^{(-h)}$
are identical.

If the point set $Q(u)$ of $u$ contains just one single point $r(u)$,
then $plain\_ctr_{j-1}(r(u)) = 0$. (Even though its single counter may be larger.)

Otherwise $Q(u) \setminus \{r(u)\}$ contains at least one additional point $p(u),p(u) \ne r(u)$.
If $\tilde u \in \cF_1$, then $plain\_ctr_{j-1}(p(u)) = 0$.
Otherwise $j-h-1  \ge 2$ and $\tilde u$ is a growing bag. By the induction hypothesis,
$\tilde u$ contains at least two points $p_1(u),p_2(u)$ with $plain\_ctr_{j-h-2}(p_1(u))
= plain\_ctr_{j-h-2}(p_2(u)) = 0$. One of these points may become the representative
of $\tilde u$ (and, consequently, of all the $h$ bags $u^{(-h+1)},u^{(-h+2)},\ldots,u^{(0)} = u$
that are identical to $\tilde u$), and, as a result its $(j-1)$-level plain counter
may become positive. However, the other one will have plain counter equal
to 0 on all levels $j-h-1,j-h,\ldots,j-1$. Thus
either $plain\_ctr_{j-1}(p_1(u)) = 0$ or $plain\_ctr_{j-1}(p_2(u)) = 0$ must hold.
Hence in both cases $Q(u) \setminus \{r(u)\}$ contains at least one point $q(u)$
with $plain\_ctr_{j-1}(q(u)) = 0$.

We showed that in all cases $Q(u)$ contains at least one point $q(u)$
with $plain\_ctr_{j-1}(q(u)) = 0$.
Similarly, the bag $w$ also contains a point
$q(w) \in Q(w)$ with $plain\_ctr_{j-1}(q(w)) = 0$. Since $u,w \in \chi(v)$,
it follows that $q(u),q(w) \in Q(v)$. Moreover, $Q(u) \cap Q(w) = \emptyset$, and so
$q(u)$ and $q(w)$ are distinct.   % points in $Q(v)$.
By Lemma \ref{kv}, since $v$ is a small bag,
$K(v) = Q(v)$, which completes the proof.
\QED
\end{proof}
%^^^^^^^^^^^^^^^^^^

Next, we provide an upper bound for plain counters of points in $Q$.
\begin{lemma} \label{toim:l}
For any point $p \in Q$, $plain\_ctr_\ell(q) \le \gamma + \eta$.
\end{lemma}
\begin{proof}
Consider a point $q \in Q$,
and suppose that $plain\_ctr_\ell(q) > 0$. Let $i \in [\ell]$ be the smallest
index such that the plain counter of $q$ is incremented during
the $i$-level processing, i.e., $plain\_ctr_{i-1}(q) = 0, plain\_ctr_i(q) = 1$.
It follows that the $i$-level host bag $v = v_i(p)$ is active and small, and
also $q = r(v)$. Moreover $|Q(v)| \ge 2$. Denote $\beta = \gamma + \eta$.
If $i > \ell - \beta$ then the plain counter of $q$ is incremented at most
$\beta$ times, on levels $i,i+1,\ldots,\ell$. Hence in this case $plain\_ctr_\ell(q) \le \beta$,
as required. Otherwise, let $j$ denote the smallest level of an $\hat \cF$-ancestor
$u$ of $v$ such that $u$ is a growing bag.
By Lemma \ref{prosper}, $j$ is well-defined, with $i < j \le i + \beta \le \ell$.
Consider the $j-i$ immediate $\hat \cF$-ancestors of $v = \hat v^{(0)}$, i.e.,
$\hat v^{(1)}= \hat \pi(\hat v^{(0)}),\ldots,\hat v^{(j-i)} = u = \hat \pi(\hat v^{(j-i-1)})$.
The bags $\hat v^{(1)},\ldots,\hat v^{(j-i-1)}$ are identical to $v$, and
have the same representative $r(v) = q$. If $u$ is a large bag then all
its $\hat \cF$-ancestors are large as well. Also, $p \in Q(u)$, and for all
indices $k \ge j$, $p$ belongs to the point set of the $k$-level $\hat \cF$-ancestor
of $u$. Hence the plain counter of $p$ is not incremented during the $k$-level
processing, for all $k \ge j$.

Suppose now that $u$ is small. Since it is growing, by Lemma \ref{leasttwo},
its kernel set $K(u)$ contains at least two points $p,p'$ with plain counter
zero, i.e., $plain\_ctr_{j-1}(p) = plain\_ctr_{j-1}(p') = 0$.
On the other hand, $plain\_ctr_{j-1}(q) \ge plain\_ctr_i(q) = 1$.
Hence $q$ is not the representative of $u$. More generally, we have the following claim.
\begin{claim} \label{toim}
Let $w = v_k(q)$ be the $k$-level host bag of $q$, for some index $k \ge j$.
If $w$ is small then $q$ is not the representative of $w$.
\end{claim}
\begin{proof}
The proof is by induction on $k$. The basis $k=j$ was already proved.
\\\emph{Induction Step: Assume the correctness of the statement
for all smaller values of $k, k \ge j+1$, and prove it for $k$.}
If $w$ is not growing, then it is identical to an $\hat \cF$-descendant
$w' \in \cF_{k'}$, for some $k' < k$. Hence $r(w) = r(w')$.
By the induction hypothesis, $r(w') \ne q$, and so $r(w) \ne q$ as well.

Otherwise, $w$ is growing. By Lemma \ref{leasttwo}, its kernel set $K(w)$
contains at least two points $p,p'$ with plain counter
zero, i.e., $plain\_ctr_{k-1}(p) = plain\_ctr_{k-1}(p') = 0$.
On the other hand, $plain\_ctr_{k-1}(q) \ge plain\_ctr_i(q) = 1$.
Hence $q$ is not the representative of $w$.
%which completes
%the proof of
Claim \ref{toim} follows.\QED
\end{proof}

We now continue proving Lemma \ref{toim:l}.
\\By Claim \ref{toim}, if $v_k(q)$ is small then the plain counter
of $q$ is not incremented during the $k$-level processing, for all $k \ge j$.
If $v_k(q)$ is large, then obviously, it is not incremented either.
Hence, for any $k \ge j$, the plain counter of $q$ is not incremented
during the $k$-level processing, and so $plain\_ctr_\ell(q) = plain\_ctr_j(q)$.
Thus the plain counter of $q$ may grow only during the $k$-level processing,
for $i \le k < j$.
%Since the counter may grow by at most one on each level,
It follows that $plain\_ctr_\ell(q) \le j-i \le \beta = \gamma + \eta$.
\QED
\end{proof}

Recall that for any $q \in Q$, $load\_ctr_\ell(q) = CTR_\ell(q) + ctr_\ell(q)
= CTR_\ell(q) + single\_ctr_\ell(q) + plain\_ctr_\ell(q).$
Hence, Lemmas \ref{large:l}, \ref{single:l} and \ref{toim:l} imply the following corollary.
\begin{corollary} \label{deg:c}
For any point $q \in Q$, $load\_ctr_\ell(q) \le 2 \cdot (\gamma + \eta) + 1$.
\end{corollary}
Observe that each time that the load counter of a point $q$ is incremented,
its degree in the constructed spanner  grows by at most $O(\Delta(n))$.
(This is because the maximum degree of the $j$-level auxiliary spanner
$\tilde G_j$ is $O(\Delta(n))$, for each $j \in [\ell]$; see Observation \ref{auxdeg}.)
Hence, Corollary \ref{deg:c} implies that the maximum degree of any
$q \in Q$ in the graph $\tilde G_1 \cup \ldots \cup \tilde G_\ell$
is $O(\Delta(n) \cdot (\gamma + \eta)) = O(\Delta(n) \cdot \gamma)$.
%By Observation \ref{auxdeg}, 
The 0-level auxiliary spanner $\tilde G_0$ contributes at most
$O(\Delta(n))$ to the maximum degree of the final spanner $\tilde G$;
also, the path-spanner $H$ has maximum degree $O(\rho)$,
and the base edge set $\cB$ contributes an additive term of $O(1)$
to $\Delta(\tilde G$). (See the beginning of this section.)
We summarize the degree analysis with the next statement.
\begin{lemma} \label{degr2}
$\Delta(\tilde G) = O(\Delta(n) \cdot \gamma + \rho) = O(\Delta(n) \cdot \log_\rho (t/\eps) + \rho)$.
\end{lemma}
%\vspace{0.2in}
{\bf Deriving Theorem \ref{ourresult}.}
%Finally, %\section
Lemmas \ref{edgebound2}, \ref{time2}, \ref{lmstretch2} and \ref{degr2},
and Corollary \ref{wtcorver2}, imply Theorem \ref{ourresult}.

In other words, we devised a transformation that, given a construction of $t$-spanners with $SpSz(n)$
edges, degree $\Delta(n)$ and diameter $\Lambda(n)$ which requires $SpTm(n)$ time, and given parameters $\rho \ge 2$
and $\eps > 0$, provides a construction of $(t+\eps)$-spanners with $O(SpSz(n) \cdot \log_\rho (t/\eps))$ edges,
degree $O(\Delta(n) \cdot \log_\rho (t/\eps) + \rho)$, diameter $O(\Lambda(n) + \log_\rho n + \alpha(\rho))$, and lightness
$O(\frac{SpSz(n)}{n} \cdot \rho \cdot \log_\rho n \cdot (t^3/\eps))$. The latter construction requires
$O(SpTm(n) \cdot \log_\rho (t/\eps) + n \cdot \log n)$ time.
%(Here $\varphi = \varphi(\rho,t,\eps) = \log_\rho (t/\eps)$, and $\alpha$ is the inverse-Ackermann function.)

Substitute into this transformation a construction of $(1+\eps)$-spanners with $O(n)$ edges, degree $O(\rho)$
and diameter $O(\log_\rho n +\alpha(\rho))$, which runs within $O(n \cdot \log n)$ time \cite{ADMSS95,GR082,SE10}.
We obtain a construction of $(1+2\eps)$-spanners with $O(n)$ edges, degree $O(\rho)$, diameter $O(\log_\rho n +\alpha(\rho))$
and lightness $O(\rho \cdot \log_\rho n)$, which requires $O(n \cdot \log n)$ time as well.
(Observe that $t = 1+\eps$, and so $\log_\rho (t/\eps) = \frac{\log(1+\frac{1}{\eps})}{\log \rho} = O(1)$.
Also, we can rescale $2\eps = \eps'$.)
For $\rho = O(1)$ this proves Conjecture 1. Moreover, due to lower bounds of \cite{CG06,DES08}, this result is tight up
to constant factors in the entire range of the parameter $\rho$.
\vspace{0.25in}
\\
{\bf Acknowledgments.}
~~The second-named author is indebted to Michiel Smid, for many helpful and timely comments, and for
his constant support and willingness to help. Also, we wish to thank Adi Gottlieb, for referring us to \cite{GR082},
and for many helpful discussions.

\clearpage
\pagenumbering{roman}
\appendix
\centerline{\LARGE\bf Appendix}

\section{The General Result} \label{appA}
Next we explicate the dependence on $\eps$ and the doubling dimension in our main result.   %Theorem \ref{ourresult}.
%\ref{EucCom2}.
\begin{theorem}  \label{Gen}
For any $n$-point metric $M$ with an arbitrary (not necessarily constant) doubling dimension $dim(M)$, any $\eps>0$ and any parameter $\rho \ge 2$, there exists a
$(1+\eps)$-spanner with $n \cdot \eps^{-O(dim(M))}$ edges, degree $\rho \cdot \eps^{-O(dim(M))}$, diameter $O(\log_\rho n + \alpha(\rho))$ and lightness $(\rho \cdot
\log_\rho n) \cdot \eps^{-O(dim(M))}$.
The running time of this construction is $(n \cdot \log n) \cdot \eps^{-O(dim(M))}$.
\end{theorem}
%[[S: Mention that the general statement takes arbitrary values of $\eps$ and $dim(M)$,
%whereas before we assumed $\eps$ and $dim(M)$ to be constant]]

\section{Proof of Theorem \ref{twoprop}} \label{appB}
This appendix is devoted to the proof of Theorem \ref{twoprop}.
For Euclidean metrics Arya et al.\ \cite{ADMSS95}
proved this theorem for the case $\rho = 2$, and
the authors of the current paper generalized it in
\cite{SE10} to the entire range of the degree parameter $\rho$.
For doubling metrics the proof of this theorem is based on
the works of \cite{GR082} and \cite{SE10}.
We provide it here for the sake of completeness.

Let $M = (P,\delta)$ be an $n$-point doubling metric.
A $(1+\eps)$-spanner $H$ for $M$
is called a \emph{tree-like spanner}, if it contains a tree $T$
that satisfies the following conditions:
\begin{enumerate}
\item Each vertex $v$ of $T$ is assigned a representative point $r(v) \in P$.
\item
There is a 1-1 correspondence between the points of $P$ and the representatives of the leaves of $T$.
%\item
%Each internal vertex in $T$ has at least two children. (Thus there are at most $2n-1$ vertices in $T$.)
\item
Each internal vertex is assigned a unique representative.
%which can be selected
%arbitrarily
%from the points of its descendant leaves.
%(In other words, any two distinct internal vertices
%in $T$ are assigned distinct representatives.
(Thus, each point of $P$
will be the representative of at most two vertices of $T$.)   % a leaf and possibly an internal vertex.)
In particular, there are at most $2n$ vertices in $T$.
\item For any two points $p,q \in P$, there is a $(1+\eps)$-spanner path in $H$ between $p$ and $q$
that is composed of three consecutive parts: (a) a path ascending the edges of $T$,
(b) a single edge, and (c) a path descending the edges of $T$. (Each edge $e = (u,v)$ in $T$
is translated into an edge $(r(u),r(v))$ in $H$.)
%(The weight of every edge in the tree is the metric distance between the corresponding
%representatives.)
\end{enumerate}
We say that such a tree $T$ is a \emph{tree-skeleton} of the spanner $H$.
%Informally, a tree-like spanner is close to a tree.
%in the sense that every
%spanner path is mapped to a path

Gottlieb and Roditty \cite{GR082} proved the following theorem.  % (see also Gottleb et al.\ ArXiv'12).
%[S: This result of \cite{GR082} uses ideas from the works of
(See also \cite{GGN04,CGMZ05,CoG06,Rod07,GR081}
for a number of earlier related works.)
\begin{theorem} [\cite{GR082}] \label{GR}
For any $n$-point doubling metric $M = (P,\delta)$ and any $\eps > 0$, one can
build in $O(n \cdot \log n)$ time a $(1+\eps)$-spanner $H$
and a tree-skeleton $T$ for $H$, such that  both $H$ and $T$ have constant degree.
\end{theorem}

The spanner of Gottlieb and Roditty \cite{GR082} may have a large diameter.
To reduce the diameter, we employ the following tree-shortcutting theorem from \cite{SE10}.
%[[S: I've added reference also to the technical report \cite{SE11}, because there we prove running time bounds;
%perhaps this is redundant?]]
%
%, Theorem 2.9 in \cite{SE11}
\begin{theorem} [Theorem 3 in \cite{SE10}] \label{SE}
Let $T$ be an arbitrary $n$-vertex tree, and denote by $M_T$ the tree metric induced by $T$.
One can build in
%The running time of this construction is
$O(n \log_\rho n)$ time,
for any $\rho \ge 2$, a 1-spanner $G_\rho$ for $M_T$ with $|G_\rho| = O(n)$,
$\Delta(G_\rho) \le \Delta(T) + 2\rho$, and $\Lambda(G_\rho) = O(\log_\rho n + \alpha(\rho))$.
%and $\Psi(G_\rho) = O(\rho \cdot \log_\rho n)$.
%where $\Delta(T)$ denotes the maximum degree of $T$.
\end{theorem}
%Solomon and Elkin \cite{SE10} introduced a shortcutting technique that is summarized

Next, we describe a spanner construction $H^*$ that satisfies all conditions of Theorem \ref{twoprop}.

We start by building the spanner $H$ and its tree-skeleton $T$ that are guaranteed by Theorem \ref{GR}.
Note that $T$ contains at most $2n = O(n)$ vertices. Next,
we build the 1-spanner $G_\rho$ for the tree metric $M_T = (P,\delta_T)$ induced by $T$ that
is guaranteed by Theorem \ref{SE}.
Notice that the edge weights of $G_\rho$
are assigned according to the distance function $\delta_T$ of the tree metric $M_T$.
The 1-spanner $G_\rho$ is converted into a graph $G^*_\rho$
over the point set $P$ in the following way.
Each edge $(u,v)$ of $G_\rho$, for a pair $u,v$ of vertices in $T$,
is translated into the edge $(r(u),r(v))$
between their respective representatives.
%Even though $G_\rho$ is not a spanning graph of $M$
%we can translate it to be such a graph $G^*_\rho$
%in the obvious way. (Each edge in $G_\rho$ is replaced with an edge that connects the representatives
%corresponding to the endpoints of that edge, and the weight is set to be the original metric distance
%between these representatives.)
Finally, let $H^*$ be the spanner obtained from
the union of the graphs $H$ and $G^*_\rho$.

It is easy to see
that the graph $H^*$ satisfies all conditions of Theorem \ref{twoprop}.


\clearpage
\begin{thebibliography}{10}\setlength{\itemsep}{-1ex}\small

\bibitem{ABN11}
I.~Abraham, Y.~Bartal, and O.~Neiman.
\newblock Advances in metric embedding theory.
\newblock {\em Advances in Mathematics}, 228(6):3026–--3126, 2011.

\bibitem{AWY05}
P.~K. Agarwal, Y.~Wang, and P.~Yin.
\newblock Lower bound for sparse {E}uclidean spanners.
\newblock In {\em Proc. of 16th SODA}, pages 670--671, 2005.

\bibitem{ADDJS93}
I.~Alth$\ddot{\mbox{o}}$fer, G.~Das, D.~P. Dobkin, D.~Joseph, and J.~Soares.
\newblock On sparse spanners of weighted graphs.
\newblock {\em Discrete \& Computational Geometry}, 9:81--100, 1993.

\bibitem{ADMSS95}
S.~Arya, G.~Das, D.~M. Mount, J.~S. Salowe, and M.~H.~M. Smid.
\newblock {E}uclidean spanners: short, thin, and lanky.
\newblock In {\em Proc. of 27th STOC}, pages 489--498, 1995.

\bibitem{AMS94}
S.~Arya, D.~M. Mount, and M.~H.~M. Smid.
\newblock Randomized and deterministic algorithms for geometric spanners of
  small diameter.
\newblock In {\em Proc. of 35th FOCS}, pages 703--712, 1994.

\bibitem{AS94}
S.~Arya and M.~H.~M. Smid.
\newblock Efficient construction of a bounded degree spanner with low weight.
\newblock In {\em Proc. of 2nd ESA}, pages 48--59, 1994.

\bibitem{Ass83}
P.~Assouad.
\newblock Plongements lipschitziens dans ${{\mathbb R}}^n$.
\newblock {\em Bull. Soc. Math. France}, 111(4):429–--448, 1983.

\bibitem{BGK12}
Y.~Bartal, L.~Gottlieb, and R.~Krauthgamer.
\newblock The traveling salesman problem: low-dimensionality implies a
  polynomial time approximation scheme.
\newblock In {\em Proc. of 44th STOC}, pages 663--672, 2012.

\bibitem{BCFMS10}
P.~Bose, P.~Carmi, M.~Farshi, A.~Maheshwari, and M.~H.~M. Smid.
\newblock Computing the greedy spanner in near-quadratic time.
\newblock {\em Algorithmica}, 58(3):711--729, 2010.

\bibitem{CG06}
H.~T.-H. Chan and A.~Gupta.
\newblock Small hop-diameter sparse spanners for doubling metrics.
\newblock In {\em Proc. of 17th SODA}, pages 70--78, 2006.

\bibitem{CGMZ05}
H.~T.-H. Chan, A.~Gupta, B.~M. Maggs, and S.~Zhou.
\newblock On hierarchical routing in doubling metrics.
\newblock In {\em Proc. of 16th SODA}, pages 762--771, 2005.

\bibitem{CLN12}
T.-H.~H. Chan, M.~Li, and L.~Ning.
\newblock Incubators vs zombies: Fault-tolerant, short, thin and lanky spanners
  for doubling metrics.
\newblock {\em {T}echnical Report, CoRR}, abs/1207.0892, July, 2012.

\bibitem{CDNS92}
B.~Chandra, G.~Das, G.~Narasimhan, and J.~Soares.
\newblock New sparseness results on graph spanners.
\newblock In {\em Proc. of 8th SOCG}, pages 192--201, 1992.

\bibitem{CDNS95}
B.~Chandra, G.~Das, G.~Narasimhan, and J.~Soares.
\newblock New sparseness results on graph spanners.
\newblock {\em Int. J. Comput. Geometry Appl.}, 5:125--144, 1995.

\bibitem{CDS01}
D.~Z. Chen, G.~Das, and M.~H.~M. Smid.
\newblock Lower bounds for computing geometric spanners and approximate
  shortest paths.
\newblock {\em Discrete Applied Mathematics}, 110(2-3):151--167, 2001.

\bibitem{Chew86}
L.~P. Chew.
\newblock There is a planar graph almost as good as the complete graph.
\newblock In {\em Proc. of 2nd SOCG}, pages 169--177, 1986.

\bibitem{Clark87}
K.~L. Clarkson.
\newblock Approximation algorithms for shortest path motion planning.
\newblock In {\em Proc. of 19th STOC}, pages 56--65, 1987.

\bibitem{Clark99}
K.~L. Clarkson.
\newblock Nearest neighbor queries in metric spaces.
\newblock {\em Discrete Comput. Geom.}, 110(1):63--–93, 1999.

\bibitem{CoG06}
R.~Cole and L.~Gottlieb.
\newblock Searching dynamic point sets in spaces with bounded doubling
  dimension.
\newblock In {\em Proc. of 38th STOC}, pages 574--583, 2006.

\bibitem{CLRS90}
T.~H. Corman, C.~E. Leiserson, R.~L. Rivest, and C.~Stein.
\newblock {\em Introduction to Algorithms, 2nd edition}.
\newblock McGraw-Hill Book Company, Boston, MA, 2001.

\bibitem{DHN93}
G.~Das, P.~J. Heffernan, and G.~Narasimhan.
\newblock Optimally sparse spanners in 3-dimensional {E}uclidean space.
\newblock In {\em Proc. of 9th SOCG}, pages 53--62, 1993.

\bibitem{DN94}
G.~Das and G.~Narasimhan.
\newblock A fast algorithm for constructing sparse {E}uclidean spanners.
\newblock In {\em Proc. of 10th SOCG}, pages 132--139, 1994.

\bibitem{DNS95}
G.~Das, G.~Narasimhan, and J.~S. Salowe.
\newblock A new way to weigh malnourished {E}uclidean graphs.
\newblock In {\em Proc. of 6th SODA}, pages 215--222, 1995.

\bibitem{DES08}
Y.~Dinitz, M.~Elkin, and S.~Solomon.
\newblock Shallow-low-light trees, and tight lower bounds for {E}uclidean
  spanners.
\newblock In {\em Proc. of 49th FOCS}, pages 519--528, 2008.

\bibitem{ES13}
M.~Elkin and S.~Solomon.
\newblock Fast constructions of light-weight spanners for general graphs.
\newblock In {\em Proc. of 24th SODA}, 2013.
\newblock To appear.

\bibitem{ES12}
M.~Elkin and S.~Solomon.
\newblock Optimal {E}uclidean spanners: really short, thin and lanky.
\newblock {\em Technical Report CS12-04, Ben-Gurion University}, April 4, 2012.

\bibitem{GGN04}
J.~Gao, L.~J. Guibas, and A.~Nguyen.
\newblock Deformable spanners and applications.
\newblock In {\em Proc. of 20th SoCG}, pages 190--199, 2004.

\bibitem{GKK12}
L.~Gottlieb, A.~Kontorovich, and R.~Krauthgamer.
\newblock Efficient regression in metric space via approximate lipschitz
  extension.
\newblock Manuscript, 2012.

\bibitem{GR081}
L.~Gottlieb and L.~Roditty.
\newblock Improved algorithms for fully dynamic geometric spanners and
  geometric routing.
\newblock In {\em Proc. of 19th SODA}, pages 591--600, 2008.

\bibitem{GR082}
L.~Gottlieb and L.~Roditty.
\newblock An optimal dynamic spanner for doubling metric spaces.
\newblock In {\em Proc. of 16th ESA}, pages 478--489, 2008.

\bibitem{GLN02}
J.~Gudmundsson, C.~Levcopoulos, and G.~Narasimhan.
\newblock Fast greedy algorithms for constructing sparse geometric spanners.
\newblock {\em SIAM J. Comput.}, 31(5):1479--1500, 2002.

\bibitem{GLNS02}
J.~Gudmundsson, C.~Levcopoulos, G.~Narasimhan, and M.~H.~M. Smid.
\newblock Approximate distance oracles for geometric graphs.
\newblock In {\em Proc. of 13th SODA}, pages 828--837, 2002.

\bibitem{GLNS08}
J.~Gudmundsson, C.~Levcopoulos, G.~Narasimhan, and M.~H.~M. Smid.
\newblock Approximate distance oracles for geometric spanners.
\newblock {\em ACM Transactions on Algorithms}, 4(1), 2008.

\bibitem{GNS05}
J.~Gudmundsson, G.~Narasimhan, and M.~H.~M. Smid.
\newblock Fast pruning of geometric spanners.
\newblock In {\em Proc. of 22nd STACS}, pages 508--520, 2005.

\bibitem{GKL03}
A.~Gupta, R.~Krauthgamer, and J.~R. Lee.
\newblock Bounded geometries, fractals, and low-distortion embeddings.
\newblock In {\em Proc. of 44th FOCS}, pages 534–--543, 2003.

\bibitem{HPM06}
S.~Har-Peled and M.~Mendel.
\newblock Fast construction of nets in low-dimensional metrics and their
  applications.
\newblock {\em SIAM J. Comput.}, 35(5):1148--1184, 2006.

\bibitem{HP00}
Y.~Hassin and D.~Peleg.
\newblock Sparse communication networks and efficient routing in the plane.
\newblock In {\em Proc. of 19th PODC}, pages 41--50, 2000.

\bibitem{Keil88}
J.~M. Keil.
\newblock Approximating the complete {E}uclidean graph.
\newblock In {\em Proc. of 1st SWAT}, pages 208--213, 1988.

\bibitem{KG92}
J.~M. Keil and C.~A. Gutwin.
\newblock Classes of graphs which approximate the complete {E}uclidean graph.
\newblock {\em Discrete \& Computational Geometry}, 7:13--28, 1992.

\bibitem{KL04}
R.~Krauthgamer and J.~R. Lee.
\newblock Navigating nets: Simple algorithms for proximity search.
\newblock In {\em Proc. of 15th SODA}, pages 791--–801, 2004.

\bibitem{LSW94}
H.~P. Lenhof, J.~S. Salowe, and D.~E. Wrege.
\newblock New methods to mix shortest-path and minimum spanning trees.
\newblock manuscript, 1994.

\bibitem{MP00}
Y.~Mansour and D.~Peleg.
\newblock An approximation algorithm for min-cost network design.
\newblock {\em DIMACS Series in Discr. Math and TCS}, 53:97--106, 2000.

\bibitem{NS07}
G.~Narasimhan and M.~Smid.
\newblock {\em Geometric Spanner Networks}.
\newblock Cambridge University Press, 2007.

\bibitem{RS98}
S.~Rao and W.~D. Smith.
\newblock Approximating geometrical graphs via ``spanners'' and ``banyans''.
\newblock In {\em Proc. of 30th STOC}, pages 540--550, 1998.

\bibitem{Rod07}
L.~Roditty.
\newblock Fully dynamic geometric spanners.
\newblock In {\em Proc. of 23rd SoCG}, pages 373--380, 2007.

\bibitem{Sal91}
J.~S. Salowe.
\newblock Construction of multidimensional spanner graphs, with applications to
  minimum spanning trees.
\newblock In {\em Proc. of 7th SoCG}, pages 256--261, 1991.

\bibitem{Salowe92}
J.~S. Salowe.
\newblock On {E}uclidean spanner graphs with small degree.
\newblock In {\em Proc. of 8th SoCG}, pages 186--191, 1992.

\bibitem{Smid09}
M.~H.~M. Smid.
\newblock The weak gap property in metric spaces of bounded doubling dimension.
\newblock In {\em Proc. of Efficient Algorithms}, pages 275--289, 2009.

\bibitem{Sol11}
S.~Solomon.
\newblock An optimal time construction of {E}uclidean sparse spanners with tiny
  diameter.
\newblock In {\em Proc. of 22nd SODA}, pages 820--839, 2011.

\bibitem{Sol12}
S.~Solomon.
\newblock Fault-tolerant spanners for doubling metrics: Better and simpler.
\newblock {\em {T}echnical Report, CoRR}, abs/1207.7040, July, 2012.

\bibitem{SE10}
S.~Solomon and M.~Elkin.
\newblock Balancing degree, diameter and weight in {E}uclidean spanners.
\newblock In {\em Proc. of 18th ESA}, pages 48--59, 2010.

\bibitem{Tal04}
K.~Talwar.
\newblock Bypassing the embedding: algorithms for low dimensional metrics.
\newblock In {\em Proc. of 36th STOC}, pages 281–--290, 2004.

\bibitem{Vai91}
P.~M. Vaidya.
\newblock A sparse graph almost as good as the complete graph on points in $k$
  dimensions.
\newblock {\em Discrete \& Computational Geometry}, 6:369--381, 1991.

\end{thebibliography}
\end {document}